\documentclass[runningheads]{llncs}
\usepackage[T1]{fontenc}
\usepackage{graphicx}
\usepackage{amssymb}
\usepackage{amsmath}
\usepackage{mathrsfs}
\usepackage{xfrac}
\usepackage{tikz}
\usetikzlibrary{calc,arrows,petri,shapes.geometric}
\tikzstyle{place}=[circle,thick,draw=blue!75,fill=blue!20,minimum size=6mm]
\tikzstyle{red place}=[place,draw=red!75,fill=red!20]
\tikzstyle{transition}=[rectangle,thick,draw=black!75,fill=black!20,minimum size=5mm]
\usepackage{subcaption}
\pgfdeclarelayer{bg}
\pgfsetlayers{bg,main}
\usepackage{lmodern}

\newtheorem{prop}{Property}

\usepackage{placeins}

\newcommand{\naturals}{\mathbb{N}}

\DeclareMathOperator{\ran}{ran}

\newcommand{\pre}[1]{{^{\bullet}#1}}
\newcommand{\post}[1]{#1^{\bullet}}
\DeclareMathOperator{\ndeg}{deg}

\newcommand{\seqop}{\rightarrow}
\newcommand{\parop}{\land}
\newcommand{\choiceop}{\times}
\newcommand{\loopop}{\circlearrowleft}

\newcommand{\propdef}{DEF} 
\newcommand{\propmin}{MIN} 
\newcommand{\propone}{$\neg$TRIV} 
\newcommand{\proptwo}{FIN} 
\newcommand{\propinf}{INF} 
\newcommand{\propthree}{COL} 
\newcommand{\propfour}{IND$_L$} 
\newcommand{\propfive}{MON} 
\newcommand{\propsix}{COMP} 
\newcommand{\propseven}{PERM} 
\newcommand{\propeight}{ROB$_{\ell}$} 
\newcommand{\propnine}{$\neg$SUB} 
\newcommand{\propnotsup}{$\neg$SUP} 
\newcommand{\propadd}{ADD} 

\newcommand{\sizename}{\text{size}}
\newcommand{\size}{C_{\sizename}}
\newcommand{\mismatchname}{\text{MM}}
\newcommand{\mismatch}{C_{\mismatchname}}
\newcommand{\connhetname}{\text{CH}}
\newcommand{\connhet}{C_{\connhetname}}
\newcommand{\crossconnname}{\text{CC}}
\newcommand{\crossconn}{C_{\crossconnname}}
\newcommand{\tokensplitname}{\text{ts}}
\newcommand{\tokensplit}{C_{\tokensplitname}}
\newcommand{\separabilityname}{\text{sep}}
\newcommand{\separability}{C_{\separabilityname}}
\newcommand{\controlflowname}{\text{CFC}}
\newcommand{\controlflow}{C_{\controlflowname}}
\newcommand{\avgconnname}{\text{acd}}
\newcommand{\avgconn}{C_{\avgconnname}}
\newcommand{\maxconnname}{\text{mcd}}
\newcommand{\maxconn}{C_{\maxconnname}}
\newcommand{\sequentialityname}{\text{seq}}
\newcommand{\sequentiality}{C_{\sequentialityname}}
\newcommand{\depthname}{\text{depth}}
\newcommand{\depth}{C_{\depthname}}
\newcommand{\diametername}{\text{diam}}
\newcommand{\diameter}{C_{\diametername}}
\newcommand{\cyclicityname}{\text{cyc}}
\newcommand{\cyclicity}{C_{\cyclicityname}}
\newcommand{\netconnname}{\text{CNC}}
\newcommand{\netconn}{C_{\netconnname}}
\newcommand{\densityname}{\text{dens}}
\newcommand{\density}{C_{\densityname}}
\newcommand{\duplicatename}{\text{dup}}
\newcommand{\duplicate}{C_{\duplicatename}}
\newcommand{\emptyseqname}{\emptyset}
\newcommand{\emptyseq}{C_{\emptyseqname}}

\newcommand{\propitemf}[2]{\item \textbf{#1~($#2$)}:}
\newcommand{\propitemc}[3]{\item \textbf{#1~($#2$}~\cite{#3}\textbf{)}:}

\newcommand{\yes}{\textcolor{green!75!black}{\checkmark}}

\newcommand{\no}{\textcolor{red}{\times}}

\newcommand{\netminwf}{W_0}
\newcommand{\opset}{\mathcal{O}}

\def\scalefactor{0.7}

\makeatletter
\newcommand{\shiftleft}{\hspace*{-\@totalleftmargin}}

\begin{document}

\title{Exploring Complexity: An Extended Study of Formal Properties for Process Model Complexity Measures}
\titlerunning{Exploring Complexity}
%
\author{Patrizia Schalk\inst{1} \and 
Adam Burke\inst{2} \and
Robert Lorenz\inst{1}}
\authorrunning{P. Schalk et al.}
%
\institute{University of Augsburg, Universitätsstraße 6a, 86159 Augsburg, Germany\\
\email{\{patrizia.schalk,robert.lorenz\}@uni-a.de} \and
Queensland University of Technology, 2 George St, 4000 Brisbane City, Australia\\
\email{at.burke@qut.edu.au} 
} 
\maketitle              
\begin{abstract}
A good process model is expected not only to reflect the behavior of the process, but also to be as easy to read and understand as possible.
Because preferences vary across different applications, numerous measures provide ways to reflect the complexity of a model with a numeric score.
However, this abundance of different complexity measures makes it difficult to select one for analysis.
Furthermore, most complexity measures are defined for BPMN or EPC, but not for workflow nets.

This paper is an extended analysis of complexity measures and their formal properties. It adapts existing complexity measures to the world of workflow nets. 
It then compares these measures with a set of properties originally defined for software complexity, as well as new extensions to it.
We discuss the importance of the properties in theory by evaluating whether matured complexity measures should fulfill them or whether they are optional. 
We find that not all inspected properties are mandatory, but also demonstrate that the behavior of evolutionary process discovery algorithms is influenced by some of these properties.
Our findings help analysts to choose the right complexity measure for their use-case.

This paper is an extended analysis of process model complexity measures and their formal properties.
\end{abstract}
\section{Introduction}
\label{sec:intro}
Process discovery concerns finding a model for a business process~\cite{Aal16}.
The goal is to automatically construct an understandable model that contains all relevant behavior, so stakeholders and process analysts can make business decisions.
Especially large processes that contain many sub-processes tend to produce cluttered and complex models.
This is one of the reasons why the Directly Follows Miner~\cite{LeePW19} is so popular in practice:
Directly Follows Graphs have only one type of node and easy semantics for the arrows.
However, this model type struggles with distinguishing concurrency from loops.
BPMN, EPCs and workflow nets, on the other hand, feature easy ways to model concurrency but have several types of nodes that each have different semantics.
Consequently, researchers have developed complexity measures for these model types to evaluate which combinations of nodes make models difficult to understand.
These measures tend to count one type of complex structure in a model.
Because there are many types of such structures, numerous complexity measures were proposed~\cite{Men08}.
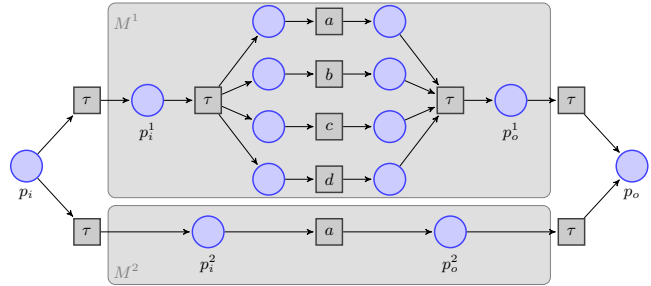
\begin{figure}
\centering
\scalebox{\scalefactor}{
\begin{tikzpicture}[node distance = 1.15cm,>=stealth',bend angle=0,auto]
    \node[place] (start) at (0,0) [label=below:$p_i$] {};
    \node[transition, right of=start,yshift=1.25cm,label=center:$\tau$] (t1) {}
    edge [pre] (start);
    \node[transition, right of=start,yshift=-1.25cm,label=center:$\tau$] (t11) {}
    edge [pre] (start);
    \node [place,right of=t1,label=below:$p_i^1$] (p1) {}
    edge [pre] (t1);
    \node [right of=t11] (p2dummy) {};
    \node [place,right of=p2dummy,label=below:$p_i^2$] (p2) {}
    edge [pre] (t11);
    \node [transition,right of=p1,label=center:$\tau$] (t2) {}
    edge [pre] (p1);
    \node [place,right of=t2,yshift=1.5cm] (p3) {}
    edge [pre] (t2);
    \node [place,right of=t2,yshift=0.5cm] (p4) {}
    edge [pre] (t2);
    \node [place,right of=t2,yshift=-0.5cm] (p5) {}
    edge [pre] (t2);
    \node [place,right of=t2,yshift=-1.5cm] (p6) {}
    edge [pre] (t2);
    \node [transition,right of=p3,label=center:$a$] (a) {}
    edge [pre] (p3);
    \node [transition,right of=p4,label=center:$b$] (b) {}
    edge [pre] (p4);
    \node [transition,right of=p5,label=center:$c$] (c) {}
    edge [pre] (p5);
    \node [transition,right of=p6,label=center:$d$] (d) {}
    edge [pre] (p6);
    \node [place,right of=a] (p7) {}
    edge [pre] (a);
    \node [place,right of=b] (p8) {}
    edge [pre] (b);
    \node [place,right of=c] (p9) {}
    edge [pre] (c);
    \node [place,right of=d] (p10) {}
    edge [pre] (d);
    \node[transition,right of=p8,yshift=-0.5cm,label=center:$\tau$] (t3) {}
    edge [pre] (p7)
    edge [pre] (p8)
    edge [pre] (p9)
    edge [pre] (p10);
    \node [place,right of=t3,label=below:$p_o^1$] (p11) {}
    edge [pre] (t3);
    \node [right of=p2] (dummy) {};
    \node [transition,right of=dummy,label=center:$a$] (e) {}
    edge [pre] (p2);
    \node [right of=e] (dummy) {};
    \node [place,right of=dummy,label=below:$p_o^2$] (p12) {}
    edge [pre] (e);
    \node [right of=p12] (p12dummy) {};
    \node [transition,right of=p11,label=center:$\tau$] (t4) {}
    edge [pre] (p11);
     \node [transition,right of=p12dummy,label=center:$\tau$] (t44) {}
    edge [pre] (p12);
    \node [place,right of=t4,yshift=-1.25cm,label=below:$p_o$] (end) {}
    edge [pre] (t4)
    edge [pre] (t44);
    \begin{pgfonlayer}{bg}
    	\draw[rounded corners,draw=gray,fill=lightgray!50!white] ($(p1) - (0.75, 1.85)$) rectangle ($(p11) + (0.75, 1.85)$);
		\node[gray] at ($(p1) + (-0.4,1.5)$) {$M^{1}$};
		\draw[rounded corners,draw=gray,fill=lightgray!50!white] ($(p2dummy) - (0.75, 1)$) rectangle ($(p12dummy) + (0.75, 0.5)$);
		\node[gray] at ($(p2dummy) - (0.4,0.75)$) {$M^{2}$};
	\end{pgfonlayer}
\end{tikzpicture}}
\caption{A workflow net $M$ with $14$ places and $11$ transitions.}
\label{fig:example}
\end{figure}
Formal properties for such measures help us understand their behavior.
Creators of new measures can and should clearly state which properties the measure satisfies, and designers of discovery algorithms may prove desirable simplicity properties are always maintained.
The same is true for software programs, so as a step towards this ideal setting, Weyuker~\cite{Wey88} introduced nine formal properties that complexity measures for software should fulfill. 
We find six of these properties useful for process models as well, while three are less meaningful for this setting. 
We also draw on work from process mining on properties for complexity~\cite{Car05} and other quality dimensions~\cite{Aal18}, and propose new properties.
Since process discovery algorithms benefit most from formal properties for complexity measures, we perform the analyses on workflow nets, like that in Figure~\ref{fig:example}, which have a strong theory and are therefore often returned by discovery algorithms.
Our analyses hold for BPMN and EPCs as well, since workflow nets can be translated into these modeling languages~\cite{FavFV15}.
The paper is structured as follows:
Section~\ref{sec:related-work} gives an overview of the analyses already performed on complexity measures, before Section~\ref{sec:preliminaries} presents the basic definitions we need for our analyses.
Section~\ref{sec:properties} then defines the properties that we use for the comparative analysis of the complexity measures defined in Section~\ref{sec:analysis}. 
Section~\ref{sec:analysis}, by far the largest section of the paper, performs detailed analyses, including examples.
In Section~\ref{sec:discussion}, we show and discuss the implications for the properties and complexity measures.
Section~\ref{sec:conclusion} concludes the paper.

\newpage
\section{Related Work}
\label{sec:related-work}
There are many factors that make a process model difficult to understand.
Mendling~\cite{Men08} identified 28 such factors and defined complexity measures for EPC models based on them.
He demonstrated their relevance by showing that these complexity measures predict modeling errors in the SAP reference.
Reijers et al.~\cite{ReiM11} found that some of these measures are also tied to the understandability of EPC models by conducting a case study with students of three universities.
According to their study, especially the number of nodes, the density of the network, the average degree of connectors and the cross-connectivity between nodes influence how understandable a process model is.
Lieben et al.~\cite{LieDJJ18} analyzed if existing complexity measures often agree in their score.
They found that they are not completely distinct, but essentially cover four dimensions of complexity, which they call
\textsc{Token Behavior Complexity}, \textsc{Node IO Complexity}, \textsc{Path Complexity} and \textsc{Degree of Connectedness}.
They conclude that it is not necessary to use every complexity measure to adequately evaluate the complexity of a model.
Instead, it is sufficient to use one complexity measure of each dimension, as well as two measures that do not fit to any dimension.

Until now, little is known about formal properties of complexity measures. 
For the quality dimensions fitness, precision and generalization, Syring et al.~\cite{SyrTA19} use desirable properties defined by van der Aalst~\cite{Aal18} to compare measures.
Since they focus on the behavior and abstract from the representation of the model, they did not include properties for complexity measures.
Weyuker~\cite{Wey88}, on the other hand, focuses on properties for software complexity measures.
Because software and process models use similar control structures, these properties are a good starting point to analyze complexity measures for process models.
Cardoso\cite{Car05} therefore analyzed the control flow complexity with Weyuker's properties.

Whether a complexity measure fulfills formal properties is important for algorithms that find process models by optimizing over quality criteria.
An example of such a discovery algorithm is the Evolutionary Tree Miner~(ETM)~\cite{BuiDA14}.
This evolutionary algorithm continuously mutates a randomly generated set of process models until a model with good fitness, precision, generalization, and simplicity is found.
The importance of each dimension can be set by weights, and the concrete measures can be chosen freely.
In turn, the results of the ETM depend on the chosen quality measures and their properties.
Because properties of the quality dimensions fitness, precision, and generalization are already analyzed, we focus on evaluating whether Weyuker's properties~\cite{Wey88} are useful for complexity measures of process models.

\newpage
\section{Preliminaries}
\label{sec:preliminaries}
Let $\mathbb{N} := \{1,2, \dots\}$ and $\mathbb{N}_0 := \mathbb{N} \cup \{0\}$. 
Further, let $\mathbb{R}$ be the set of real numbers and $\mathbb{R}^+_0 := \{r \in \mathbb{R} \mid r \geq 0\}$.
To define workflow nets, we first need a definition of Petri nets.

\begin{definition}[Petri net]
\label{def:petri-net}
A (simple) \emph{Petri net} is a 3-tuple $N = (P, T, F)$, where $P$ is a finite set of \emph{places}, $T$ is a finite set of \emph{transitions}, $P \cap T = \emptyset$ and $F$ is a \emph{flow relation} with $F \subseteq (P \times T) \cup (T \times P)$.
For any place $p \in P$ of $N$, we call $\pre{p} = \{t \in T \mid (t, p) \in F\}$ the \emph{preset of $p$} and $\post{p} = \{t \in T \mid (p, t) \in F\}$ the \emph{postset of $p$}. 
We define the pre- and postset of transitions $t \in T$ accordingly.
\end{definition}

Workflow nets are frequently used for the modelling of business processes.
In this paper, we focus on \emph{labeled workflow nets}---a special subclass of Petri nets---for our analyses of complexity measures.

\begin{definition}[Labeled workflow nets]
\label{def:workflow-net}
A labeled \emph{workflow net} is a 6-tuple $W = (P, T, F, \ell, p_i, p_o)$ where $(P, T, F)$ is a Petri net, $p_i, p_o \in P$ with $p_i \neq p_o$, 
\begin{itemize}
\item $p_i$ is the only place in $W$ for which $|\pre{p_i}| = 0$,
\item $p_o$ is the only place in $W$ for which $|\post{p_o}| = 0$,
\item every node in $P \cup T$ lies on some path from $p_i$ to $p_o$
\end{itemize}
and, for a finite alphabet $A$ with $\tau \not\in A$, $\ell : T \rightarrow A \cup \{\tau\}$ is a labeling function, assigning a transition label to each transition.
We call $t \in T$ with $\ell(t) = \tau$ a \emph{silent transition}.
\end{definition}

With $\mathcal{M}$, we denote the set of all possible workflow nets.
Figure~\ref{fig:example} and Figure~\ref{fig:netminwf} illustrate workflow nets with their places drawn as circles, and transitions drawn as squares.
When the labeling of transitions doesn't matter, we leave these squares empty. 
Otherwise, we write the label of the transition in its center.
\begin{figure}
\centering
\scalebox{\scalefactor}{
\begin{tikzpicture}[node distance = 1.5cm,>=stealth',bend angle=0,auto]
    \node[place] (start) at (0,0) [label=below:$p_i$] {};
    \node[transition, right of=start] (t1) [label=below:$t$] {}
    edge [pre] (start);
    \node[place, right of=t1] (end) [label=below:$p_o$] {}
    edge [pre] (t1);
    
    \node at (-1.5,0) {$\netminwf$:};
\end{tikzpicture}}
\caption{The smallest possible workflow net, $\netminwf$.}
\label{fig:netminwf}
\end{figure}

Many complexity measures take connectors for choices or parallelism into account.
Therefore, to make our analyses also valid for BPMN and EPC models, we define which structures in a workflow net we understand as a connector.

\begin{definition}[Connectors in workflow nets]
\label{def:connectors}
Take an arbitrary workflow net $W = (P, T, F, \ell, p_i, p_o)$, a transition $t \in T$ and a place $p \in P$.
\begin{itemize}
\item If $|\post{p}| > 1$, we call $p$ an \emph{xor-split}.
\item If $|\pre{p}| > 1$, we call $p$ an \emph{xor-join}.
\item If $|\post{t}| > 1$, we call $t$ an \emph{and-split}.
\item If $|\pre{t}| > 1$, we call $t$ an \emph{and-join}.
\end{itemize}
Accordingly, we define 
\begin{itemize}
\item The set of xor-splits in $W$ as $\mathcal{S}_{\text{xor}}^W := \{p \in P \mid |\post{p}| > 1\}$,
\item the set of xor-joins in $W$ as $\mathcal{J}_{\text{xor}}^W := \{p \in P \mid |\pre{p}| > 1\}$,
\item the set of and-splits in $W$ as $\mathcal{S}_{\text{and}}^W := \{t \in T \mid |\post{t}| > 1\}$,
\item the set of and-joins in $W$ as $\mathcal{J}_{\text{and}}^W := \{t \in T \mid |\pre{t}| > 1\}$.
\end{itemize}
Note that these sets are not necessarily disjoint.
The set of xor-connectors in $W$ is $\mathcal{C}_{\text{xor}}^W := \mathcal{S}_{\text{xor}}^W \cup \mathcal{J}_{\text{xor}}^W$, the set of and-connectors in $W$ is $\mathcal{C}_{\text{and}}^W := \mathcal{S}_{\text{and}}^W \cup \mathcal{J}_{\text{and}}^W$ and the set of all connectors is $\mathcal{C}^W := \mathcal{C}_{\text{xor}}^W \cup \mathcal{C}_{\text{and}}^W$.
\end{definition}

In Definition~\ref{def:connectors}, we don't define or-connectors for workflow nets.
This is because there are multiple ways to model an or-connector in this model type.
However, all of these ways use parallel splits and exclusive choices, so if a complexity measure punishes a connector representing an inclusive choice, it can punish those connectors instead.
The complexity of a model may not only depend on the connector types but also on the labeling function, leading to the next definition.

\begin{definition}[Relabeling of a Workflow net]
Let $A$ and $B$ be two finite alphabets. 
Let $W \in \mathcal{M}$ be a labeled workflow net $W = (P, T, F, \ell, p_i, p_o)$ with $\ell : T \rightarrow A \cup \{\tau\}$. 
We call a labeling $\ell' : T \rightarrow B \cup \{\tau\}$ a \emph{uniform relabeling} of $\ell$ if for all transitions $t_1, t_2 \in T: \ell(t_1) = \ell(t_2) \leftrightarrow \ell'(t_1) = \ell'(t_2)$ and for all transitions $t \in T: \ell(t) = \tau \leftrightarrow \ell'(t) = \tau$.
We denote the set of all uniform relabelings of $\ell$ by $\mathcal{R}_{\ell}$.
For any $\ell' \in \mathcal{R}_{\ell}$, we call $W_{\ell'} = (P, T, F, \ell', p_i, p_o)$ a \emph{relabeling} of $W$.
\end{definition}

For one of Weyuker's properties, we want to check if changing the order of activities impacts a complexity measure.
Cardoso~\cite{Car05} implements this by allowing activities to change labels, and connectors to change type.
We exploit the expressiveness of workflow nets and define a permutation as a net with the same transitions, but different control flow.
Thus, we get closer to what Weyuker defined as permutations for software programs.

\begin{definition}[Permutations of Workflow Nets]
Let $W$ be a workflow net with $W = (P, T, F, \ell, p_i, p_o)$, $P'$ be a set of places, $F' \subseteq (P' \times T) \cup (T \times P')$ be an arbitrary flow relation and $p_i', p_o' \in P'$ two places in $P'$. 
Using these notions, we call $Perm(W) := \{W' = (P', T, F', \ell, p_i', p_o') \mid W' \text{ is a workflow net}\}$ the set of \emph{permutations} of $W$.
\end{definition}

We also want to investigate what happens if we combine two or more models.
For the operations, we take inspiration from the creation of block-structured workflow nets~\cite{LeeFA13} and use these operations on arbitrary nets.
The operations defined in Definition~\ref{def:operations} include sequential composition~($\seqop$), parallel composition~($\parop$), exclusive choice~($\choiceop$) and iteration~($\loopop$).

We have two reasons to focus on these operations: 
First, they are very well known and established due to the inductive miner~\cite{LeeFA13},\cite[p.222]{Aal16} relying on them, so many process analysts are already familiar with them.
Second, the ETM heavily relies on these operations, which gives us the opportunity to directly test the impact of certain properties on a well-established discovery algorithm.
The ETM does not operate directly on workflow nets, but on \emph{process trees}~\cite{Aal16}, which can easily be converted to workflow nets with the operations of Definition~\ref{def:operations}.

\begin{definition}[Operations on workflow nets]
\label{def:operations}
Let $n \in \mathbb{N}$ with $n \geq 2$. For $j \in \{1, \dots, n\}$, let $M^j = (P^j, T^j, F^j, p_i^j, p_o^j)$ be workflow nets with disjoint sets of places and transitions ($P^1 \cap \dots \cap P^n = \emptyset$ and $T^1 \cap \dots \cap T^n = \emptyset$).
We define 
\begin{itemize}
\item The \emph{sequential composition} $\seqop(M^1, \dots, M^n) := (P_{\seqop}, T_{\seqop}, F_{\seqop}, p_i^1, p_i^n)$, with
\begin{itemize}
\item $P_{\seqop} := \bigcup_{j=1}^n P^j$, 
\item $T_{\seqop} := \bigcup_{j=1}^n T^j \cup \bigcup_{j=1}^{n-1}\{t_j^*\}$ for new transitions $t_1^*, \dots, t_{n-1}^*$ and 
\item $F_{\seqop} := \bigcup_{j=1}^n F_j \cup \bigcup_{j=1}^{n-1} \{(p_o^j, t_j^*), (t_j^*, p_i^{j+1})\}$.
\end{itemize}
\item The \emph{parallel composition} $\parop\hspace*{1mm}(M_1, \dots, M_n) := (P_{\parop}, T_{\parop}, F_{\parop}, p_i^*, p_o^*)$, where $p_i^*$ and $p_o^*$ are new places and
\begin{itemize}
\item $P_{\parop} := \bigcup_{j=1}^n P^j \cup \{p_i^*, p_o^*\}$,
\item $T_{\parop} := \bigcup_{j=1}^n T^j \cup \{t_i^*, t_o^*\}$ for new transitions $t_i^*$ and $t_o^*$ and
\item $F_{\parop} := \bigcup_{j=1}^n F^j \cup \{(p_i^*, t_i^*), (t_o^*, p_o^*)\} \cup  \bigcup_{j=1}^n \{(t_i^*, p_i^j), (p_o^j, t_o^*)\}$.
\end{itemize}
\item The \emph{choice} $\choiceop\hspace*{1mm}(M_1, \dots, M_n) := (P_{\choiceop}, T_{\choiceop}, F_{\choiceop}, p_i^*, p_o^*)$, where $p_i^*$ and $p_o^*$ are new places and
\begin{itemize}
\item $P_{\choiceop} := \bigcup_{j=1}^n P^j \cup \{p_i^*, p_o^*\}$,
\item $T_{\choiceop} := \bigcup_{j=1}^n (T^j \cup \{t_j^*, s_j^*\})$ and
\item $F_{\choiceop} := \bigcup_{j=1}^n (F^j \cup \{(p_i^*, t_j^*), (t_j^*, p_i^j), (p_o^j, s_j^*), (s_j^*, p_o^*)\}$.
\end{itemize}
\item The \emph{iteration} $\loopop\hspace*{1mm}(M_1, \dots, M_n) := (P_{\loopop}, T_{\loopop}, F_{\loopop},, p_i^*, p_o^*)$, where $p_i^*$ and $p_o^*$ are new places and
\begin{itemize}
\item $P_{\loopop} := \bigcup_{j=1}^n P^j \cup \{p_i^*, p_o^*, p^*, q^*\}$,
\item $T_{\loopop} := \bigcup_{j=1}^n (T^j \cup \{t_j^*, s_j^*\}) \cup \{t^*, s^*\}$ and
\item $F_{\loopop} := \{(p_i^*, t^*), (t^*, p^*), (q^*, s^*), (s^*, p_o^*)\} \cup \bigcup_{j=1}^n F_j$ \\
\phantom{$F_{\loopop} := $} $\cup \hspace*{1mm} \bigcup_{j = 1}^{n-1}\{(p^*, t_j^*), (t_j^*, p_i^j), (p_o^j, s_j^*), (s_j^*, q^*)\}$ \\
\phantom{$F_{\loopop} := $} $\cup \hspace*{1mm} \{(q^*, s_n^*), (s_n^*, p_i^n), (p_o^n, t_n^*), (t_n^*, p^*)\}$.
\end{itemize}
\end{itemize}
\end{definition}
Figure~\ref{fig:operations} shows how these operations combine workflow nets.

\tikzset{pics/net/.style args={#1}{
code={
\begin{scope}
	\node [place] (_pi) [label=below:$p_i^{#1}$] {};
	\node (_dots) [right of=_pi] {$\dots$}
	edge [pre] (_pi);
	\node [place, right of=_dots] (_po) [label=below:$p_o^{#1}$] {}
	edge [pre] (_dots);
	\begin{pgfonlayer}{bg}
		\draw[rounded corners,draw=gray,fill=lightgray!50!white] ($(_pi) - (0.75, 0.95)$) rectangle ($(_po) + (0.75, 0.75)$);
		\node[gray] at ($(_pi) + (-0.4,0.5)$) {$M^{#1}$};
	\end{pgfonlayer}
\end{scope}
}}}

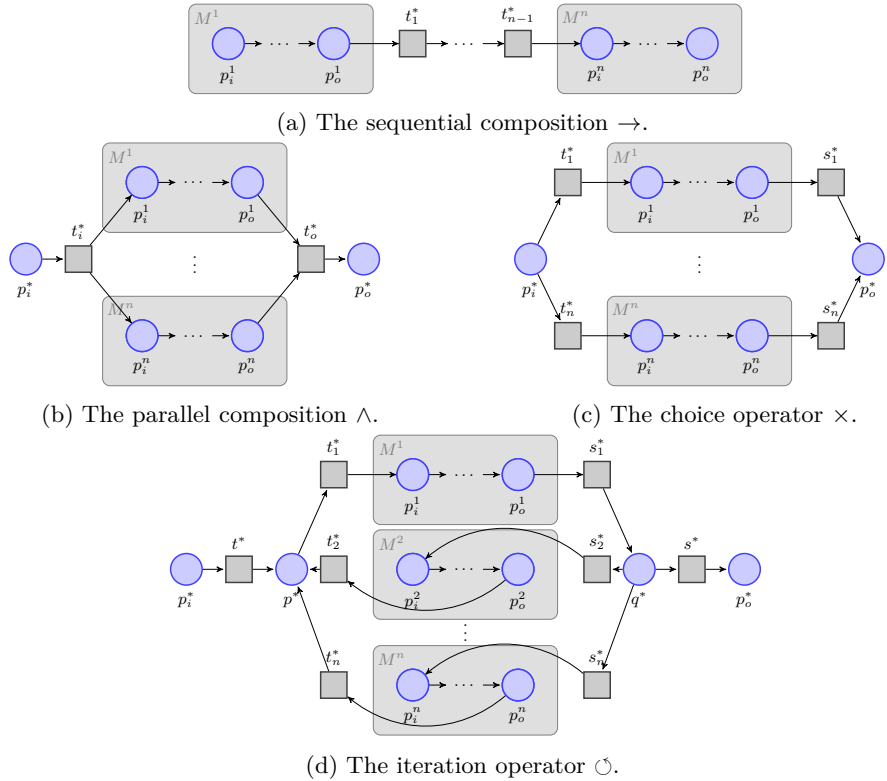
\begin{figure}[ht]
\begin{center}
\begin{subfigure}{\textwidth}
\centering
\scalebox{\scalefactor}{
\begin{tikzpicture}[node distance = 1cm,>=stealth',bend angle=0,auto]
	\node[place] (pi1) at (0,0) {};
	\node[right of=pi1] (dummy) {};
	\node[place,right of=dummy] (po1) {};
	\node[transition, right of=po1,xshift=0.5cm] (t1) [label=above:$t_1^*$] {};
	\node[right of=t1] (dots) {$\dots$}; 
	\node[transition, right of=dots] (tn1) [label=above:$t_{n-1}^*$] {};
	\node[place,right of=tn1,xshift=0.5cm] (pin) {};
	\path (pi1) pic {net={1}};
	\path (pin) pic {net={n}};
	\draw[->] (po1) to (t1);
	\draw[->] (t1) to (dots);
	\draw[->] (dots) to (tn1);
	\draw[->] (tn1) to (pin);
\end{tikzpicture}}
\caption{The sequential composition $\seqop$.}
\label{fig:seq-comp}
\end{subfigure}
\begin{subfigure}{0.45\textwidth}
\scalebox{\scalefactor}{
\begin{tikzpicture}[node distance = 1cm,>=stealth',bend angle=0,auto]
	\node[place] (pi) at (0,0) [label=below:$p_i^*$] {};
	\node[transition, right of=pi] (ti) [label=above:$t_i^*$] {}
	edge [pre] (pi);
	\node[place, above right of=ti, xshift=0.5cm, yshift=0.75cm] (pi1) {}; 
	\node[place, below right of=ti, xshift=0.5cm, yshift=-0.75cm] (pin) {}; 
	\node (d1) [right of=pi1] {};
	\node (dn) [right of=pin] {};
	\node[place, right of=d1] (po1) {};
	\node[place, right of=dn] (pon) {};
	\node[transition, below right of=po1, xshift=0.5cm, yshift=-0.75cm] (to) [label=above:$t_o^*$] {};
	\node at ($0.5*(ti) + 0.5*(to)$) {$\vdots$};
	\node[place, right of=to] (po) [label=below:$p_o^*$] {}
	edge [pre] (to); 
	\path (pi1) pic {net={1}};
	\path (pin) pic {net={n}};
	\draw[->] (ti) to (pi1);
	\draw[->] (ti) to (pin);
	\draw[->] (po1) to (to);
	\draw[->] (pon) to (to);
\end{tikzpicture}}
\caption{The parallel composition $\parop$.}
\label{fig:par-comp}
\end{subfigure}
\hfill
\begin{subfigure}{0.45\textwidth}
\scalebox{\scalefactor}{
\begin{tikzpicture}[node distance = 1cm,>=stealth',bend angle=0,auto]
	\node[place] (pi) at (0,0) [label=below:$p_i^*$] {};
	\node[transition, above right of=pi, yshift=0.75cm] (t1) [label=above:$t_1^*$] {}
	edge [pre] (pi);
	\node[transition, below right of=pi, yshift=-0.75cm] (tn) [label=above:$t_n^*$] {}
	edge [pre] (pi);
	\node[place, right of=t1, xshift=0.5cm] (pi1) {}; 
	\node[place, right of=tn, xshift=0.5cm] (pin) {}; 
	\node (d1) [right of=pi1] {};
	\node (dn) [right of=pin] {};
	\node[place, right of=d1] (po1) {};
	\node[place, right of=dn] (pon) {};
	\node[transition, right of=po1, xshift=0.5cm] (s1) [label=above:$s_1^*$] {};
	\node[transition, right of=pon, xshift=0.5cm] (sn) [label=above:$s_n^*$] {};
	\node at ($0.25*(t1) + 0.25*(tn) + 0.25*(s1) + 0.25*(sn)$) {$\vdots$};
	\node[place, below right of=s1, yshift=-0.75cm] (po) [label=below:$p_o^*$] {}
	edge [pre] (s1)
	edge [pre] (sn); 
	\path (pi1) pic {net={1}};
	\path (pin) pic {net={n}};
	\draw[->] (t1) to (pi1);
	\draw[->] (tn) to (pin);
	\draw[->] (po1) to (s1);
	\draw[->] (pon) to (sn);
\end{tikzpicture}}
\caption{The choice operator $\choiceop$.}
\label{fig:choice}
\end{subfigure}
\begin{subfigure}{\textwidth}
\centering
\scalebox{\scalefactor}{
\begin{tikzpicture}[node distance = 1cm,>=stealth',bend angle=0,auto]
	\node[place] (start) at (0,0) [label=below:$p_i^*$] {};
	\node[transition, right of=start] (t) [label=above:$t^*$] {}
	edge [pre] (start);
	\node[place, right of=t] (pi) [label=below:$p^*$] {}
	edge [pre] (t);
	\node[transition, above of=pi] (t1) at ($(pi) + (0.8,0.8)$) [label=above:$t_1^*$] {}
	edge [pre] (pi);
	\node[transition] (t2) at ($(pi) + (0.8,0)$) [label=above:$t_2^*$] {}
	edge [post] (pi);
	\node[transition, below of=pi] (tn) at ($(pi) + (0.8,-1.2)$) [label=above:$t_n^*$] {}
	edge [post] (pi);
	\node[place, right of=t1, xshift=0.5cm] (pi1) {}; 
	\node[place, right of=t2, xshift=0.5cm] (pi2) {}; 
	\node[place, right of=tn, xshift=0.5cm] (pin) {}; 
	\node (d1) [right of=pi1] {};
	\node (d2) [right of=pi2] {};
	\node (dn) [right of=pin] {};
	\node[place, right of=d1] (po1) {};
	\node[place, right of=d2] (po2) {};
	\node[place, right of=dn] (pon) {};
	\node[transition, right of=po1, xshift=0.5cm] (s1) [label=above:$s_1^*$] {};
	\node[transition, right of=po2, xshift=0.5cm] (s2) [label=above:$s_2^*$] {};
	\node[transition, right of=pon, xshift=0.5cm] (sn) [label=above:$s_n^*$] {};
	\node at ($0.25*(t2) + 0.25*(tn) + 0.25*(s2) + 0.25*(sn)$) {$\vdots$};
	\node[place] (po) at ($(s1) + (0.8,-1.8)$) [label=below:$q^*$] {}
	edge [pre] (s1)
	edge [post] (s2)
	edge [post] (sn); 
	\node[transition, right of=po] (s) [label=above:$s^*$] {}
	edge [pre] (po);
	\node[place, right of=s] (end) [label=below:$p_o^*$] {}
	edge [pre] (s);
	\path (pi1) pic {net={1}};
	\path (pi2) pic {net={2}};
	\path (pin) pic {net={n}};
	\draw[->] (t1) to (pi1);
	\draw[->] (po1) to (s1);
	\draw[->, bend right=40] (s2) to (pi2);
	\draw[->, bend left=40] (po2) to (t2);
	\draw[->, bend right=40] (sn) to (pin);
	\draw[->, bend left=40] (pon) to (tn);
\end{tikzpicture}}
\caption{The iteration operator $\loopop$.}
\label{fig:iteration}
\end{subfigure}
\end{center}
\caption{A schematic overview of the four operations of Definition~\ref{def:operations}.}
\label{fig:operations}
\end{figure}

Finally, for this section, we define complexity measures.
Existing complexity measures all are functions that take a model as input and return a \emph{complexity score}---a real value that reflects how complex the net is.

\begin{definition}[Complexity Measure]
\label{def:complexity}
Let $\mathcal{M}$ be the set of all possible workflow nets. 
A \emph{complexity measure} $C$ is a function $C : \mathcal{M} \rightarrow \mathbb{R}$.
\end{definition}

Complexity influences many properties of a process model, like its understandability, correctness, or the time needed to execute certain algorithms.

\section{Properties of Complexity Measures}
\label{sec:properties}
In the first part of this section, we define Weyuker's properties for workflow nets.
Since this set of properties is not intended to be complete, we propose simple extensions that further deepen the understanding of the analyzed complexity measures in a second part.
In both parts, we take $C$ as a placeholder for any complexity measure.

\subsection{Properties of Weyuker}
The first property of Weyuker states that a complexity measure should be able to return more than a single score. 
This property is important independent of the use-case, since a complexity measure that assigns the same score to all input is not useful for any purpose regarding complexity.

\begin{prop}[{\propone} (W1)] 
$C$ is \textbf{not trivial}, i.e. it doesn't assign the same complexity score to all process models: 
\[\exists M_1, M_2 \in \mathcal{M}: C(M_1) \neq C(M_2)\]
\end{prop}

A non-trivial measure can still be of little use if the amount of assigned scores is low.
Property two therefore states only finitely many inputs should receive the same complexity score.

\begin{prop}[{\proptwo} (W2)]
There are only \textbf{finitely many} process models that receive the same complexity score from $C$:
\[\forall c \in \mathbb{R}: |\{M \in \mathcal{M} \mid C(M) = c\}| < \infty\]
\end{prop}

If a complexity measure assigns a unique score to each workflow net, {\proptwo} is fulfilled, but the measure is too specific:
It would find a numerical representation of the input instead of rating its complexity.
This is reflected in the third property.

\begin{prop}[{\propthree} (W3)]
$C$ allows for \textbf{collisions}, i.e. different process models can get the same complexity score:
\[\exists M_1, M_2 \in \mathcal{M}: M_1 \neq M_2 \land C(M_1) = C(M_2)\]
\end{prop}

Weyuker states that not the output of a program should influence its complexity, but the details of the implementation~\cite{Wey88}.
Accordingly, the modeled process should not define the complexity of the model.
Put differently, if we have two models with the same language, it should be possible for these models to get different complexity scores. 
For a workflow net $M \in \mathcal{M}$, we define its language $L(M)$ as done in~\cite{Aal16}.

\begin{prop}[{\propfour} (W4)]
$C$ is \textbf{independent} of the input model's language:
\[\exists M_1, M_2 \in \mathcal{M}: L(M_1) = L(M_2) \land C(M_1) \neq C(M_2)\]
\end{prop}

It seems sensible that a small part of a process model should not be more complex than the entire model.
For example, the complexity of the model $M^1$ of Figure~\ref{fig:example} should not be higher than the complexity of the entire model, $M$.
Weyuker therefore proposes that complexity measures should be monotone in this sense.
Van der Aalst~\cite{Aal18} agrees with this, as his properties \textbf{RecPro3}, \textbf{PrecPro1} and \textbf{GenPro3} are closely related to \propfive.

\begin{prop}[{\propfive} (W5)]
$C$ is \textbf{monotone}, i.e. the complexity score of a composed model cannot be lower than the complexity score of one of its parts: Let $\opset := \{\seqop, \parop, \choiceop, \loopop\}$.
For $M_1, \dots, M_n \in \mathcal{M}$ and $M_i \in \{M_1, \dots, M_n\}$ we have:
\[\forall \oplus \in \opset: C(\oplus (M_1, \dots, M_n)) \geq C(M_i)\]
\end{prop}

For the next property, suppose we have two models with the same complexity score.
Weyuker argues that composing each of them with a third one does not necessarily lead to the same complexity score, since one part of a model can influence the complexity of another part.
We can observe this with the nets in Figure~\ref{fig:example}:
The fact that the activity $a$ of $M^2$ is also present in $M^1$ increases the perceived complexity of $M$. 
If $M^1$ would not contain the activity label $a$, the net would be simpler.
It is desirable that a complexity measure finds such dependencies.

\begin{prop}[{\propsix} (W6)]
$C$ is \textbf{composition sensitive}, i.e. there are process models of equal complexity according to $C$, whose complexity differ when composed with a third process model. Let $\opset := \{\seqop, \parop, \choiceop, \loopop\}$.
\begin{align*}
\forall \oplus \in \opset: \exists M_1, M_2, M_3 \in \mathcal{M}: C(M_1) &= C(M_2) 
\land C(M_1 \oplus M_3) \neq C(M_2 \oplus M_3)
\end{align*}
\end{prop}

One way of transforming a program is to reorder its instructions.
Weyuker argues that such a reordering should impact the complexity.
This is also true for process models:
If we add more arcs to $M$ in Figure~\ref{fig:example}, it looks more complex.

\begin{prop}[{\propseven} (W7)]
$C$ is sensitive for \textbf{permutations}, i.e. changing the start- and end-points of arcs in a workflow net can have an impact on the complexity score:
\[\exists M \in \mathcal{M}: \exists M' \in Perm(M): C(M) \neq C(M')\]
\end{prop}

Weyuker states that renaming the variables of a program should not impact its complexity.
Cardoso~\cite{Car05} translates the renaming of variables to the uniform renaming of activity names.
We also use this translation for workflow nets.

\begin{prop}[{\propeight} (W8)]
$C$ is \textbf{robust} against relabelings, i.e. uniformly changing the labeling does not affect complexity:
\[\forall M = (P,T,F,\ell,p_i,p_o) \in \mathcal{M}: \forall \ell' \in \mathcal{R}_{\ell}: C(M) = C(M_{\ell'})\]
\end{prop}

The final property concerns the combination of two process models. 
According to Weyuker, it should be possible that the complexity of a composed model is higher than the sum of complexity scores of the parts.

\begin{prop}[{\propnine} (W9)]
$C$ is \textbf{not subadditive}, i.e. the complexity score of a composed model can be greater than the sum of the complexity of its parts: Let $\opset := \{\seqop, \parop, \choiceop, \loopop\}$.
\[\forall \oplus \in \opset: \exists M_1, M_2 \in \mathcal{M}: C(M_1 \oplus M_2) > C(M_1) + C(M_2)\]
\end{prop}

The {\propnine} property is especially interesting if a measure does not fulfill it.
In this case, we know that the complexity score of a combination of two nets is always less or equal to the sum of complexity scores of the parts. 

\subsection{Extensions for Weyuker's Properties}
In this part, we define properties that extend the ones already defined.
A simple but useful property for the complexity of Petri nets was introduced by Morasca~\cite{Mor99}, who states that a complexity measure should always be defined and return non-negative values.
Obviously, this property is also important for complexity measures of process models, since negative complexity scores would be difficult to interpret.

\begin{prop}[\propdef]
The complexity measure $C$ is \textbf{defined} for every process model $M$ and returns non-negative values:
\[\forall M \in \mathcal{M}: C(M) \in \mathbb{R}^+_0\]
\end{prop}

Often, we want a complexity measure to have a minimum value to avoid generating models with peculiar structures that lower the complexity.
For example, if we want to lower the density of a model, we can add a long chain of $\tau$-transitions after the initial place.
A complexity measure not having a minimum value is a hint that this can happen.

\begin{prop}[\propmin]
The complexity measure $C$ has a \textbf{minimum} that can be reached by a process model:
\[\exists m \in \mathcal{M}: \forall M \in \mathcal{M}: C(m) \leq C(M)\]
\end{prop}

The last two properties we want to investigate are related to the {\propnine} property.
If a complexity measure fulfills {\propnine} it makes sense to analyze it for superadditivity.
While subadditive complexity measures can help in estimating the complexity of a composed net, superadditive measures imply that one should take care when using the operations of Definition~\ref{def:operations}.

\begin{prop}[\propnotsup]
$C$ is \textbf{not superadditive}, i.e. the complexity score of a composed model can be less than the sum of the complexity scores of its parts: Let $\opset := \{\seqop, \parop, \choiceop, \loopop\}$.
\[\forall \oplus \in \opset: \exists M_1, M_2 \in \mathcal{M}: C(M_1 \oplus M_2) < C(M_1) + C(M_2)\]
\end{prop} 

Similarly, if a complexity measure fulfills neither {\propnine} nor {\propnotsup}, it is interesting to investigate whether this is because the measure is additive.
Additive measures remove the need to recalculate the complexity for composed nets.

\begin{prop}[\propadd]
$C$ is \textbf{additive}, i.e. the complexity score of a composed model is exactly the sum of the complexity scores of its composed parts: Let $\opset := \{\seqop, \parop, \choiceop, \loopop\}$.
\[\forall \oplus \in \opset: \forall M_1, M_2 \in \mathcal{M}: C(M_1 \oplus M_2) = C(M_1) + C(M_2)\]
\end{prop}

Of course, one could imagine more properties for complexity measures.
We focus especially on Weyuker's properties and their extensions, since they are well-known for the analysis of software complexity measures and because they were already successfully used on a complexity measure for process models.

\section{Analysis of Complexity Measures}
\label{sec:analysis}
In this section, we analyze the complexity measures that were categorized into dimensions by Lieben et al.~\cite{LieDJJ18} regarding the properties defined in Section~\ref{sec:properties}.
We structure our analyses by the complexity dimension the complexity measures belong to.
For each measure, we first give a brief introduction in how the measure works and translate its definition for BPMN into a definition for workflow nets.
Then, we go through all properties of Section~\ref{sec:properties} and give a short indicator whether the measure fulfills the property~($\yes$) or not~($\no$), followed by a detailed justification for that claim.
Note that we repeat some example nets for different complexity measures.
This is because the analyses are intended to work like a dictionary, where the analyses of a complexity measure can be understood without reading the other analyses.

Before starting with the analyses, we show a Theorem that will prove useful for measures that are based on the connectors of a workflow net.
\begin{theorem}{Connector metric finitude} 
Let $C$ be a complexity measure and $\ran(C)$ be the set of complexity scores $C$ can return.
If $C$ takes only connectors as an input, infinitely many workflow nets can be constructed for each value $c \in \ran(C)$.
\label{thm:conn-inf}
\end{theorem}
\begin{proof}
Let $C$ be a complexity measure, $c \in \ran(C)$ and $W_c$ be a workflow net with $C(W_c) = c$. 
Then we can insert a sequence of places and transitions between the initial place and the first connector without changing the total number of connectors or their degrees. \hfill $\square$
\end{proof}
Note that if the additional transitions are silent, they also do not change the language of the model. 

\subsection{Token Behavior Complexity}
The \textsc{Token Behavior Complexity} dimension~\cite{LieDJJ18} contains complexity measures that record the number of consumed and produced tokens by executing activities.
This dimension is the largest of the ones found by Lieben et al., containing seven complexity measures in total.

\def\name{\sizename}
\def\C{\size}
\subsubsection{Size} 
The size of a process model is one of the oldest ways to measure its complexity \cite{Men08,ReiM11,Mor99}.
For BPMN models, Mendling defines this complexity measure as the number of nodes in the model \cite[p.118]{Men08}.
There are variations of this measure that take the number of arcs into account or give the transitions weights according to the complexity of the activity that it represents~\cite{Mor99}.
While these kinds of measures are also important, we directly adapt the measure of Mendling.
The analyes for the aforementioned variations are similar to the ones we perform here.
Let $W = (P,T,F,\ell,p_i,p_o)$ be an arbitrary workflow net.
\begin{equation}
\label{eq:size}
\C(W) = |P| + |T|
\end{equation}
Figure~\ref{fig:size-examples} shows some simple example nets and their respective complexity score.
\begin{figure}[ht]
\begin{center}
\begin{minipage}{0.2525\textwidth}
\centering
\scalebox{\scalefactor}{
\begin{tikzpicture}[node distance = 1.5cm,>=stealth',bend angle=0,auto]
	\node [place,tokens=1] (start) [label=below:$p_i$] {};
	\node [transition] (t1) [right of=start,label=center:$a$] {}
	edge [pre] (start);
	\node [place] (p1) [right of=t1,label=below:$p_o$] {}
	edge [pre] (t1);
	\node at (0,1.15) {$W_1^{\name}$:};
	\draw[opacity=0] ($(start)-(0.25,1.75)$) rectangle ($(p1) + (0.25,1.75)$);
\end{tikzpicture}}
\end{minipage}
\begin{minipage}{0.25\textwidth}
\centering
\scalebox{\scalefactor}{
\begin{tikzpicture}[node distance = 1.5cm,>=stealth',bend angle=0,auto]
	\node [place,tokens=1] (start) [label=below:$p_i$] {};
	\node [transition] (t2) [right of=start,label=center:$a$] {}
	edge [pre] (start);
	\node [transition] (t1) [above of=t2,label=center:$a$] {}
	edge [pre] (start);
	\node [transition] (t3) [below of=t2,label=center:$a$] {}
	edge [pre] (start);
	\node [place] (p1) [right of=t2,label=below:$p_o$] {}
	edge [pre] (t1)
	edge [pre] (t2)
	edge [pre] (t3);
	\node at (0,1.15) {$W_2^{\name}$:};
	\draw[opacity=0] ($(start)-(0.25,1.75)$) rectangle ($(p1) + (0.25,1.75)$);
\end{tikzpicture}}
\end{minipage}
\begin{minipage}{0.425\textwidth}
\centering
\scalebox{\scalefactor}{
\begin{tikzpicture}[node distance = 1.5cm,>=stealth',bend angle=0,auto]
	\node [place,tokens=1] (start) [label=below:$p_i$] {};
	\node [transition] (t1) [right of=start,label=center:$a$] {}
	edge [pre] (start);
	\node [place] (p1) [right of=t1] {}
	edge [pre] (t1);
	\node [transition] (t2) [right of=p1,label=center:$b$] {}
	edge [pre] (p1);
	\node [place] (p2) [right of=t2,label=below:$p_o$] {}
	edge [pre] (t2);
	\node at (0,1.15) {$W_3^{\name}$:};
	\draw[opacity=0] ($(start)-(0.25,1.75)$) rectangle ($(p2) + (0.25,1.75)$);
\end{tikzpicture}}
\end{minipage}
\end{center}
\caption{Three simple workflow nets, $W_1^{\name}, W_2^{\name}, W_3^{\name}$, with complexity scores $\C(W_1^{\name}) = 3$ and $\C(W_2^{\name}) = \C(W_3^{\name}) = 5$.}
\label{fig:size-examples}
\end{figure}

\begin{description}
\propitemf{\propone}{\yes} 
For the two nets $W_1^{\name}$ and $W_2^{\name}$ of Figure~\ref{fig:size-examples}, we get the complexity scores $\C(W_1^{\name}) = 3 \neq 5 = \C(W_2^{\name})$.

\propitemf{\proptwo}{\yes} 
Let $c \in \mathbb{R}$.
There are only finitely many workflow nets with complexity $c$: 
Suppose we have a directed graph $G = (V,E)$ with $|V| = c$. 
Then, $|E| \leq c^2$, since all nodes may be connected to all other nodes.
If we want to choose only some of these edges, we get $\sum_{i = 0}^{c^2} \binom{c^2}{i} = 2^{c^2}$ possible edge sets $E$.
Not all graphs that we can construct in this way can be converted to workflow nets, but every workflow net of size $c$ can be interpreted as such a graph.
Therefore, $|\{M \in \mathcal{M} \mid \C(M) = c\}| < 2^{c^2} < \infty$.

\propitemf{\propthree}{\yes} 
The workflow nets $W_2^{\name}$ and $W_3^{\name}$ of Figure~\ref{fig:size-examples} are different in structure, but both get the complexity score $\C(W_2^{\name}) = 5 = \C(W_3^{\name})$.

\propitemf{\propfour}{\yes} 
Take the workflow nets $W_1^{\name}$ and $W_2^{\name}$ of Figure~\ref{fig:size-examples}. 
Their languages are $L(W_1^{\name}) = \{\varepsilon, a\} = L(W_2^{\name})$, but their complexity scores are $\C(W_1^{\name}) = 3 \neq 5 = \C(W_2^{\name})$.

\propitemf{\propfive}{\yes} 
For this claim, we use the following Theorem:
\begin{theorem}
\label{thm:size}
Let $n \geq 2$. For any workflow nets $M_1, \dots, M_n \in \mathcal{M}$, we get:
\begin{itemize}
\item $\C(\seqop(M_1, \dots, M_n)) = \C(M_1) + \dots + \C(M_n) + n-1$,
\item $\C(\parop(M_1, \dots, M_n)) = \C(M_1) + \dots + \C(M_n) + 4$,
\item $\C(\choiceop(M_1, \dots, M_n)) = \C(M_1) + \dots + \C(M_n) + 2n+2$ and
\item $\C(\loopop(M_1, \dots, M_n)) = \C(M_1) + \dots + \C(M_n) + 2n+6$.
\end{itemize}
\end{theorem}
\begin{proof}
Follows directly from Definition~\ref{def:operations}. \hfill $\square$
\end{proof}
With Theorem~\ref{thm:size}, we immediately get that 
\[\C(\oplus(M_1, \dots, M_n)) = \C(M_1) + \dots + \C(M_n) + k\]
for $\oplus \in \{\seqop, \parop, \choiceop, \loopop\}$ and some $k > 0$. 
Therefore, we can deduce that for any $M_i \in \{M_1, \dots, M_n\}$ we have:
\[\C(\oplus(M_1, \dots, M_n)) \geq \C(M_i).\]

\propitemf{\propsix}{\no}
Let $M_1, M_2, M_3 \in \mathcal{M}$ be three arbitrary workflow nets, where $\C(M_1) = \C(M_2)$.
Theorem~\ref{thm:size} gives:
\begin{itemize}
\item $\C(M_1 \seqop M_3) = \C(M_1) + \C(M_3) + 1$ \\
\phantom{$\C(M_1\seqop M_3)$ }$= \C(M_2) + \C(M_3) + 1 = \C(M_2 \seqop M_3)$,
\item $\C(M_1 \parop M_3) = \C(M_1) + \C(M_3) + 4$ \\
\phantom{$\C(M_1 \parop M_3)$ }$= \C(M_2) + \C(M_3) + 4 = \C(M_2 \parop M_3)$,
\item $\C(M_1 \choiceop M_3) = \C(M_1) + \C(M_3) + 6$ \\
\phantom{$\C(M_1 \choiceop M_3)$ }$= \C(M_2) + \C(M_3) + 6 = \C(M_2 \choiceop M_3)$ and
\item $\C(M_1 \loopop M_3) = \C(M_1) + \C(M_3) + 10$ \\
\phantom{$\C(M_1 \loopop M_3)$ }$= \C(M_2) + \C(M_3) + 10 = \C(M_2 \loopop M_3)$.
\end{itemize}
Thus, none of the operations of Definition~\ref{def:operations} are sensitive to compositions.

\propitemf{\propseven}{\yes} 
According to our definition of permutations, the net $W_4^{\name}$ of Figure~\ref{fig:size-perm} is a permutation of the net $W_3^{\name}$ of Figure~\ref{fig:size-examples}, but their complexity scores are $\C(W_3^{\name}) = 5 \neq 6 = \C(W_4^{\size})$.
\begin{figure}[ht]
\centering
\scalebox{\scalefactor}{
\begin{tikzpicture}[node distance = 1.5cm,>=stealth',bend angle=0,auto]
	\node [place,tokens=1] (start) [label=below:$p_i$] {};
	\node [transition] (t1) [right of=start,label=center:$a$] {}
	edge [pre] (start);
	\node [place] (p1) [above right of=t1] {}
	edge [pre] (t1);
	\node [place] (p2) [below right of=t1] {}
	edge [pre] (t1);
	\node [transition] (t2) [below right of=p1,label=center:$b$] {}
	edge [pre] (p1)
	edge [pre] (p2);
	\node [place] (end) [right of=t2,label=below:$p_o$] {}
	edge [pre] (t2);
	\node at (0,1.15) {$W_4^{\name}$:};
\end{tikzpicture}}
\caption{A workflow net $W_4^{\name}$ with complexity score $\C(W_4^{\name}) = 6$.}
\label{fig:size-perm}
\end{figure}

\propitemf{\propeight}{\yes} 
The labeling of the transitions doesn't influence the number of places and transitions in a workflow net. 

\propitemf{\propnine}{\yes} 
According to Theorem~\ref{thm:size}, we can take any $M_1, M_2 \in \mathcal{M}$, as well as any $\oplus \in \{\seqop, \parop, \choiceop, \loopop\}$ and get $\C(M_1 \oplus M_2) > \C(M_1) + \C(M_2)$.

\propitemf{\propdef}{\yes} $\C$ is defined for all workflow nets $M = (P,T,F,p_i,p_o) \in \mathcal{M}$.
By definition, $|P|$ and $|T|$ are non-negative values, so $\C(M) \geq 0$.

\propitemf{\propmin}{\yes} 
By definition, the workflow net with the smallest amount of nodes is the net $\netminwf$ of Figure~\ref{fig:netminwf}.
So $\C(\netminwf) = 3$ is the minimum score of this complexity metric.

\propitemf{\propinf}{\yes} 
For every $c \in \mathbb{N}$ with $c \geq 3$, we can find a workflow net with complexity $c$ by introducing two places $p_i$ and $p_o$, as well as $c - 2$ transitions that have an edge coming from $p_i$ and one leading to $p_o$. 
Therefore, we get $|\{c \in \mathbb{R} \mid \exists M \in \mathcal{M}: \C(M) = c\}| \geq \{c \in \mathbb{N} \mid c \geq 3\}| = \infty$.

\propitemf{\propnotsup}{\no} Follows directly from Theorem~\ref{thm:size}.

\propitemf{\propadd}{\no} Follows directly from Theorem~\ref{thm:size}.
\end{description}

\def\name{\mismatchname}
\def\C{\mismatch}
\subsubsection{Connector Mismatch}
A connector mismatch happens when a split connector of one type is paired with a join connector of another type.
Such mismatches lead to undesired behavior for the tokens in the net.
Furthermore, since humans tend to prefer symmetry, such a structure is difficult to understand and indicates errors in the model~\cite{Men08}.
Analyzing how many mismatches occur in a workflow net is not trivial, as it requires analyzing all paths starting from a split-node. 
In a graph, with $n$ vertices, the number of paths can be exponential in $n$, so this would be a time-intensive analysis.

To approximate the amount of connector mismatches, Mendling~\cite[p.125]{Men08} defines the number of mismatches of a type as the difference of arcs leaving a split node and arcs entering a join node of that type.
We adapt this definition for workflow nets by defining, for $W = (P,T,F,\ell,p_i,p_o) \in \mathcal{M}$,
\begin{equation}
\label{eq:mismatch-and}
MM_{\text{and}}^W := \left|\sum_{t \in \mathcal{S}_{\text{and}}^W} |\post{t}| - \sum_{t \in \mathcal{J}_{\text{and}}^W} |\pre{t}|\right|
\end{equation}
as the mismatch between concurrent split and join connectors and 
\begin{equation}
\label{eq:mismatch-xor}
MM_{\text{xor}}^W := \left|\sum_{p \in \mathcal{S}_{\text{xor}}^W} |\post{t}| - \sum_{p \in \mathcal{J}_{\text{xor}}^W} |\pre{p}|\right|
\end{equation}
as the mismatch between xor split and join connectors.
According to the measure defined by Mendling, the total amount of mismatches in the model is defined as the sum of these mismatches. 
For workflow nets, we therefore define
\begin{equation}
\label{eq:c-mismatch}
\C(W) = MM_{\text{and}}^W + MM_{\text{xor}}^W.
\end{equation}
The original definition also takes or-connectors into account.
Since we don't have designated connectors for inclusive choices in workflow nets, we ignore this part of the complexity measure.
Figure~\ref{fig:mm-examples} shows some simple example nets and their respective complexity score.
\begin{figure}[ht]
\begin{center}
\begin{minipage}{0.21\textwidth}
\centering
\scalebox{\scalefactor}{
\begin{tikzpicture}[node distance = 1.5cm,>=stealth',bend angle=0,auto]
	\node [place,tokens=1] (start) [label=below:$p_i$] {};
	\node [transition] (t1) [above right of=start,label=center:$a$] {}
	edge [pre] (start);
	\node [transition] (t2) [below right of=start,label=center:$b$] {}
	edge [pre] (start);
	\node [place] (p1) [below right of=t1,label=below:$p_o$] {}
	edge [pre] (t1)
	edge [pre] (t2);
	\node at (0,1.15) {$W_1^{\name}$:};
\end{tikzpicture}}
\end{minipage}
\begin{minipage}{0.37\textwidth}
\centering
\scalebox{\scalefactor}{
\begin{tikzpicture}[node distance = 1.5cm,>=stealth',bend angle=0,auto]
	\node [place,tokens=1] (start) [label=below:$p_i$] {};
	\node [transition] (t1) [above right of=start,label=center:$a$] {}
	edge [pre] (start);
	\node [transition] (t2) [below right of=start,label=center:$b$] {}
	edge [pre] (start);
	\node [place] (p1) [right of=t1] {}
	edge [pre] (t1);
	\node [place] (p2) [right of=t2] {}
	edge [pre] (t2);
	\node [transition] (t3) [below right of=p1,label=center:$c$] {}
	edge [pre] (p1)
	edge [pre] (p2);
	\node [place] (p3) [right of=t3,label=below:$p_o$] {}
	edge [pre] (t3);
	\node at (0,1.15) {$W_2^{\name}$:};
\end{tikzpicture}}
\end{minipage}
\begin{minipage}{0.37\textwidth}
\centering
\scalebox{\scalefactor}{
\begin{tikzpicture}[node distance = 1.5cm,>=stealth',bend angle=0,auto]
	\node [place,tokens=1] (start) [label=below:$p_i$] {};
	\node [transition] (t1) [right of=start,label=center:$a$] {}
	edge [pre] (start);
	\node [place] (p1) [above right of=t1] {}
	edge [pre] (t1);
	\node [place] (p2) [below right of=t1] {}
	edge [pre] (t1);
	\node [transition] (t3) [right of=p1,label=center:$b$] {}
	edge [pre] (p1);
	\node [transition] (t4) [right of=p2,label=center:$c$] {}
	edge [pre] (p2);
	\node [place] (p3) [below right of=t3,label=below:$p_o$] {}
	edge [pre] (t3)
	edge [pre] (t4);
	\node at (0,1.15) {$W_3^{\name}$:};
\end{tikzpicture}}
\end{minipage}
\end{center}
\caption{Three simple workflow nets, $W_1^{\name}, W_2^{\name}, W_3^{\name}$, with complexity scores $\C(W_1^{\name}) = 0$ and $\C(W_2^{\name}) = \C(W_3^{\name}) = 4$.}
\label{fig:mm-examples}
\end{figure}

\begin{description}
\propitemf{\propone}{\yes} 
For the two nets $W_1^{\name}$ and $W_2^{\name}$ of Figure~\ref{fig:mm-examples}, we get the complexity scores $\C(W_1^{\name}) = 0 \neq 4 = \C(W_2^{\name})$.

\propitemf{\proptwo}{\no}
Follows directly from Theorem~\ref{thm:conn-inf}. 
Figure~\ref{fig:mm-fin} shows how to construct such a sequence of places and transitions.
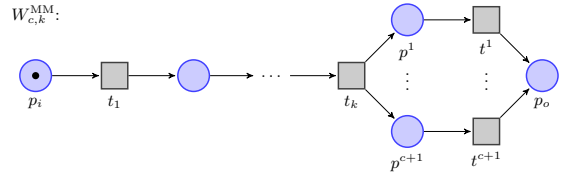
\begin{figure}[ht]
\begin{center}
\scalebox{\scalefactor}{
\begin{tikzpicture}[node distance = 1.5cm,>=stealth',bend angle=0,auto]
	\node [place,tokens=1] (start) [label=below:$p_i$] {};
	\node [transition] (t1) [right of=start,label=below:$t_1$] {}
	edge [pre] (start);
	\node [place] (p1) [right of=t1] {}
	edge [pre] (t1);
	\node (dots) [right of=p1] {$\dots$}
	edge [pre] (p1);
	\node [transition] (tk) [right of=dots,label=below:$t_k$] {}
	edge [pre] (dots);
	\node [place] (p1) [above right of=tk,label=below:$p^1$] {}
	edge [pre] (tk);
	\node [place] (pk) [below right of=tk,label=below:$p^{c+1}$] {}
	edge [pre] (tk);
	\node [transition] (t1) [right of=p1,label=below:$t^1$] {}
	edge [pre] (p1);
	\node [transition] (tk) [right of=pk,label=below:$t^{c+1}$] {}
	edge [pre] (pk);
	\node [place] (end) [below right of=t1,label=below:$p_o$] {}
	edge [pre] (t1)
	edge [pre] (tk);
	\node at ($0.5*(p1) + 0.5*(pk)$) {$\vdots$};
	\node at ($0.5*(t1) + 0.5*(tk)$) {$\vdots$};
	\node at (0,1.15) {$W_{c,k}^{\name}$:};
\end{tikzpicture}}
\end{center}
\caption{A workflow net $W_{c,k}^{\name}$ with $k + c + 1$ transitions and complexity score $\C(W_{c,k}^{\name}) = 2 \cdot (c+1)$.}
\label{fig:mm-fin}
\end{figure}

\propitemf{\propthree}{\yes} 
The workflow nets $W_2^{\name}$ and $W_3^{\name}$ of Figure~\ref{fig:mm-examples} are different in structure, but both get the complexity score $\C(W_2^{\name}) = 4 = \C(W_3^{\name})$.

\propitemf{\propfour}{\yes} 
Take the workflow nets $W_1^{\name}$ and $W_2^{\name}$ of Figure~\ref{fig:mm-examples}.
Their languages are $L(W_1^{\name}) = \{\varepsilon, a, b\} = L(W_2^{\name})$, but their complexity scores are $\C(W_1^{\name}) = 0 \neq 4 = \C(W_2^{\name})$.

\propitemf{\propfive}{\no} 
Take the workflow nets $W_2^{\name}$ and $W_3^{\name}$ of Figure~\ref{fig:mm-examples}, as well as any operation $\oplus \in \{\seqop, \parop, \choiceop, \loopop\}$. 
Each operation introduces a split-connector of exactly one type, as well as a join-connector of the same type.
The number of outgoing arcs of the split-connector are exactly the number of incoming arcs of the join-connector.
However, $W_2^{\name}$ introduces $2$ arcs leaving an and-connector and $2$ arcs entering an xor-connector, while $W_3^{\name}$ introduces $2$ arcs leaving an xor-connector and $2$ arcs entering an and-connector.
In total, we therefore get $\C(W_2^{\name} \oplus W_3^{\name}) = 0 < 4 = \C(W_2^{\name})$.

\propitemf{\propsix}{\yes} 
Take the workflow nets $W_2^{\name}$ and $W_3^{\name}$ of Figure~\ref{fig:mm-examples}. 
As the caption of this figure suggests, $\C(W_2^{\name}) = \C(W_3^{\name})$, but if we combine these nets with any of the operations $\oplus \in \{\seqop, \parop, \choiceop, \loopop\}$ with $W_3^{\name}$, we get $\C(W_2^{\name} \oplus W_3^{\name}) = 0 \neq 8 = \C(W_3^{\name} \oplus W_3^{\name})$.

\propitemf{\propseven}{\yes} 
Figure~\ref{fig:mm-perm} shows two nets with different complexity scores according to $\C$ that are permutations of each other.
\begin{figure}[ht]
\begin{center}
\scalebox{\scalefactor}{
\begin{tikzpicture}[node distance = 1.5cm,>=stealth',bend angle=0,auto]
	\node [place,tokens=1] (start) [label=below:$p_i$] {};
	\node [transition] (t1) [above right of=start,label=below:$t_1$] {}
	edge [pre] (start);
	\node [transition] (t2) [below right of=start,label=below:$t_2$] {}
	edge [pre] (start);
	\node [place] (p1) [right of=t1,label=below:$p_1$] {}
	edge [pre] (t1);
	\node [place] (p3) [right of=t2,label=below:$p_3$] {}
	edge [pre] (t2);
	\node [place] (p2) at ($0.5*(p1) + 0.5*(p3)$) [label=below:$p_2$] {}
	edge [pre] (t2);
	\node [transition] (t3) [below right of=p1,label=below:$t_3$] {}
	edge [pre] (p1)
	edge [pre] (p2)
	edge [pre] (p3);
	\node [place] (p4) [right of=t3,label=below:$p_4$] {}
	edge [pre] (t3);
	\node [transition] (t4) [right of=p4,label=below:$t_4$] {}
	edge [pre] (p4);
	\node [place] (end) [right of=t4,label=below:$p_o$] {}
	edge [pre] (t4);
	\node at (0,1.15) {$W_4^{\name}$:};
\end{tikzpicture}}
\ \\
\ \\
\scalebox{\scalefactor}{
\begin{tikzpicture}[node distance = 1.5cm,>=stealth',bend angle=0,auto]
	\node [place,tokens=1] (start) [label=below:$p_i$] {};
	\node [transition] (t1) [right of=start,label=below:$t_1$] {}
	edge [pre] (start);
	\node [place] (p1) [above right of=t1,label=below:$p_1$] {}
	edge [pre] (t1);
	\node [place] (p2) [below right of=t1,label=below:$p_2$] {}
	edge [pre] (t1);
	\node [transition] (t2) [right of=p1,label=below:$t_2$] {}
	edge [pre] (p1);
	\node [transition] (t3) [right of=p2,label=below:$t_3$] {}
	edge [pre] (p2);
	\node [place] (p3) [right of=t2,label=below:$p_3$] {}
	edge [pre] (t2);
	\node [place] (p4) [right of=t3,label=below:$p_4$] {}
	edge [pre] (t3);
	\node [transition] (t4) [below right of=p3,label=below:$t_4$] {}
	edge [pre] (p3)
	edge [pre] (p4);
	\node [place] (end) [right of=t4,label=below:$p_o$] {}
	edge [pre] (t4);
	\node at (0,1.15) {$W_5^{\name}$:};
\end{tikzpicture}}
\end{center}
\caption{Two workflow nets, $W_4^{\name}$ and $W_5^{\name}$, where $W_4^{\name} \in Perm(W_5^{\name})$. $W_5^{\name}$ contains no connector mismatches, so $\C(W_5^{\name}) = 0$, while $\C(W_4^{\name}) = 2$.}
\label{fig:mm-perm}
\end{figure}
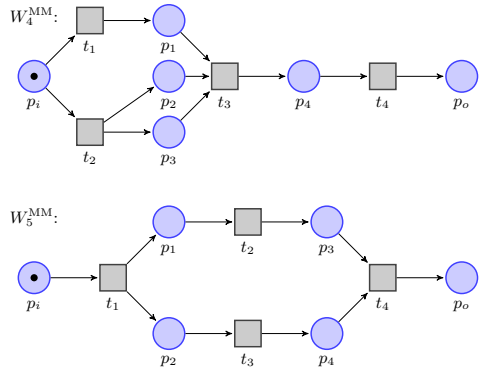

\propitemf{\propeight}{\yes} 
$\C$ depends only on the connector types in the net and their out- and in-degree.
However, we cannot change these by performing a relabeling on the net.

\propitemf{\propnine}{\no}
Let $M_1 = (P_1, T_1, F_1, p_i^1, p_o^1)$ and $M_2 = (P_2, T_2, F_2, p_i^2, p_o^2)$ be workflow nets and let $M = M_1 \oplus M_2$ for some $\oplus \in \{\seqop, \parop, \choiceop, \loopop\}$.
We get:
\begin{align*}
MM_{\text{and}}^M &= \left| \sum_{t \in \mathcal{S}_{\text{and}}^M} |\post{t}| - \sum_{t \in \mathcal{J}_{\text{and}}^M} |\pre{t}| \right| \\
&= \left| \sum_{t \in \mathcal{S}_{\text{and}}^{M_1} \cup \mathcal{S}_{\text{and}}^{M_2}} |\post{t}| - \sum_{t \in \mathcal{J}_{\text{and}}^{M_1} \cup \mathcal{J}_{\text{and}}^{M_2}} |\pre{t}| \right| \\
&\leq \left| \sum_{t \in \mathcal{S}_{\text{and}}^{M_1}} |\post{t}| - \sum_{t \in \mathcal{J}_{\text{and}}^{M_1}} |\pre{t}| \right| +  \left| \sum_{t \in \mathcal{S}_{\text{and}}^{M_2}} |\post{t}| - \sum_{t \in \mathcal{J}_{\text{and}}^{M_2}} |\pre{t}| \right| \\
&= MM_{\text{and}}^{M_1} + MM_{\text{and}}^{M_2}
\end{align*}
where the second equality is due to the fact that $T_1 \cap T_2 = \emptyset$ by definition.
Analogously, we get $MM_{\text{xor}}^M \leq MM_{\text{xor}}^{M_1} + MM_{\text{xor}}^{M_2}$ and therefore
\begin{align*}
\C(W) &= MM_{\text{and}}^M + MM_{\text{xor}}^M \\
&\leq (MM_{\text{and}}^{M_1} + MM_{\text{and}}^{M_2}) + (MM_{\text{xor}}^{M_1} + MM_{\text{xor}}^{M_2}) \\
&= (MM_{\text{and}}^{M_1} + MM_{\text{xor}}^{M_1}) + (MM_{\text{and}}^{M_2} + MM_{\text{xor}}^{M_2}) \\
&= \C(M_1) + \C(M_2).
\end{align*}

\propitemf{\propdef}{\yes} 
Since $\mathcal{S}_{\text{and}}^W$, $\mathcal{S}_{\text{xor}}^W$, $\mathcal{J}_{\text{and}}^W$, $\mathcal{J}_{\text{xor}}^W$ are defined for all workflow nets $W \in \mathcal{M}$, $\C(W)$ is also defined for all such workflow nets.
Furthermore, since we sum only over absolute values to compute $\C(W)$, we get $\C(W) \geq 0$ for all workflow nets $W \in \mathcal{M}$.

\propitemf{\propmin}{\yes} The minimum for this metric is $0$.
A workflow net without any connectors, like $\netminwf$ of Figure~\ref{fig:netminwf}, receives this complexity score.

\propitemf{\propinf}{\yes} 
For any $c \in \mathbb{N}$, Figure~\ref{fig:mm-fin} shows how to construct a workflow net with complexity score $2c + 2$. 
Therefore, there are infinitely many possible complexity scores: \\
$|\{c \in \mathbb{R} \mid \exists M \in \mathcal{M}: \C(M) = c\}| \geq |\{n \in \mathbb{N} \mid \exists c \in \mathbb{N}: n = 2c + 2\}| = \infty$.

\propitemf{\propnotsup}{\yes} 
Consider again the workflow nets $W_2^{\name}$ and $W_3^{\name}$ in Figure~\ref{fig:mm-examples}.
Since none of the operations $\oplus \in \{\seqop, \parop, \choiceop, \loopop\}$ introduce new connector mismatches, we get $\C(W_2^{\name} \oplus W_3^{\name}) = 0 < 8 = \C(W_2^{\name}) + \C(W_3^{\name})$.

\propitemf{\propadd}{\no} 
We can use the same counter-example as for the property {\propnotsup} to see that $\C$ is not additive.
\end{description}

\def\name{\connhetname}
\def\C{\connhet}
\subsubsection{Connector Heterogeneity}
According to Mendling~\cite[pp.126-127]{Men08}, an EPC model that contains many connector types is more prone to errors, since it is more likely to contain a connector mismatch.
He therefore introduces the connector heterogeneity metric, which calculates the entropy of connectors in the model.
Since we don't have designated connectors for inclusive choices, we only calculate the entropy of concurrent and exclusive choice connectors.
Let $W = (P,T,F,\ell,p_i,p_o)$ be a workflow net.
\begin{equation}
\label{eq:c-heterogeneity}
\C(W) = - \left(\frac{|\mathcal{C}_{\text{and}}^W|}{|\mathcal{C}^W|} \cdot \log_2\left(\frac{|\mathcal{C}_{\text{and}}^W|}{|\mathcal{C}^W|}\right) + \frac{|\mathcal{C}_{\text{xor}}^W|}{|\mathcal{C}^W|} \cdot \log_2\left(\frac{|\mathcal{C}_{\text{xor}}^W|}{|\mathcal{C}^W|}\right)\right)
\end{equation}
A value near $0$ means that there are only connectors of one type present in the net, while a value near $1$ means that we have equally many parallel and choice connectors.
Figure~\ref{fig:ch-examples} shows some simple example nets and their respective complexity score.

\begin{figure}[ht]
\begin{center}
\begin{minipage}{0.36\textwidth}
\centering
\scalebox{\scalefactor}{
\begin{tikzpicture}[node distance = 1.5cm,>=stealth',bend angle=0,auto]
	\node [place,tokens=1] (start) [label=below:$p_i$] {};
	\node [transition] (t1) [above right of=start,label=center:$a$] {}
	edge [pre] (start);
	\node [transition] (t2) [below right of=start,label=center:$b$] {}
	edge [pre] (start);
	\node [place] (p1) [right of=t1,label=below:$p_1$] {}
	edge [pre] (t1);
	\node [place] (p3) [right of=t2,label=below:$p_3$] {}
	edge [pre] (t2);
	\node [place] (p2) at ($0.5*(p1) + 0.5*(p3)$) [label=below:$p_2$] {}
	edge [pre] (t1)
	edge [pre] (t2);
	\node [transition] (t3) [right of=p1,label=center:$c$] {}
	edge [pre] (p1)
	edge [pre] (p2);
	\node [transition] (t4) [right of=p3,label=center:$d$] {}
	edge [pre] (p2)
	edge [pre] (p3);
	\node [place] (end) [below right of=t3,label=below:$p_o$] {}
	edge [pre] (t3)
	edge [pre] (t4);
	\node at (0,1.15) {$W_1^{\name}$:};
	\draw[opacity=0] (-0.35,-1.75) rectangle (5.5,1.5);
\end{tikzpicture}}
\end{minipage}
\begin{minipage}{0.36\textwidth}
\centering
\scalebox{\scalefactor}{
\begin{tikzpicture}[node distance = 1.5cm,>=stealth',bend angle=0,auto]
	\node [place,tokens=1] (start) [label=below:$p_i$] {};
	\node [transition] (t1) [above right of=start,label=center:$a$] {}
	edge [pre] (start);
	\node [transition] (t2) [below right of=start,label=center:$b$] {}
	edge [pre] (start);
	\node [place] (p1) [right of=t1,label=below:$p_1$] {}
	edge [pre] (t1);
	\node [place] (p2) [right of=t2,label=below:$p_2$] {}
	edge [pre] (t2);
	\node [transition] (t3) [right of=p1,label=center:$c$] {}
	edge [pre] (p1);
	\node [transition] (t4) [right of=p2,label=center:$d$] {}
	edge [pre] (p2);
	\node [place] (p3) [below right of=t3,label=below:$p_o$] {}
	edge [pre] (t3)
	edge [pre] (t4);
	\node at (0,1.15) {$W_2^{\name}$:};
	\draw[opacity=0] (-0.35,-1.75) rectangle (5.5,1.5);
\end{tikzpicture}}
\end{minipage}
\begin{minipage}{0.2\textwidth}
\centering
\scalebox{\scalefactor}{
\begin{tikzpicture}[node distance = 1.5cm,>=stealth',bend angle=0,auto]
	\node [place,tokens=1] (start) [label=below:$p_i$] {};
	\node [transition] (t1) [above right of=start,label=center:$a$] {}
	edge [pre] (start);
	\node [transition] (t3) [below right of=start,label=center:$c$] {}
	edge [pre] (start);
	\node [transition] (t2) at ($0.5*(t1) + 0.5*(t3)$) [label=center:$b$] {}
	edge [pre] (start);
	\node [place] (p3) [below right of=t1,label=below:$p_o$] {}
	edge [pre] (t1)
	edge [pre] (t2)
	edge [pre] (t3);
	\node at (0,1.15) {$W_3^{\name}$:};
	\draw[opacity=0] (-0.35,-1.75) rectangle (2.5,1.5);
\end{tikzpicture}}
\end{minipage}
\end{center}
\caption{Workflow nets, $W_1^{\name}, W_2^{\name}, W_3^{\name}$, with $\C(W_2^{\name}) = \C(W_3^{\name}) = 0$ and $\C(W_1^{\name}) = -(\frac{4}{7} \cdot \log_2(\frac{4}{7}) + \frac{3}{7} \cdot \log_2(\frac{3}{7})) \approx 0.9852$.}
\label{fig:ch-examples}
\end{figure}

\begin{description}
\propitemf{\propone}{\yes} 
For the two nets $W_1^{\name}$ and $W_2^{\name}$ of Figure~\ref{fig:ch-examples}, we get the complexity scores $\C(W_1^{\name}) \approx 0.9852 \neq 0 = \C(W_2^{\name})$.

\propitemf{\proptwo}{\no}
Follows directly from Theorem~\ref{thm:conn-inf}. 
Figure~\ref{fig:ch-fin} shows how to construct such a sequence of places and transitions.
\begin{figure}[ht]
\begin{center}
\scalebox{\scalefactor}{
\begin{tikzpicture}[node distance = 1.5cm,>=stealth',bend angle=0,auto]
	\node [place,tokens=1] (start) [label=below:$p_i$] {};
	\node [transition] (t1) [right of=start,label=below:$t_1$] {}
	edge [pre] (start);
	\node [place] (p0) [right of=t1] {}
	edge [pre] (t1);
	\node (dots) [right of=p0] {$\dots$}
	edge [pre] (p0);
	\node [transition] (tk) [right of=dots,label=below:$t_k$] {}
	edge [pre] (dots);
	\node [place] (p) [right of=tk,label=below:$p$] {}
	edge [pre] (tk);
	\node [transition] (t11) [above right of=p,yshift=0.75cm,label=below:$t_1^1$] {}
	edge [pre] (p);
	\node [place] (p1) [yshift=-0.5cm,above right of=t11] {}
	edge [pre] (t11);
	\node [place] (p2) [yshift=0.5cm,below right of=t11] {}
	edge [pre] (t11);
	\node [transition] (t12) [right of=p1,label=below:$t_1^2$] {}
	edge [pre] (p1);
	\node [transition] (t13) [right of=p2,label=below:$t_1^3$] {}
	edge [pre] (p2);
	\node [place] (p5) [right of=t12] {}
	edge [pre] (t12);
	\node [place] (p6) [right of=t13] {}
	edge [pre] (t13);
	\node [transition] (t14) [below right of=p5,yshift=0.5cm,label=below:$t_1^4$] {}
	edge [pre] (p5)
	edge [pre] (p6);
	\node [transition] (tn1) [below right of=p,yshift=-0.65cm,label=below:$t_n^1$] {}
	edge [pre] (p);
	\node [place] (p3) [yshift=-0.5cm,above right of=tn1] {}
	edge [pre] (tn1);
	\node [place] (p4) [yshift=0.5cm,below right of=tn1] {}
	edge [pre] (tn1);
	\node [transition] (tn2) [right of=p3,label=below:$t_n^2$] {}
	edge [pre] (p3);
	\node [transition] (tn3) [right of=p4,label=below:$t_n^3$] {}
	edge [pre] (p4);
	\node [place] (p7) [right of=tn2] {}
	edge [pre] (tn2);
	\node [place] (p8) [right of=tn3] {}
	edge [pre] (tn3);
	\node [transition] (tn4) [below right of=p7,yshift=0.5cm,label=below:$t_n^4$] {}
	edge [pre] (p7)
	edge [pre] (p8);
	\node [place] (end) [below right of=t14,yshift=-0.65cm,label=below:$p_o$] {}
	edge [pre] (t14)
	edge [pre] (tn4);
	\node at ($0.5*(t11) + 0.5*(tn1)$) {$\vdots$};
	\node at ($0.5*(t13) + 0.5*(tn2)$) {$\vdots$};
	\node at ($0.5*(t14) + 0.5*(tn4)$) {$\vdots$};
	\node at (0,1.15) {$W_{k,n}^{\name}$:};
\end{tikzpicture}}
\end{center}
\caption{A construction plan for workflow nets $W_{k,n}^{\name}$ with $k + 4n$ transitions and $\C(W_{k,n}^{\name}) = -\left(\frac{1}{n+1} \cdot \log_2(\frac{1}{n+1}) + \frac{n}{n+1} \cdot \log_2(\frac{n}{n+1})\right)$ for $k,n \in \mathbb{N}$ with $n \geq 2$.}
\label{fig:ch-fin}
\end{figure}
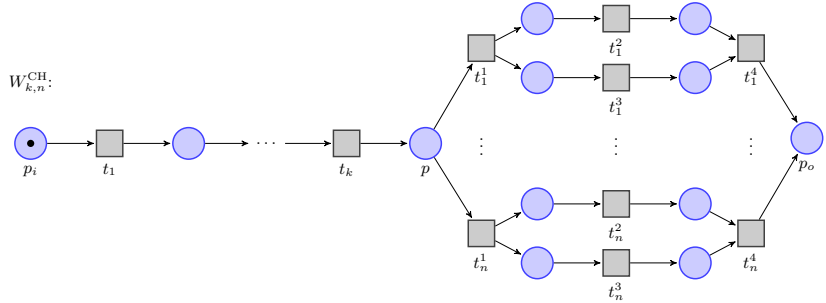

\propitemf{\propthree}{\yes} 
The workflow nets $W_2^{\name}$ and $W_3^{\name}$ of Figure~\ref{fig:ch-examples} are different in structure, but both get the complexity score $\C(W_2^{\name}) = 0 = \C(W_3^{\name})$.

\propitemf{\propfour}{\yes} 
Take the workflow nets $W_1^{\name}$ and $W_2^{\name}$ of Figure~\ref{fig:ch-examples}.
Their languages are $L(W_1^{\name}) = \{\varepsilon, a, b, ac, bd\} = L(W_2^{\name})$, but their complexity scores are $\C(W_1^{\name}) = 0.9852 \neq 0 = \C(W_2^{\name})$.

\propitemf{\propfive}{\no} 
The measure $\C$ does not fulfill this property for any of the operations $\oplus \in \{\seqop, \parop, \choiceop, \loopop\}$. 
To construct counter examples, consider the workflow net $W_1^{\name}$ of Figure~\ref{fig:ch-examples} and the workflow net $W_{1,3}^{\name}$ of Figure~\ref{fig:ch-fin}.
\begin{itemize}
\item $\C(W_1^{\name} \seqop W_{1,3}^{\name}) = -\left(\frac{10}{12} \cdot \log_2(\frac{10}{12}) + \frac{2}{12} \cdot \log_2(\frac{2}{12})\right) \approx 0.6500$ \\
\phantom{$\C(W_1^{\name} \seqop W_{1,3}^{\name})$ }$< 0.9852 \approx \C(W_1^{\name})$,
\item $\C(W_1^{\name} \parop W_{1,3}^{\name}) = -\left(\frac{12}{14} \cdot \log_2(\frac{12}{14}) + \frac{2}{14} \cdot \log_2(\frac{2}{14})\right) \approx 0.5917$ \\
\phantom{$\C(W_1^{\name} \parop W_{1,3}^{\name})$ }$< 0.9852 \approx \C(W_1^{\name})$,
\item $\C(W_1^{\name} \choiceop W_{1,3}^{\name}) = -\left(\frac{10}{14} \cdot \log_2(\frac{10}{14}) + \frac{4}{14} \cdot \log_2(\frac{4}{14})\right) \approx 0.8631$ \\
\phantom{$\C(W_1^{\name} \choiceop W_{1,3}^{\name})$ }$< 0.9852 \approx \C(W_1^{\name})$,
\item $\C(W_1^{\name} \loopop W_1^{\name}) = -\left(\frac{10}{14} \cdot \log_2(\frac{10}{14}) + \frac{4}{14} \cdot \log_2(\frac{4}{14})\right) \approx 0.8631$ \\
\phantom{$\C(W_1^{\name} \loopop W_1^{\name})$ }$< 0.9852 \approx \C(W_1^{\name})$.
\end{itemize}

\propitemf{\propsix}{\yes} 
Consider the two workflow nets $W_4^{\name}$ and $W_5^{\name}$ of Figure~\ref{fig:ch-comp} and the workflow net $W_3^{\name}$ of Figure~\ref{fig:ch-examples}.
\begin{figure}[ht]
\begin{center}
\scalebox{\scalefactor}{
\begin{tikzpicture}[node distance = 1.5cm,>=stealth',bend angle=0,auto]
	\node [place,tokens=1] (start) [label=below:$p_i$] {};
	\node [transition] (t1) [above right of=start] {}
	edge [pre] (start);
	\node [transition] (t2) [below right of=start] {}
	edge [pre] (start);
	\node [place] (p1) [right of=t1,yshift=0.5cm] {}
	edge [pre] (t1);
	\node [place] (p2) [right of=t1,yshift=-0.5cm] {}
	edge [pre] (t1);
	\node [place] (p3) [right of=t2,yshift=0.5cm] {}
	edge [pre] (t2);
	\node [place] (p4) [right of=t2,yshift=-0.5cm] {}
	edge [pre] (t2);
	\node [transition] (t3) [right of=p1] {}
	edge [pre] (p1);
	\node [transition] (t4) [right of=p2] {}
	edge [pre] (p2);
	\node [transition] (t5) [right of=p3] {}
	edge [pre] (p3);
	\node [transition] (t6) [right of=p4] {}
	edge [pre] (p4);
	\node [place] (p5) [right of=t3] {}
	edge [pre] (t3);
	\node [place] (p6) [right of=t4] {}
	edge [pre] (t4);
	\node [place] (p7) [right of=t5] {}
	edge [pre] (t5);
	\node [place] (p8) [right of=t6] {}
	edge [pre] (t6);
	\node [transition] (t7) [right of=p5,yshift=-0.5cm] {}
	edge [pre] (p5)
	edge [pre] (p6);
	\node [transition] (t8) [right of=p7,yshift=-0.5cm] {}
	edge [pre] (p7)
	edge [pre] (p8);
	\node [place] (end) [below right of=t7,label=below:$p_o$] {}
	edge [pre] (t7)
	edge [pre] (t8);
	\node at (0,1.15) {$W_4^{\name}$:};
\end{tikzpicture}}
\ \\
\ \\
\scalebox{\scalefactor}{
\begin{tikzpicture}[node distance = 1.5cm,>=stealth',bend angle=0,auto]
	\node [place,tokens=1] (start) [label=below:$p_i$] {};
	\node [transition] (t1) [right of=start] {}
	edge [pre] (start);
	\node [place] (p1) [above right of=t1] {}
	edge [pre] (t1);
	\node [place] (p2) [below right of=t1] {}
	edge [pre] (t1);
	\node [transition] (t2) [right of=p1,yshift=0.5cm] {}
	edge [pre] (p1);
	\node [transition] (t3) [right of=p1,yshift=-0.5cm] {}
	edge [pre] (p1);
	\node [transition] (t4) [right of=p2,yshift=0.5cm] {}
	edge [pre] (p2);
	\node [transition] (t5) [right of=p2,yshift=-0.5cm] {}
	edge [pre] (p2);
	\node [place] (p3) [right of=t2,yshift=-0.5cm] {}
	edge [pre] (t2)
	edge [pre] (t3);
	\node [place] (p4) [right of=t4,yshift=-0.5cm] {}
	edge [pre] (t4)
	edge [pre] (t5);
	\node [transition] (t6) [below right of=p3] {}
	edge [pre] (p3)
	edge [pre] (p4);
	\node [place] (end) [right of=t6, label=below:$p_o$] {}
	edge [pre] (t6);
	\node at (0,1.15) {$W_5^{\name}$:};
	\draw[opacity=0] (-0.35,-1.75) rectangle (5.5,1.5);
\end{tikzpicture}}
\end{center}
\caption{Two workflow nets with different structure but the same complexity score $\C(W_4^{\name}) = \C(W_5^{\name}) = -\left(\frac{2}{6} \cdot \log_2(\frac{2}{6}) + \frac{4}{6} \cdot \log_2(\frac{4}{6})\right) \approx 0.9183$.}
\label{fig:ch-comp}
\end{figure}
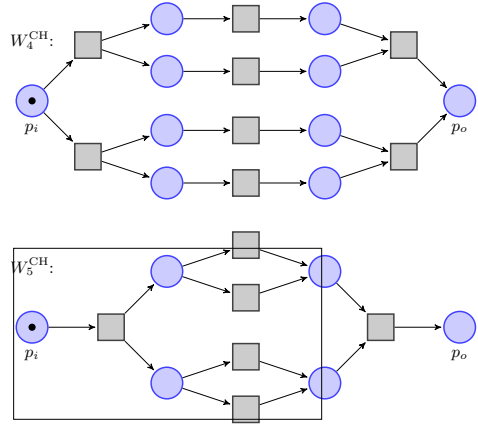
For their complexity scores, we have $\C(W_4^{\name}) = 0.9183 = \C(W_5^{\name})$, but for the operations of Definition~\ref{def:operations}, we get:
\begin{itemize}
\item $\C(W_4^{\name} \seqop W_3^{\name}) = -\left(\frac{4}{8} \cdot \log_2(\frac{4}{8}) + \frac{4}{8} \cdot \log_2(\frac{4}{8})\right) = 0$ and \\
$\C(W_5^{\name} \seqop W_3^{\name}) = -\left(\frac{2}{8} \cdot \log_2(\frac{2}{8}) + \frac{6}{8} \cdot \log_2(\frac{6}{8})\right) \approx 0.8113$,
\item $\C(W_4^{\name} \parop W_3^{\name}) = -\left(\frac{5}{9} \cdot \log_2(\frac{5}{9}) + \frac{4}{9} \cdot \log_2(\frac{4}{9})\right) \approx 0.9911$ and \\
$\C(W_5^{\name} \parop W_3^{\name}) = -\left(\frac{3}{9} \cdot \log_2(\frac{3}{9}) + \frac{6}{9} \cdot \log_2(\frac{6}{9})\right) \approx 0.9183$,
\item $\C(W_4^{\name} \choiceop W_3^{\name}) = -\left(\frac{4}{9} \cdot \log_2(\frac{4}{9}) + \frac{5}{9} \cdot \log_2(\frac{5}{9})\right) \approx 0.9911$ and \\
$\C(W_5^{\name} \choiceop W_3^{\name}) = -\left(\frac{2}{9} \cdot \log_2(\frac{2}{9}) + \frac{7}{9} \cdot \log_2(\frac{7}{9})\right) \approx 0.7642$,
\item $\C(W_4^{\name} \loopop W_3^{\name}) = -\left(\frac{4}{9} \cdot \log_2(\frac{4}{9}) + \frac{5}{9} \cdot \log_2(\frac{5}{9})\right) \approx 0.9911$ and \\
$\C(W_5^{\name} \loopop W_3^{\name}) = -\left(\frac{2}{9} \cdot \log_2(\frac{2}{9}) + \frac{7}{9} \cdot \log_2(\frac{7}{9})\right) \approx 0.7642$.
\end{itemize}

\propitemf{\propseven}{\yes} 
The workflow nets $W_6^{\name}$ and $W_7^{\name}$ of Figure~\ref{fig:ch-perm} are permutations of each other, but get different complexity scores:
\begin{align*}
\C(W_6^{\name}) &= -\left(\frac{2}{4} \cdot \log_2\left(\frac{2}{4}\right) + \frac{2}{4} \cdot \log_2\left(\frac{2}{4}\right)\right) = 1 \\
&\neq 0.9182 \approx -\left(\frac{4}{6} \cdot \log_2\left(\frac{4}{6}\right) + \frac{2}{6} \cdot \log_2\left(\frac{2}{6}\right)\right) = \C(W_7^{\name})
\end{align*}
\begin{figure}[ht]
\begin{center}
\begin{minipage}{0.55\textwidth}
\centering
\scalebox{\scalefactor}{
\begin{tikzpicture}[node distance = 1.5cm,>=stealth',bend angle=0,auto]
	\node [place,tokens=1] (start) [label=below:$p_i$] {};
	\node [transition] (t1) [right of=start,label=below:$t_1$] {}
	edge [pre] (start);
	\node [place] (p1) [above right of=t1,label=below:$p_1$] {}
	edge [pre] (t1);
	\node [place] (p2) [below right of=t1,label=below:$p_2$] {}
	edge [pre] (t1);
	\node [transition] (t2) [right of=p1,label=below:$t_2$] {}
	edge [pre] (p1);
	\node [transition] (t3) [right of=p2,label=below:$t_3$] {}
	edge [pre] (p2);
	\node [place] (p3) [right of=t2,label=below:$p_3$] {}
	edge [pre] (t2);
	\node [place] (p4) [right of=t3,label=below:$p_4$] {}
	edge [pre] (t3);
	\node [transition] (t4) [below right of=p3,label=below:$t_4$] {}
	edge [pre] (p3)
	edge [pre] (p4);
	\node [transition] (t5) [below of=t3,label=below:$t_5$] {}
	edge [pre,bend left=20] (start);
	\node [place] (end) [right of=t4,label=below:$p_o$] {}
	edge [pre] (t4)
	edge [pre,bend left=20] (t5);
	\node at (0,1.15) {$W_6^{\name}$:};
	\draw[opacity=0] ($(start)-(0.25,3.25)$) rectangle ($(end) + (0.25,1.5)$);
\end{tikzpicture}}
\end{minipage}
\begin{minipage}{0.4\textwidth}
\centering
\scalebox{\scalefactor}{
\begin{tikzpicture}[node distance = 1.5cm,>=stealth',bend angle=0,auto]
	\node [place,tokens=1] (start) [label=below:$p_i$] {};
	\node [transition] (t1) [above right of=start,yshift=0.7cm,label=below:$t_1$] {}
	edge [pre] (start);
	\node [transition] (t2) [below right of=start,yshift=-0.7cm,label=below:$t_2$] {}
	edge [pre] (start);
	\node [place] (p2) [right of=t1,yshift=-0.5cm,label=below:$p_2$] {}
	edge [pre] (t1);
	\node [place] (p1) at ($(p2) + (0,1)$) [label=below:$p_1$] {}
	edge [pre] (t1);
	\node [place] (p3) [right of=t2,yshift=0.5cm,label=below:$p_3$] {}
	edge [pre] (t2);
	\node [place] (p4) at ($(p3) - (0,1)$) [label=below:$p_4$] {}
	edge [pre] (t2);
	\node [transition] (t3) [right of=p2,yshift=0.5cm,label=below:$t_3$] {}
	edge [pre] (p1)
	edge [pre] (p2);
	\node [transition] (t4) [right of=p3,yshift=-0.5cm,label=below:$t_4$] {}
	edge [pre] (p3)
	edge [pre] (p4);
	\node [transition] (t5) at ($0.5*(p2) + 0.5*(p3)$) [label=below:$t_5$] {}
	edge [pre] (start);
	\node [place] (end) [below right of=t3,yshift=-0.7cm,label=below:$p_o$] {}
	edge [pre] (t3)
	edge [pre] (t4)
	edge [pre] (t5);
	\node at (0,1.15) {$W_7^{\name}$:};
	\draw[opacity=0] ($(start)-(0.25,3)$) rectangle ($(end) + (0.25,2.75)$);
\end{tikzpicture}}
\end{minipage}
\end{center}
\caption{Two workflow nets, $W_6^{\name}$ and $W_7^{\name}$, where $W_7^{\name} \in Perm(W_6^{\name})$, but $\C(W_6^{name}) \neq \C(W_7^{\name})$.}
\label{fig:ch-perm}
\end{figure}

\propitemf{\propeight}{\yes} 
$\C$ depends only on the amount of connectors of each type, which is independent of the labeling.

\propitemf{\propnine}{\yes} 
Take the workflow nets $W_4^{\name}$ and $W_5^{\name}$ of Figure~\ref{fig:ch-comp}. 
Both of these workflow nets have complexity $0$ according to $\C$, but combining them using any of the operations $\seqop, \parop, \choiceop, \loopop$ yields:
\begin{itemize}
\item $\C(W_4^{\name} \seqop W_5^{\name}) = -\left(\frac{6}{12} \cdot \log_2(\frac{6}{12}) + \frac{6}{12} \cdot \log_2(\frac{6}{12})\right)$\\
\phantom{$\C(W_4^{\name} \seqop W_5^{\name})$ }$= 1 > 0 + 0 = \C(W_4^{\name}) + \C(W_5^{\name})$,
\item $\C(W_4^{\name} \parop W_5^{\name}) = -\left(\frac{8}{14} \cdot \log_2(\frac{8}{14}) + \frac{6}{14} \cdot \log_2(\frac{6}{14})\right)$\\
\phantom{$\C(W_4^{\name} \parop W_5^{\name})$ }$\approx 0.9852 > 0 + 0 = \C(W_4^{\name}) + \C(W_5^{\name})$,
\item $\C(W_4^{\name} \choiceop W_5^{\name}) = -\left(\frac{6}{14} \cdot \log_2(\frac{6}{14}) + \frac{8}{14} \cdot \log_2(\frac{8}{14})\right)$\\
\phantom{$\C(W_4^{\name} \choiceop W_5^{\name})$ }$\approx 0.9852 > 0 + 0 = \C(W_4^{\name}) + \C(W_5^{\name})$,
\item $\C(W_4^{\name} \loopop W_5^{\name}) = -\left(\frac{6}{14} \cdot \log_2(\frac{6}{14}) + \frac{8}{14} \cdot \log_2(\frac{8}{14})\right)$\\
\phantom{$\C(W_4^{\name} \loopop W_5^{\name})$ }$\approx 0.9852 > 0 + 0 = \C(W_4^{\name}) + \C(W_5^{\name})$.
\end{itemize}

\propitemf{\propdef}{\no} 
The connector heterogeneity metric is undefined for workflow nets without connectors, since we would divide by $0$ in this case.
We can avoid this by defining that $\C$ should be $0$ in this special case, since a workflow net without any connectors is sequential and therefore easy to understand.
In the defined cases, the metric only returns non-negative values, since $\log_2(r) < 0$ for some $0 < r \leq 1$.
Adding two non-positive values results in a non-positive value, so negating this value results in a non-negative value.

\propitemf{\propmin}{\yes} 
The minimum possible complexity score of $\C$ is $0$.
This score is rewarded to a workflow net that contains only one connector type, like the net $W_2^{\name}$ in Figure~\ref{fig:ch-examples}, which contains only xor-connectors. 

\propitemf{\propinf}{\yes} 
Let $k, n \in \mathbb{N}$ with $n \geq 2$ be fixed. 
Figure~\ref{fig:ch-fin} shows how to construct a workflow net with $2$ xor-connectors (the places $p$ and $p_o$) and $2n$ and-connectors (the transitions $t_1^1, \dots t_n^1, t_1^4, \dots, t_n^4$).
We therefore get 
\begin{align*}
\C(W_{k,n}^{\name}) &= -\left(\frac{1}{n+1} \cdot \log_2\left(\frac{1}{n+1}\right) + \frac{n}{n+1} \cdot \log_2\left(\frac{n}{n+1}\right)\right) \\
&> -\left(\frac{1}{n+2} \cdot \log_2\left(\frac{1}{n+2}\right) + \frac{n+1}{n+2} \cdot \log_2\left(\frac{n+1}{n+2}\right)\right) \\
&= \C(W_{k,n+1}^{\name})
\end{align*}
Thus, $\C(W_{k,n}^{\name}) \neq \C(W_{k,m}^{\name})$ for any $n \neq m \in \mathbb{N}$ with $n,m \geq 2$, so Figure~\ref{fig:ch-fin} shows how to construct infinitely many workflow nets of different complexity according to $\C$.
In other words, we get the following result:
$|\{c \in \mathbb{R} \mid \exists M \in \mathcal{M}: \C(M) = c\}| \geq |\{n \in \mathbb{N} \mid n \geq 2\}| = \infty$.

\propitemf{\propnotsup}{\yes} The counter examples found for property {\propfive} also show that $\C$ is not superadditive.

\propitemf{\propadd}{\no} Since $\C$ is neither subadditive nor superadditive, it can't be additive. 
\end{description}

\def\name{\crossconnname}
\def\C{\crossconn}
\subsubsection{Cross-Connectivity}
Vanderfeesten et al.~\cite{VanRMAC08} introduced the cross-connectivity metric to measure how difficult it is to understand the connection between two nodes in a model.
This measure is based on paths between all pairs of nodes.
Let $N = (P, T, F, p_i, p_o)$ be a workflow net (note that we choose $N$ as the identifier here to avoid confusion with the weight-function needed for this metric).
First, each node $v \in P \cup T$ of the net gets a weight depending on its type and degree $\ndeg(v) = |\pre{n}| + |\post{n}|$:
\begin{equation}
w_N(v) := \begin{cases}
\frac{1}{\ndeg(v)} & \text{ if } v \in \mathcal{C}_{\text{xor}}^N \\
1 & \text{ if } v \in \mathcal{C}_{\text{and}}^N \\
1 & \text{ otherwise.}
\end{cases}
\end{equation}
For an edge $(u,v) \in F$, we define its weight as the product of its endpoints, so $w_N((u,v)) := w_N(u) \cdot w_N(v)$.
Let $\rho = v_1, \dots, v_k$ be a path of length $k \geq 2$. 
The weight of the path $\rho$ is defined as the product of the edges it uses, so $w_N(\rho) = w_N((v_1, v_2)) \cdot \ldots \cdot w_N((v_{k-1},v_k))$.
Let $\mathcal{P}_{v_1, v_2}$ be the set of all paths from $v_1$ to $v_2$ of length at least $2$.
We define the value of a connection between two nodes $v_1$ and $v_2$ as
\begin{equation}
V_N(v_1, v_2) := \max(\{w_N(\rho) \mid \rho \in \mathcal{P}_{v_1, v_2}\} \cup \{0\}).
\end{equation}
We use the special value $0$ for cases where no path from $v_1$ to $v_2$ exist.
The original work of Vanderfeesten et al.~\cite{VanRMAC08} then defines the cross-connectivity metric as the sum of all values between nodes divided by the number of nodes times the number of nodes minus one.
We directly translate this definition to workflow nets, but subtract the result from one, so a low complexity score refers to models with low complexity.
\begin{equation}
\label{eq:c-cross-connectivity-original}
\C(N) = 1 - \frac{\sum_{v_1, v_2 \in P \cup T} V_N(v_1, v_2)}{(|P| + |T|) \cdot (|P| + |T| - 1)}.
\end{equation}
Figure~\ref{fig:cc-mon} shows a composed workflow net and its complexity score, as well as the scores for the input models.
\begin{figure}[ht]
\begin{center}
\scalebox{\scalefactor}{
\begin{tikzpicture}[node distance = 1.5cm,>=stealth',bend angle=0,auto]
	\node [place,tokens=1] (start) [label=above:$p_i$,label=below:\textcolor{purple}{$1$}] {};
	\node [transition] (t1) [right of=start,label=above:$t_1$,label=below:\textcolor{purple}{$1$}] {}
	edge [pre] (start);
	\node [place] (p1) [above right of=t1,yshift=0.5cm,label=above:$p_1$,label=below:\textcolor{purple}{$\frac{1}{3}$}] {}
	edge [pre] (t1);
	\node [place] (p2) [below right of=t1,yshift=-0.5cm,label=above:$p_2$,label=below:\textcolor{purple}{$1$}] {}
	edge [pre] (t1);
	\node [transition] (t2) [above right of=p1,label=above:$t_2$,label=below:\textcolor{purple}{$1$}] {}
	edge [pre] (p1);
	\node [transition] (t3) [below right of=p1,label=above:$t_3$,label=below:\textcolor{purple}{$1$}] {}
	edge [pre] (p1);
	\node (dummy) [above right of=p2] {};
	\node [place] (p3) [below right of=t2,label=above:$p_3$,label=below:\textcolor{purple}{$\frac{1}{3}$}] {}
	edge [pre] (t2)
	edge [pre] (t3);
	\node [place] (p4) [below right of=dummy,label=above:$p_4$,label=below:\textcolor{purple}{$1$}] {};
	\node [transition] (t4) at ($0.5*(p2) + 0.5*(p4)$) [label=above:$t_4$,label=below:\textcolor{purple}{$1$}] {}
	edge [pre] (p2)
	edge [post] (p4);
	\node [transition] (t5) [below right of=p3,label=above:$t_5$,yshift=-0.5cm,label=below:\textcolor{purple}{$1$}] {}
	edge [pre] (p3)
	edge [pre] (p4);
	\node [place] (end) [right of=t5,label=above:$p_o$,label=below:\textcolor{purple}{$1$}] {}
	edge [pre] (t5);
	
	\node at (7.5,2) {
	\begin{tabular}{|c|c|c|c|c|} \hline
	 $V_{W_1^{\name}}$ & $p_1$ & $t_2$ & $t_3$ & $p_3$ \\ \hline
	 $p_1$ & $0$ & $\sfrac{1}{3}$ & $\sfrac{1}{3}$ & $\sfrac{1}{9}$ \\ \hline
	 $t_2$ & $0$ & $0$ & $0$ & $\sfrac{1}{3}$ \\ \hline
	 $t_3$ & $0$ & $0$ & $0$ & $\sfrac{1}{3}$ \\ \hline
	 $p_3$ & $0$ & $0$ & $0$ & $0$ \\ \hline
	\end{tabular}
	};
	
	\node at (10.5,2) {
	\begin{tabular}{|c|c|c|c|} \hline
	 $V_{W_2^{\name}}$ & $p_2$ & $t_4$ & $p_4$ \\ \hline
	 $p_2$ & $0$ & $1$ & $1$ \\ \hline
	 $t_4$ & $0$ & $0$ & $1$ \\ \hline
	 $p_4$ & $0$ & $0$ & $0$ \\ \hline
	\end{tabular}
	};
	
	\node at (9,-3.5) {
	\begin{tabular}{|c|c|c|c|c|c|c|c|c|c|c|c|c|} \hline
	 $V_{W_{\parop}^{\name}}$ & $p_i$ & $t_1$ & $p_1$ & $p_2$ & $t_2$ & $t_3$ & $t_4$ & $p_3$ & $p_4$ & $t_6$ & $p_o$ \\ \hline
	 $p_i$ & $0$ & $1$ & $\sfrac{1}{3}$ & $1$ & $\sfrac{1}{3}$ & $\sfrac{1}{3}$ & $1$ & $\sfrac{1}{9}$ & $1$ & $1$ & $1$ \\ \hline
	 $t_1$ & $0$ & $0$ &$\sfrac{1}{3}$ & $1$ & $\sfrac{1}{3}$ & $\sfrac{1}{3}$ & $1$ & $\sfrac{1}{9}$ & $1$ & $1$ & $1$ \\ \hline
	 $p_1$ & $0$ & $0$ & $0$ & $0$ & $\sfrac{1}{3}$ & $\sfrac{1}{3}$ & $0$ & $\sfrac{1}{9}$ & $0$ & $\sfrac{1}{9}$ & $\sfrac{1}{9}$ \\ \hline
	 $p_2$ & $0$ & $0$ & $0$ & $0$ & $0$ & $0$ & $1$ & $0$ & $1$ & $1$ & $1$ \\ \hline
	 $t_2$ & $0$ & $0$ & $0$ & $0$ & $0$ & $0$ & $0$ & $\sfrac{1}{3}$ & $0$ & $\sfrac{1}{3}$ & $\sfrac{1}{3}$ \\ \hline
	 $t_3$ & $0$ & $0$ & $0$ & $0$ & $0$ & $0$ & $0$ & $\sfrac{1}{3}$ & $0$ & $\sfrac{1}{3}$ & $\sfrac{1}{3}$ \\ \hline
	 $t_4$ & $0$ & $0$ & $0$ & $0$ & $0$ & $0$ & $0$ & $0$ & $1$ & $1$ & $1$ \\ \hline
	 $p_3$ & $0$ & $0$ & $0$ & $0$ & $0$ & $0$ & $0$ & $0$ & $0$ & $\sfrac{1}{3}$ & $\sfrac{1}{3}$ \\ \hline
	 $p_4$ & $0$ & $0$ & $0$ & $0$ & $0$ & $0$ & $0$ & $0$ & $0$ & $1$ & $1$ \\ \hline
	 $t_5$ & $0$ & $0$ & $0$ & $0$ & $0$ & $0$ & $0$ & $0$ & $0$ & $0$ & $1$ \\ \hline
	 $p_o$ & $0$ & $0$ & $0$ & $0$ & $0$ & $0$ & $0$ & $0$ & $0$ & $0$ & $0$ \\ \hline
	\end{tabular}
	};
	
	\node at (0,2) {$W_{\parop}^{\name}$:};
	
	\begin{pgfonlayer}{bg}
		\draw[rounded corners,draw=gray,fill=lightgray!50!white] ($(p1) - (0.75, 1.75)$) rectangle ($(p3) + (0.75, 1.75)$);
		\node[gray] at ($(p1) + (-0.4,1.25)$) {$W_1^{\name}$};
		\draw[rounded corners,draw=gray,fill=lightgray!50!white] ($(p2) - (0.75, 1)$) rectangle ($(p4) + (0.75, 1)$);
		\node[gray] at ($(p2) + (-0.4,-0.75)$) {$W_2^{\name}$};
	\end{pgfonlayer}
\end{tikzpicture}}
\end{center}
\caption{A workflow net $W_{\parop}^{\name}$ and its $V_{W_{\parop}^{\name}}$-values for all pairs of vertices. We get $\C(W_{\parop}^{\name}) = 1 - \frac{21 + 16 \cdot \sfrac{1}{3} + 5 \cdot \sfrac{1}{9}}{11 \cdot 10} = 0.7\overline{5}$, $\C(W_1^{\name}) = 1 - \frac{4 \cdot \sfrac{1}{3} + \sfrac{1}{9}}{4 \cdot 3} \approx 0.8796$ and $\C(W_2^{\name}) = 1 - \frac{3}{6} = 0.5$.}
\label{fig:cc-mon}
\end{figure}

\begin{description}
\propitemf{\propone}{\yes} 
For the two nets $W_1^{\name}$ and $W_2^{\name}$ of Figure~\ref{fig:cc-mon}, we get the complexity scores $\C(W_1^{\name}) \approx 0.8796 \neq 0.5 = \C(W_2^{\name})$.

\propitemf{\proptwo}{\no} 
Figure~\ref{fig:cc-fin} shows how to construct infinitely many workflow nets with complexity score $\frac{1}{2}$. 
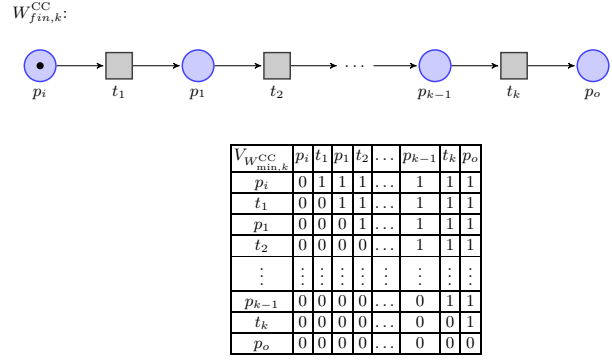
\begin{figure}[ht]
\begin{center}
\scalebox{\scalefactor}{
\begin{tikzpicture}[node distance = 1.5cm,>=stealth',bend angle=0,auto]
	\node [place,tokens=1] (start) [label=below:$p_i$] {};
	\node [transition] (t1) [right of=start,label=below:$t_1$] {}
	edge [pre] (start);
	\node [place] (p1) [right of=t1,label=below:$p_1$] {}
	edge [pre] (t1);
	\node [transition] (t2) [right of=p1,label=below:$t_2$] {}
	edge [pre] (p1);
	\node (dots) [right of=t2] {$\dots$}
	edge [pre] (t2);
	\node [place] (p2) [right of=dots,label=below:$p_{k-1}$] {}
	edge [pre] (dots);
	\node [transition] (t4) [right of=p2,label=below:$t_k$] {}
	edge [pre] (p2);
	\node [place] (end) [right of=t4,label=below:$p_o$] {}
	edge [pre] (t4);
	\node at (0,1) {$W_{fin,k}^{\name}$:};
	
	\node [below of=dots,yshift=-2cm] {
	\begin{tabular}{|c|c|c|c|c|c|c|c|c|} \hline
	 $V_{W_{\min,k}^{\name}}$ & $p_i$ & $t_1$ & $p_1$ & $t_2$ & $\dots$ & $p_{k-1}$ & $t_k$ & $p_o$ \\ \hline
	 $p_i$ & $0$ & $1$ & $1$ & $1$ & $\dots$ & $1$ & $1$ & $1$ \\ \hline
	 $t_1$ & $0$ & $0$ & $1$ & $1$ & $\dots$ & $1$ & $1$ & $1$ \\ \hline
	 $p_1$ & $0$ & $0$ & $0$ & $1$ & $\dots$ & $1$ & $1$ & $1$ \\ \hline
	 $t_2$ & $0$ & $0$ & $0$ & $0$ & $\dots$ & $1$ & $1$ & $1$ \\ \hline
	 $\vdots$ & $\vdots$ & $\vdots$ & $\vdots$ & $\vdots$ & $\vdots$ & $\vdots$ & $\vdots$ & $\vdots$ \\ \hline
	 $p_{k-1}$ & $0$ & $0$ & $0$ & $0$ & $\dots$ & $0$ & $1$ & $1$ \\ \hline
	 $t_k$ & $0$ & $0$ & $0$ & $0$ & $\dots$ & $0$ & $0$ & $1$ \\ \hline
	 $p_o$ & $0$ & $0$ & $0$ & $0$ & $\dots$ & $0$ & $0$ & $0$ \\ \hline
	\end{tabular}
	};
\end{tikzpicture}}
\end{center}
\caption{A workflow net $W_{fin,k}^{\name}$ with $k$ transitions and $k+1$ places. Its complexity is $\C(W_{fin,k}^{\name}) = 1 - \frac{k \cdot (2k+1)}{(2k+1) \cdot 2k} = \frac{1}{2}$ regardless of the choice of $k$.}
\label{fig:cc-fin}
\end{figure}

\propitemf{\propthree}{\yes} 
The workflow nets $W_{fin,1}^{\name}$ and $W_{fin,2}^{\name}$ of Figure~\ref{fig:cc-fin} are structurally different but get the same complexity scores.

\propitemf{\propfour}{\yes} The workflow nets $W_3^{\name}$ and $W_4^{\name}$ of Figure~\ref{fig:cc-indL} have the same language but receive different complexity scores.
\begin{figure}[ht]
\begin{center}
\begin{minipage}{0.2525\textwidth}
\centering
\scalebox{\scalefactor}{
\begin{tikzpicture}[node distance = 1.5cm,>=stealth',bend angle=0,auto]
	\node [place,tokens=1] (start) [label=below:$p_i$] {};
	\node [transition] (t1) [right of=start,label=center:$a$] {}
	edge [pre] (start);
	\node [place] (p1) [right of=t1,label=below:$p_o$] {}
	edge [pre] (t1);
	\node at (0,1.15) {$W_3^{\name}$:};
	\draw[opacity=0] ($(start)-(0.25,1.75)$) rectangle ($(p1) + (0.25,1.75)$);
\end{tikzpicture}}
\end{minipage}
\begin{minipage}{0.25\textwidth}
\centering
\scalebox{\scalefactor}{
\begin{tikzpicture}[node distance = 1.5cm,>=stealth',bend angle=0,auto]
	\node [place,tokens=1] (start) [label=below:$p_i$] {};
	\node [transition] (t1) [above right of=start,label=center:$a$] {}
	edge [pre] (start);
	\node [transition] (t2) [below right of=start,label=center:$a$] {}
	edge [pre] (start);
	\node [place] (p1) [below right of=t1,label=below:$p_o$] {}
	edge [pre] (t1)
	edge [pre] (t2);
	\node at (0,1.15) {$W_4^{\name}$:};
	\draw[opacity=0] ($(start)-(0.25,1.75)$) rectangle ($(p1) + (0.25,1.75)$);
\end{tikzpicture}}
\end{minipage}
\end{center}
\caption{Two workflow nets, $W_3^{\name}$ and $W_4^{\name}$ with $L(W_3^{\name}) = L(W_4^{\name}) = \{\varepsilon, a\}$, but $\C(W_3^{\name}) = 0$ and $\C(W_4^{\name}) = 1 - \frac{4 \cdot 0.5 + 0.25}{12} = \frac{13}{16}$.}
\label{fig:cc-indL}
\end{figure}
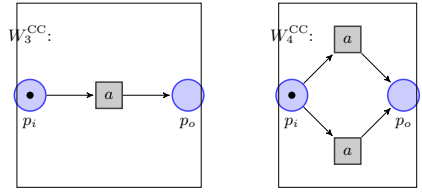

\propitemf{\propfive}{\no} 
$\C$ is not monotone for any of the operations $\oplus \in \{\seqop, \parop, \choiceop, \loopop\}$. 
For the operator $\parop$, Figure~\ref{fig:cc-mon} shows a detailed counter-example for monotonicity.
We can take the workflow nets $W_1^{\name}$ and $W_2^{\name}$ of this figure to construct counter-examples for the other operations:
\begin{itemize}
\item $\C(W_1^{\name} \seqop W_2^{\name}) = 1 - \frac{6 \cdot 1 + 16 \cdot \sfrac{1}{3} + 5 \cdot \sfrac{1}{9}}{8 \cdot 7} \approx 0.7877$ \\
\phantom{$\C(W_1^{\name} \seqop W_2^{\name})$ }$< 0.8796 = \C(W_1^{\name})$,
\item $\C(W_1^{\name} \parop W_2^{\name}) = 1 - \frac{21 \cdot 1 + 16 \cdot \sfrac{1}{3} + 5 \cdot \sfrac{1}{9}}{11 \cdot 10} = 0.7\overline{5}$ \\
\phantom{$\C(W_1^{\name} \parop W_2^{\name})$ }$< 0.8796 = \C(W_1^{\name})$,
\item $\C(W_1^{\name} \choiceop W_2^{\name}) = 1 - \frac{10 \cdot 1 + 11 \cdot \sfrac{1}{2} + 11 \cdot \sfrac{1}{3} + \sfrac{1}{4} + 6 \cdot \sfrac{1}{6} + 4 \cdot \sfrac{1}{9} + 4 \cdot \sfrac{1}{18}}{13 \cdot 12} \approx 0.8648$ \\
\phantom{$\C(W_1^{\name} \choiceop W_2^{\name})$ }$< 0.8796 = \C(W_1^{\name})$,
\item $\C(W_1^{\name} \loopop W_2^{\name}) = 1 - \frac{17 + 37 \cdot \sfrac{1}{3} + 56 \cdot \sfrac{1}{9} + 43 \cdot \sfrac{1}{27} + 74 \cdot \sfrac{1}{81}}{17 \cdot 17} \approx 0.8602$ \\
\phantom{$\C(W_1^{\name} \loopop W_2^{\name})$ }$< 0.8796 = \C(W_1^{\name})$.
\end{itemize}

\propitemf{\propsix}{\yes} Take the workflow nets $W_{fin,1}^{\name}$ and $W_{fin,2}^{\name}$ of Figure~\ref{fig:cc-fin}. 
Further, take the workflow nets $W_5^{\name}$ and $W_6^{\name}$ of Figure~\ref{fig:cc-not-tli} and the net $W_1^{\name}$ of Figure~\ref{fig:cc-mon}.
For the complexity scores, we have $W_{fin,1}^{\name} = W_{fin,2}^{\name} = \frac{1}{2}$ and $W_5^{\name} = W_6^{\name} = 1 - \frac{1 \cdot 1 + 2 \cdot \sfrac{1}{2} + 8 \cdot \sfrac{1}{3} + 3 \cdot \sfrac{1}{6}}{6 \cdot 5} = \frac{151}{180} \approx 0.8389$ and get:
\begin{itemize}
\item $\C(W_5^{\name} \seqop W_1^{\name}) = 1 - \frac{1 \cdot 1 + 2 \cdot \sfrac{1}{2} + 18 \cdot \sfrac{1}{3} + 2 \cdot \sfrac{1}{6} + 15 \cdot \sfrac{1}{9} + 9 \cdot \sfrac{1}{27} + 3 \cdot \sfrac{1}{54}}{11 \cdot 10} = \frac{163}{180}$ \\
$\approx 0.9056 \neq 0.8634 \approx \frac{3419}{3960} = 1 - \frac{3 \cdot 1 + 4 \cdot \sfrac{1}{2} + 22 \cdot \sfrac{1}{3} + 8 \cdot \sfrac{1}{6} + 9 \cdot \sfrac{1}{9} + 6 \cdot \sfrac{1}{18} + 1 \cdot \sfrac{1}{36}}{11 \cdot 10}$ \\
$= \C(W_6^{\name} \seqop W_1^{\name})$,
\item $\C(W_{fin,1} \parop W_{fin,1}) = 1 - \frac{36 \cdot 1}{10 \cdot 9} = \frac{3}{5} = 0.6$ \\
$\neq 0.6136 \approx \frac{27}{44} = 1 - \frac{51 \cdot 1}{12 \cdot 11} = \C(W_{fin,2} \parop W_{fin,1})$,
\item $\C(W_{fin,1} \choiceop W_{fin,1}) = 1 - \frac{20 \cdot 1 + 20 \cdot \sfrac{1}{2} + 1 \cdot \sfrac{1}{4}}{12 \cdot 11} = \frac{37}{48} \approx 0.7708$ \\
$\neq 0.7624 \approx \frac{555}{728} = 1 - \frac{31 \cdot 1 + 24 \cdot \sfrac{1}{2} + 1 \cdot \sfrac{1}{4}}{14 \cdot 13} = \C(W_{fin,2} \choiceop W_{fin,1})$,
\item $\C(W_{fin,1} \loopop W_{fin,1}) = 1 - \frac{22 \cdot 1 + 94 \cdot \sfrac{1}{3} + 80 \cdot \sfrac{1}{9} + 2 \cdot \sfrac{1}{27}}{16 \cdot 15} = \frac{2399}{3240} \approx 0.7404$ \\
$\neq 0.7194 \approx \frac{2972}{4131} = 1 - \frac{33 \cdot 1 + 126 \cdot \sfrac{1}{3} + 97 \cdot \sfrac{1}{9} + 2 \cdot \sfrac{1}{27}}{18 \cdot 17} = \C(W_{fin,2} \loopop W_{fin,1})$.
\end{itemize}
\begin{figure}[ht]
\begin{center}
\begin{minipage}{0.45\textwidth}
\centering
\scalebox{\scalefactor}{
\begin{tikzpicture}[node distance = 1.5cm,>=stealth',bend angle=0,auto]
	\node [place,tokens=1] (start) [label=below:$p_i$] {};
	\node [transition] (t1) [right of=start] {}
	edge [pre] (start);
	\node [place] (p1) [right of=t1] {}
	edge [pre] (t1);
	\node [transition] (t2) [above right of=p1] {}
	edge [pre] (p1);
	\node [transition] (t3) [below right of=p1] {}
	edge [pre] (p1);
	\node [place] (end) [below right of=t2,label=below:$p_o$] {}
	edge [pre] (t2)
	edge [pre] (t3);
	\node at (0,1.15) {$W_5^{\name}$:};
\end{tikzpicture}}
\end{minipage}
\begin{minipage}{0.45\textwidth}
\centering
\scalebox{\scalefactor}{
\begin{tikzpicture}[node distance = 1.5cm,>=stealth',bend angle=0,auto]
	\node[place,tokens=1] (start) [label=below:$p_i$] {};
	\node [transition] (t1) [above right of=start] {}
	edge [pre] (start);
	\node [transition] (t2) [below right of=start] {}
	edge [pre] (start);
	\node [place] (p1) [below right of=t1] {}
	edge [pre] (t1)
	edge [pre] (t2);
	\node [transition] (t3) [right of=p1] {}
	edge [pre] (p1);
	\node [place] (end) [right of=t3,label=below:$p_o$] {}
	edge [pre] (t3);
	\node at (0,1.15) {$W_6^{\name}$:};
\end{tikzpicture}}
\end{minipage}
\end{center}
\caption{Two workflow nets, $W_5^{\name}$ and $W_6^{\name}$, with the same complexity according to the measure $\C$.}
\label{fig:cc-not-tli}
\end{figure}
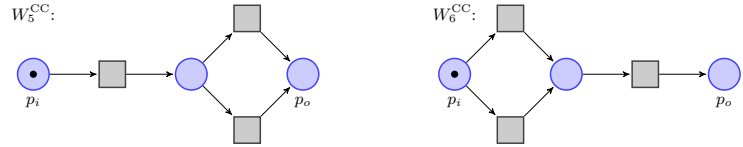

\propitemf{\propseven}{\yes} 
The workflow nets $W_7^{\name}$ and $W_8^{\name}$ in Figure~\ref{fig:cc-perm} are permutations of each other but $\C(W_7^{\name}) = 0 \neq \frac{53}{63} = \C(W_8^{\name})$.
\begin{figure}[ht]
\begin{center}
\centering
\scalebox{\scalefactor}{
\begin{tikzpicture}[node distance = 1.5cm,>=stealth',bend angle=0,auto]
	\node [place,tokens=1] (start) [label=below:$p_i$,label=above:\textcolor{purple}{$1$}] {};
	\node [transition] (t1) [right of=start,label=below:$t_1$,label=above:\textcolor{purple}{$1$}] {}
	edge [pre] (start);
	\node [place] (p1) [right of=t1,label=below:$p_1$,label=above:\textcolor{purple}{$1$}] {}
	edge [pre] (t1);
	\node [transition] (t2) [right of=p1,label=below:$t_2$,label=above:\textcolor{purple}{$1$}] {}
	edge [pre] (p1);
	\node [place] (p2) [right of=t2,label=below:$p_2$,label=above:\textcolor{purple}{$1$}] {}
	edge [pre] (t2);
	\node [transition] (t3) [right of=p2,label=below:$t_3$,label=above:\textcolor{purple}{$1$}] {}
	edge [pre] (p2);
	\node [place] (end) [right of=t3,label=below:$p_o$,label=above:\textcolor{purple}{$1$}] {}
	edge [pre] (t3);
	\node at (0,1.15) {$W_7^{\name}$:};
\end{tikzpicture}}
\ \\
\ \\
\centering
\scalebox{\scalefactor}{
\begin{tikzpicture}[node distance = 1.5cm,>=stealth',bend angle=0,auto]
	\node [place,tokens=1] (start) [label=below:$p_i$,label=above:\textcolor{purple}{$\frac{1}{2}$}] {};
	\node [transition] (t1) [above right of=start,label=below:$t_1$,label=above:\textcolor{purple}{$1$}] {}
	edge [pre] (start);
	\node [transition] (t2) [below right of=start,label=below:$t_2$,label=above:\textcolor{purple}{$1$}] {}
	edge [pre] (start);
	\node [place] (p1) [right of=t1,label=below:$p_1$,label=above:\textcolor{purple}{$\frac{1}{3}$}] {}
	edge [pre] (t1)
	edge [pre] (t2);
	\node [place] (p2) [right of=t2,label=below:$p_2$,label=above:\textcolor{purple}{$\frac{1}{3}$}] {}
	edge [pre] (t1)
	edge [pre] (t2);
	\node [transition] (t3) [below right of=p1,label=below:$t_3$,label=above:\textcolor{purple}{$1$}] {}
	edge [pre] (p1)
	edge [pre] (p2);
	\node [place] (end) [right of=t3,label=below:$p_o$,label=above:\textcolor{purple}{$1$}] {}
	edge [pre] (t3);
	\node at (0,1.15) {$W_8^{\name}$:};
\end{tikzpicture}}
\end{center}
\caption{Two workflow nets, $W_7^{\name}$ and $W_8^{\name}$ with $W_7^{\name} \in Perm(W_8^{\name})$, but $\C(W_7^{\name}) = 0$ and $\C(W_8^{\name}) = 1 - \frac{2 \cdot \sfrac{1}{2} + 12 \cdot \sfrac{1}{3} + 4 \cdot \sfrac{1}{6} + 1}{42} = \frac{53}{63}$.}
\label{fig:cc-perm}
\end{figure}
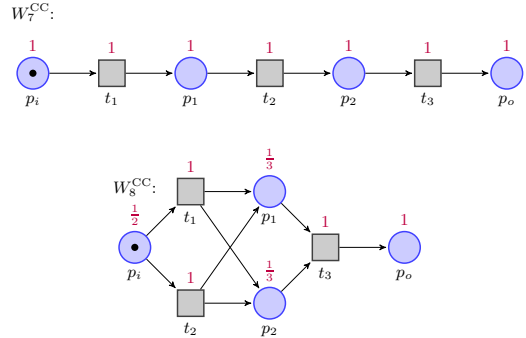

\propitemf{\propeight}{\yes} 
The cross-connectivity metric is independent of the transition labels, so uniformly changing them has no effect on the complexity score.

\propitemf{\propnine}{\yes} 
Consider the workflow net $W_{\min,1}^{\name}$ of Figure~\ref{fig:cc-min}. 
We get:
\begin{itemize}
\item $\C(W_{\min,1}^{\name} \seqop W_{\min,1}^{\name}) = 1 - \frac{98}{13 \cdot 12} = \frac{27}{78} \approx 0.3718$ \\
\phantom{$\C(W_{\min,1}^{\name} \seqop W_{\min,1}^{\name})$ }$> 0.\overline{3} = \frac{1}{6} + \frac{1}{6} = \C(W_{\min,1}^{\name}) + \C(W_{\min,1}^{\name})$,
\item $\C(W_{\min,1}^{\name} \parop W_{\min,1}^{\name}) = 1 - \frac{104}{16 \cdot 15} = \frac{17}{30} = 0.5\overline{6}$ \\
\phantom{$\C(W_{\min,1}^{\name} \parop W_{\min,1}^{\name})$ }$> 0.\overline{3} = \frac{1}{6} + \frac{1}{6} = \C(W_{\min,1}^{\name}) + \C(W_{\min,1}^{\name})$,
\item $\C(W_{\min,1}^{\name} \choiceop W_{\min,1}^{\name}) = 1 - \frac{92.24}{18 \cdot 17} = \frac{95}{136} \approx 0.6985$ \\
\phantom{$\C(W_{\min,1}^{\name} \choiceop W_{\min,1}^{\name})$ }$> 0.\overline{3} = \frac{1}{6} + \frac{1}{6} = \C(W_{\min,1}^{\name}) + \C(W_{\min,1}^{\name})$,
\item $\C(W_{\min,1}^{\name} \loopop W_{\min,1}^{\name}) = 1 - \frac{157.\overline{407}}{22 \cdot 21} \approx 0.6593$ \\
\phantom{$\C(W_{\min,1}^{\name} \loopop W_{\min,1}^{\name})$ }$> 0.\overline{3} = \frac{1}{6} + \frac{1}{6} = \C(W_{\min,1}^{\name}) + \C(W_{\min,1}^{\name})$.
\end{itemize}
\begin{figure}[ht]
\begin{center}
\scalebox{\scalefactor}{
\begin{tikzpicture}[node distance = 1.5cm,>=stealth',bend angle=0,auto]
	\node [place,tokens=1] (start) [label=below:$p_i$] {};
	\node [transition] (t1) [right of=start,label=below:$t_1$] {}
	edge [pre] (start);
	\node [place] (p1) [below right of=t1,label=below:$p_1$] {}
	edge [pre] (t1);
	\node [transition] (t2) [right of=p1,label=below:$t_2$] {}
	edge [pre] (p1);
	\node (dots) [right of=t2] {$\dots$}
	edge [pre] (t2);
	\node [transition] (t3) [right of=dots,label=below:$t_k$] {}
	edge [pre] (dots);
	\node [place] (p2) [right of=t3,label=below:$p_k$] {}
	edge [pre] (t3);
	\node [transition] (t4) [above right of=p2,label=below:$t_{k+1}$] {}
	edge [pre] (p2);
	\node [place] (end) [right of=t4,label=below:$p_o$] {}
	edge [pre] (t4);
	\node (dummy1) [above right of=t1] {};
	\node (dummy2) [right of=dummy1] {};
	\node [place] (p3) [right of=dummy2,label=below:$p_{k+1}$] {}
	edge [pre] (t4)
	edge [post] (t1);
	
	\node at (0,1) {$W_{\min,k}^{\name}$:};
	
	\node [below of=dots,yshift=-2cm] {
	\begin{tabular}{|c|c|c|c|c|c|c|c|c|c|c|} \hline
	 $V_{W_{\min,k}^{\name}}$ & $p_i$ & $t_1$ & $p_1$ & $t_2$ & $\dots$ & $t_k$ & $p_k$ & $t_{k+1}$ & $p_{k+1}$ & $p_o$ \\ \hline
	 $p_i$ & $0$ & $1$ & $1$ & $1$ & $\dots$ & $1$ & $1$ & $1$ & $1$ & $1$ \\ \hline
	 $t_1$ & $0$ & $1$ & $1$ & $1$ & $\dots$ & $1$ & $1$ & $1$ & $1$ & $1$ \\ \hline
	 $p_2$ & $0$ & $1$ & $1$ & $1$ & $\dots$ & $1$ & $1$ & $1$ & $1$ & $1$ \\ \hline
	 $t_2$ & $0$ & $1$ & $1$ & $1$ & $\dots$ & $1$ & $1$ & $1$ & $1$ & $1$ \\ \hline
	 $\vdots$ & $\vdots$ & $\vdots$ & $\vdots$ & $\vdots$ & $\dots$ & $\vdots$ & $\vdots$ & $\vdots$ & $\vdots$ & $\vdots$ \\ \hline
	 $t_k$ & $0$ & $1$ & $1$ & $1$ & $\dots$ & $1$ & $1$ & $1$ & $1$ & $1$ \\ \hline
	 $p_k$ & $0$ & $1$ & $1$ & $1$ & $\dots$ & $1$ & $1$ & $1$ & $1$ & $1$ \\ \hline
	 $t_{k+1}$ & $0$ & $1$ & $1$ & $1$ & $\dots$ & $1$ & $1$ & $1$ & $1$ & $1$ \\ \hline
	 $p_{k+1}$ & $0$ & $1$ & $1$ & $1$ & $\dots$ & $1$ & $1$ & $1$ & $1$ & $1$ \\ \hline
	 $p_o$ & $0$ & $0$ & $0$ & $0$ & $\dots$ & $0$ & $0$ & $0$ & $0$ & $0$ \\ \hline
	\end{tabular}
	};
\end{tikzpicture}}
\end{center}
\caption{A workflow net with $V_{W_{\min,k}^{\name}}(v_1, v_2) = 1$ for all $v_1 \in (P \cup T) \setminus \{p_o\}$ and $v_2 \in (P \cup T) \setminus \{p_i\}$. We get $\C(W_{\min,k}^{\name}) = 1 - \frac{(k+3) + (k+1) - 1}{(k+3) + (k+1)} = \frac{1}{2k + 4}$.}
\label{fig:cc-min}
\end{figure}
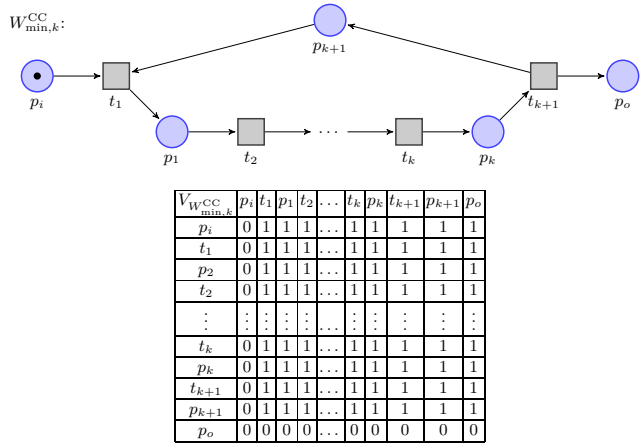

\propitemf{\propdef}{\yes} 
The cross-connectivity metric is always defined, since all workflow nets have at least $2$ places and $1$ transition and we therefore never divide by $0$. 
Also, we have $\C(N) \geq 0$ for all workflow nets $M \in \mathcal{M}$, as shown by the following Theorem:
\begin{theorem}
\label{theo:cc-bounded}
Let $M = (P, T, F, p_i, p_o) \in \mathcal{M}$ be a workflow net. Then,
\[\sum_{v_1, v_2 \in P \cup T} V_M(v_1, v_2) \leq (|P| + |T|) \cdot (|P| + |T| - 1).\]
\end{theorem}
\begin{proof}
By definition of $V_M$, for any nodes $v_1, v_2 \in P \cup T$, $V_M(v_1, v_2) \le 1$.
Therefore, if $V_M(v_1, v_2) = 1$ would be true for all $v_1, v_2 \in P \cup T$, we get 
\[\sum_{v_1, v_2 \in P \cup T} V_M(v_1, v_2) \leq (|P| + |T|) \cdot (|P| + |T|).\]
However, it is not possible for all $V_M$-values to be $1$: For any node $v \in P \cup T$, there is no path from $v$ to $p_i$, since $|\pre{p_i}| = 0$, we get $V_M(v,p_i) = 0$.
Furthermore, for any node $v \in P \cup T$ there is no path from $p_o$ to $v$, since $|\post{p_o}| = 0$, so $V_M(p_o, v) = 0$.
Therefore, there are at most $(|P| + |T| - 1)^2$ nodes $v_1, v_2$ with $V_N(v_1, v_2) = 1$, giving us
\[\sum_{v_1, v_2 \in P \cup T} V_M(v_1, v_2) \leq (|P| + |T| - 1)^2 \leq (|P| + |T|) \cdot (|P| + |T| - 1)\]
as claimed by the theorem. \hfill $\square$
\end{proof}

\propitemf{\propmin}{\no} 
According to Theorem~\ref{theo:cc-bounded}, $\C(M) \geq 0$ for every workflow net $M \in \mathcal{M}$. 
However, no workflow net can get the complexity score $0$, since this would mean that $V_N(v_1, v_2) = 1$ for every pair of nodes $v_1, v_2$ with $v_1 \neq p_i$. 
But this can never happen since $V_N(p_o, v) = 0$ for every node $v$.
Therefore, let $c \in \mathbb{R}$ be a real-number with $c > 0$ and $M$ be a workflow net with complexity $\C(M) = c$. 
If we choose a natural number $k > \frac{1}{c}$ and construct the workflow net $W_{\min,k}$ of Figure~\ref{fig:cc-min}, we get $\C(W_{\min,k}) = \frac{1}{2k+4} < \frac{1}{\sfrac{1}{c}} = c$. 
Therefore, $\C$ has no minimal value and thus doesn't fulfill this property.

\propitemf{\propinf}{\yes} 
Figure~\ref{fig:cc-min} shows how to construct infinitely many workflow nets with different complexity, since $\frac{1}{2k + 4} > \frac{1}{2(k+1) + 4}$ for any $k \in \mathbb{N}$.
Therefore, $|\{c \in \mathbb{R} \mid \exists M \in \mathcal{M}: \C(M) = c\}| \geq |\frac{1}{2k + 4} \mid k \in \mathbb{N}\}| = \infty$.

\propitemf{\propnotsup}{\yes} The counter-examples for property {\propfive} also show that $\C$ is not superadditive.

\propitemf{\propadd}{\no} The counter-examples for property {\propnine} also show that $\C$ is not additive.
\end{description}

\def\name{\tokensplitname}
\def\C{\tokensplit}
\subsubsection{Token Split}
The token split complexity measure quantifies the effort to keep track of the markings in a net if it contains concurrent splits.
With each concurrent split, the reader must keep track of the created tokens that need to be synchronized later. 
As an easy estimator for this efford, Mendling~\cite[p.128]{Men08} defines this measure as the sum of additional tokens created with each concurrent split.
We can directly translate this measure to workflow nets:
Let $W = (P,T,F,p_i,p_o)$ be a workflow net.
\begin{equation}
\label{eq:token-split}
\C(W) = \sum_{t \in T} (|\post{t}| - 1)
\end{equation}
Figure~\ref{fig:ts-examples} shows two example nets and their respective complexity scores.

\begin{figure}[ht]
\begin{center}
\scalebox{\scalefactor}{
\begin{tikzpicture}[node distance = 1.5cm,>=stealth',bend angle=0,auto]
	\node [place,tokens=1] (start) [label=below:$p_i$] {};
	\node [transition] (t1) [right of=start,label=center:$a$] {}
	edge [pre] (start);
	\node [place] (p1) [right of=t1] {}
	edge [pre] (t1);
	\node [transition] (t2) [right of=p1,label=center:$b$] {}
	edge [pre] (p1);
	\node [place] (p2) [above right of=t2] {}
	edge [pre] (t2);
	\node [place] (p3) [below right of=t2] {}
	edge [pre] (t2);
	\node [transition] (t3) [right of=p2,label=center:$c$] {}
	edge [pre] (p2);
	\node [transition] (t4) [right of=p3,label=center:$d$] {}
	edge [pre] (p3);
	\node [place] (p4) [right of=t3] {}
	edge [pre] (t3);
	\node [place] (p5) [right of=t4] {}
	edge [pre] (t4);
	\node [transition] (t5) [below right of=p4,label=center:$e$] {}
	edge [pre] (p4)
	edge [pre] (p5);
	\node [place] (p6) [right of=t5,label=below:$p_o$] {}
	edge [pre] (t5);
	\node at (0,1.15) {$W_1^{\name}$:};
\end{tikzpicture}}
\ \\
\ \\
\scalebox{\scalefactor}{
\begin{tikzpicture}[node distance = 1.5cm,>=stealth',bend angle=0,auto]
	\node [place,tokens=1] (start) [label=below:$p_i$] {};
	\node [transition] (t1) [right of=start,label=center:$a$] {}
	edge [pre] (start);
	\node [place] (p1) [above right of=t1] {}
	edge [pre] (t1);
	\node [place] (p11) [below right of=t1] {}
	edge [pre] (t1);
	\node [transition] (t2) [below right of=p1,label=center:$b$] {}
	edge [pre] (p1)
	edge [pre] (p11);
	\node [place] (p2) [above right of=t2] {}
	edge [pre] (t2);
	\node [place] (p3) [below right of=t2] {}
	edge [pre] (t2);
	\node [transition] (t3) [right of=p2,label=center:$c$] {}
	edge [pre] (p2);
	\node [transition] (t4) [right of=p3,label=center:$d$] {}
	edge [pre] (p3);
	\node [place] (p4) [right of=t3] {}
	edge [pre] (t3);
	\node [place] (p5) [right of=t4] {}
	edge [pre] (t4);
	\node [transition] (t5) [below right of=p4,label=center:$e$] {}
	edge [pre] (p4)
	edge [pre] (p5);
	\node [place] (p6) [right of=t5,label=below:$p_o$] {}
	edge [pre] (t5);
	\node at (0,1.15) {$W_2^{\name}$:};
\end{tikzpicture}}
\end{center}
\caption{Two workflow nets, $W_1^{\name}$ and $W_2^{\name}$, with different complexity scores according to the token split metric: $\C(W_1^{\name}) = 1$ and $\C(W_2^{\name}) = 2$.}
\label{fig:ts-examples}
\end{figure}
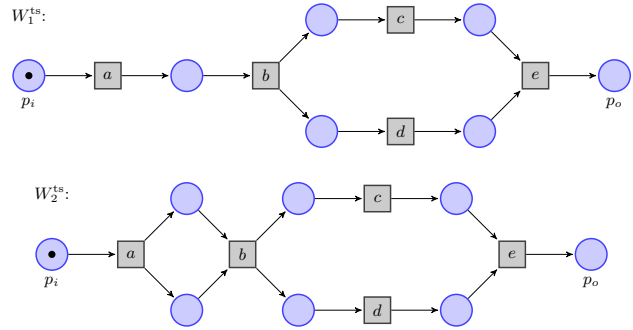

\begin{description}
\propitemf{\propone}{\yes} 
For the two nets $W_1^{\name}$ and $W_2^{\name}$ of Figure~\ref{fig:ts-examples}, we get the complexity scores $\C(W_1^{\name}) = 1 \neq 2 = \C(W_2^{\name})$.

\propitemf{\proptwo}{\no} 
Let $c \in \mathbb{N}$ be an arbitrary integer. 
Figure~\ref{fig:ts-fin} shows how to construct infinitely many workflow nets with complexity $c$. 
\begin{figure}[ht]
\begin{center}
\scalebox{\scalefactor}{
\begin{tikzpicture}[node distance = 1.5cm,>=stealth',bend angle=0,auto]
	\node [place,tokens=1] (start) [label=below:$p_i$] {};
	\node [transition] (t1) [right of=start,label=below:$t_1$] {}
	edge [pre] (start);
	\node [place] (p0) [right of=t1] {}
	edge [pre] (t1);
	\node (dots) [right of=p0] {$\dots$}
	edge [pre] (p0);
	\node [transition] (tk) [right of=dots,label=below:$t_k$] {}
	edge [pre] (dots);
	\node [place] (p1) [above right of=tk,label=below:$p^1$] {}
	edge [pre] (tk);
	\node [place] (pc) [below right of=tk,label=below:$p^{c+1}$] {}
	edge [pre] (tk);
	\node at ($0.5*(p1) + 0.5*(pc)$) {$\vdots$};
	\node [transition] (t1) [right of=p1,label=below:$t^1$] {}
	edge [pre] (p1);
	\node [transition] (tc) [right of=pc,label=below:$t^{c+1}$] {}
	edge [pre] (pc);
	\node at ($0.5*(t1) + 0.5*(tc)$) {$\vdots$};
	\node [place] (q1) [right of=t1,label=below:$q^1$] {}
	edge [pre] (t1);
	\node [place] (qc) [right of=tc,label=below:$q^{c+1}$] {}
	edge [pre] (tc);
	\node at ($0.5*(q1) + 0.5*(qc)$) {$\vdots$};
	\node [transition] (t) [below right of=q1,label=below:$t$] {}
	edge [pre] (q1)
	edge [pre] (qc);
	\node [place] (po) [right of=t,label=below:$p_o$] {}
	edge [pre] (t);
	\node at (0,1.15) {$W_{c,k}^{\name}$:};
\end{tikzpicture}}
\end{center}
\caption{A construction plan for workflow nets $W_{c,k}^{\name}$ with $k + c + 2$ transitions and $\C(W_{c,k}^{\name}) = c$ for any $c,k \in \mathbb{N}$.}
\label{fig:ts-fin}
\end{figure}
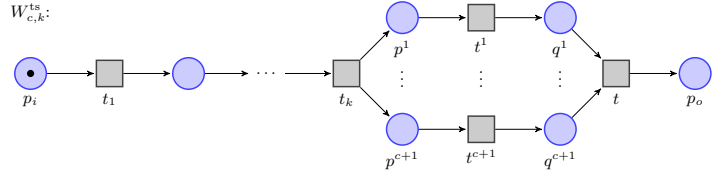

\propitemf{\propthree}{\yes} 
Figure~\ref{fig:ts-fin} shows how to construct structurally different workflow nets that get the same complexity score according to $\C$.

\propitemf{\propfour}{\yes} 
The two workflow nets $W_1^{\name}$ and $W_2^{\name}$ of Figure~\ref{fig:ts-examples} have the same language $L(W_1^{\name}) = \{\varepsilon, a, ab, abc, abd, abcd, abdc, abcde, abdce\} = L(W_2^{\name})$, but receive complexity scores $\C(W_1^{\name}) = 1 \neq 2 = \C(W_2^{\name})$.

\propitemf{\propfive}{\yes} 
To see that $\C$ fulfills this property, we use the following Theorem:
\begin{theorem}
\label{thm:ts}
Let $n \geq 2$. For any workflow nets $M_1, \dots, M_n \in \mathcal{M}$ we get:
\begin{itemize}
\item $\C(\seqop(M_1, \dots, M_n)) = \C(M_1) + \dots + \C(M_n)$,
\item $\C(\parop(M_1, \dots, M_n)) = \C(M_1) + \dots + \C(M_n) + n-1$,
\item $\C(\choiceop(M_1, \dots, M_n)) = \C(M_1) + \dots + \C(M_n)$ and
\item $\C(\loopop(M_1, \dots, M_n)) = \C(M_1) + \dots + \C(M_n)$.
\end{itemize}
\end{theorem}
\begin{proof}
The claim is obivous for the operations $\seqop$, $\choiceop$ and $\loopop$, since none of the out-going arrows of transitions in $M_1, \dots, M_n$ are changed and no new transition with more than one out-going arc is introduced by these operations.
For $\parop$, however, we introduce the new transition $t_i^*$ that has $n$ out-going arcs. Therefore, the token split complexity increases by $n-1$ in this case. \hfill $\square$
\end{proof}
Thus, according to Theorem~\ref{thm:ts}, $\C(\oplus(M_1, \dots, M_n)) \geq \C(M_i)$ for any workflow nets $M_1, \dots, M_n \in \mathcal{M}$, any operation $\oplus \in \{\seqop, \parop, \choiceop, \loopop\}$ and a model $M_i \in \{M_1, \dots, M_n\}$.

\newpage
\propitemf{\propsix}{\no} 
Let $M_1, M_2, M_3 \in \mathcal{M}$ be three arbitrary workflow nets with $\C(M_1) = \C(M_2)$. Theorem~\ref{thm:ts} implies:
\begin{itemize}
\item $\C(M_1 \seqop M_3) = \C(M_1) + \C(M_3)$ \\
\phantom{$\C(M_1 \seqop M_3)$ }$= \C(M_2) + \C(M_3) = C(M_2 \seqop M_3)$,
\item $\C(M_1 \parop M_3) = \C(M_1) + \C(M_3) + n-1$ \\
\phantom{$\C(M_1 \parop M_3)$ }$= \C(M_2) + \C(M_3) + n-1 = C(M_2 \parop M_3)$,
\item $\C(M_1 \choiceop M_3) = \C(M_1) + \C(M_3)$ \\
\phantom{$\C(M_1 \choiceop M_3)$ }$= \C(M_2) + \C(M_3) = C(M_2 \choiceop M_3)$ and
\item $\C(M_1 \loopop M_3) = \C(M_1) + \C(M_3)$ \\
\phantom{$\C(M_1 \loopop M_3)$ }$= \C(M_2) + \C(M_3) = C(M_2 \loopop M_3)$.
\end{itemize}
Therefore, $\C(M_1 \oplus M_3) = \C(M_2 \oplus M_3)$ for any $\oplus \in \{\seqop, \parop, \choiceop, \loopop\}$.

\propitemf{\propseven}{\yes} 
Figure~\ref{fig:ts-perm} shows two workflow nets, $W_3^{\name}$ and $W_4^{\name}$ that are permutations of each other, but $\C(W_4^{\name}) = 3 \neq 2 = \C(W_3^{\name})$.
\begin{figure}[ht]
\begin{center}
\begin{minipage}{0.55\textwidth}
\centering
\scalebox{\scalefactor}{
\begin{tikzpicture}[node distance = 1.5cm,>=stealth',bend angle=0,auto]
	\node [place,tokens=1] (start) [label=below:$p_i$] {};
	\node [transition] (t1) [right of=start,label=below:$t_1$] {}
	edge [pre] (start);
	\node [place] (p1) [above right of=t1,label=below:$p_1$] {}
	edge [pre] (t1);
	\node [place] (p2) [below right of=t1,label=below:$p_2$] {}
	edge [pre] (t1);
	\node [transition] (t2) [right of=p1,label=below:$t_2$] {}
	edge [pre] (p1);
	\node [transition] (t3) [right of=p2,label=below:$t_3$] {}
	edge [pre] (p2);
	\node [place] (p3) [right of=t2,label=below:$p_3$] {}
	edge [pre] (t2);
	\node [place] (p4) [right of=t3,label=below:$p_4$] {}
	edge [pre] (t3);
	\node [transition] (t4) [below right of=p3,label=below:$t_4$] {}
	edge [pre] (p3)
	edge [pre] (p4);
	\node [place] (end) [right of=t4,label=below:$p_o$] {}
	edge [pre] (t4);
	\node at (0,1.15) {$W_3^{\name}$:};
\end{tikzpicture}}
\end{minipage}
\begin{minipage}{0.44\textwidth}
\centering
\scalebox{\scalefactor}{
\begin{tikzpicture}[node distance = 1.5cm,>=stealth',bend angle=0,auto]
	\node [place,tokens=1] (start) [label=below:$p_i$] {};
	\node [transition] (t1) [right of=start,label=below:$t_1$] {}
	edge [pre] (start);
	\node [place] (p1) [right of=t1,label=below:$p_1$] {}
	edge [pre] (t1);
	\node [place] (p2) [below of=p1,label=below:$p_2$] {}
	edge [pre] (t1);
	\node [place] (p3) [below of=p2,label=below:$p_3$] {}
	edge [pre] (t1);
	\node [transition] (t2) [right of=p1,label=below:$t_2$] {}
	edge [pre] (p1)
	edge [pre] (p2)
	edge [pre] (p3);
	\node [transition] (t3) [above of=t2,label=below:$t_3$] {}
	edge [pre] (p1);
	\node [transition] (t4) [above of=t1,label=below:$t_4$] {}
	edge [post] (p1);
	\node [place] (p4) [left of=t3,label=below:$p_4$] {}
	edge [pre] (t3)
	edge [post] (t4);
	\node [place] (end) [right of=t2,label=below:$p_o$] {}
	edge [pre] (t2);
	\node at (0,1.15) {$W_4^{\name}$:};
\end{tikzpicture}}
\end{minipage}
\end{center}
\caption{Two workflow nets, $W_3^{\name}$ and $W_4^{\name}$, with $W_4^{\name} \in Perm(W_3^{\name})$, $\C(W_3^{\name}) = 1$ and $\C(W_4^{\name}) = 2$.}
\label{fig:ts-perm}
\end{figure}

\propitemf{\propeight}{\yes} 
Changing the labeling function neither changes the set of transitions nor their postsets, so $\C$ is robust against relabelings.

\propitemf{\propnine}{\no} 
Theorem~\ref{thm:ts} directly implies that only the operation $\parop$ is not subadditive. 
For all other operations, the number of token splits does not strictly increase, no matter which workflow nets we use as input.

\propitemf{\propdef}{\yes} 
The set of transitions $T$ is defined for all workflow nets, as well as the postset of a transition.
Furthermore, $\C$ only returns non-negative values, since every summand is non-negative:
If there was a negative summand in $\C$, this would mean that there is a transition $t \in T$ with $\post{t} = \emptyset$. 
But this is impossible in a workflow net, since $t$ must lie on a path from $p_i$ to $p_o$ and therefore needs at least one outgoing edge.

\propitemf{\propmin}{\yes} 
For the smallest possible workflow net, $\netminwf$, we get $\C(\netminwf) = 0$.
We know from the analysis of property {\propdef} that $\C(M) \geq 0$ for all workflow nets $M \in \mathcal{M}$, so $\C$ reaches its minimum value for the workflow net $\netminwf$. 

\propitemf{\propinf}{\yes} 
Figure~\ref{fig:ts-fin} shows the construction of workflow nets $W_{c,k}^{\name}$ with complexity score $c$ for any $c,k \in \mathbb{N}$. 
Therefore, we can construct infinitely many nets with complexity score $c$: $|\{c \in \mathbb{N} \mid \exists M \in \mathcal{M}: \C(M) = c\}| \geq |\mathbb{N}| = \infty$.

\propitemf{\propnotsup}{\yes} 
Theorem~\ref{thm:ts} directly implies that $\C$ is not superadditive for any of the operations $\seqop, \parop, \choiceop, \loopop$.

\propitemf{\propadd}{\no} Theorem~\ref{thm:ts} shows that $\C$ fulfills this property for the operations $\seqop, \choiceop$ and $\loopop$, but not for the operation $\parop$. 
Thus, the property holds not for all of the operations of Definition~\ref{def:operations}.
\end{description}

\def\name{\separabilityname}
\def\C{\separability}
\subsubsection{Separability}
Separability measures how easy it is to analyze parts of a model in isolation.
The idea is that a process model is easier to understand if we can break it down into multiple isolated pieces. 
Mendling~\cite[p.122]{Men08} proposes a separability metric based on cut-vertices in the graph of an EPC model. 
In graph theory, cut-vertices are vertices in a graph whose removal would split the graph into at least two connected components.
For an undirected graph $G = (V,E)$, the set of its connected components is 
\[Conn(G) := \{G' = (V', E') \mid \forall u,v \in V': \text{there is a path from } u \text{ to } v \text{ in } G'\}.\]
With this definition, the set of cut-vertices of an undirected graph $G$ is 
\[\mathcal{V}(G) := \left\{v \in V \mid \left|Conn((V \setminus \{v\}, E \setminus \{\{x,y\} \mid x = v \lor y = v\}))\right| \geq 2 \right\}.\]
This definition can also be applied to directed graphs by ignoring the direction of arcs in the graph.

Mendling defines the separability measure as the number of cut-vertices of the process model divided by the total number of vertices.
We can directly translate this to workflow nets, but keep in mind that the input- and output places of a workflow net can never be cut-vertices.
Let $W = (P,T,F,p_i,p_o)$ be a workflow net and $\overline{W} = (P \cup T, \{\{x,y\} \mid (x,y) \in F \lor (y,x) \in F\})$ be the undirected graph produced by ignoring the arc directions in $W$.
The separability measure is then defined as
\begin{equation}
\label{eq:c-separability}
\C(W) = 1 - \frac{|\mathcal{V}(\overline{W})|}{|P| + |T| - 2}
\end{equation}
In contrast to the original definition, we subtract the fraction from $1$ to ensure that workflow nets with many cut-vertices get a low complexity score.
This way, a low complexity score refers to less complex models.
Figure~\ref{fig:sep-examples} shows three example nets and their respective complexity scores.
\begin{figure}[ht]
\begin{center}
\begin{minipage}{0.2525\textwidth}
\centering
\scalebox{\scalefactor}{
\begin{tikzpicture}[node distance = 1.5cm,>=stealth',bend angle=0,auto]
	\node [place,tokens=1] (start) [label=below:$p_i$] {};
	\node [transition] (t1) [right of=start,label=center:$a$] {}
	edge [pre] (start);
	\node [place] (p1) [right of=t1,label=below:$p_o$] {}
	edge [pre] (t1);
	\node at (0,1.15) {$W_1^{\name}$:};
	\draw[opacity=0] ($(start)-(0.25,1.75)$) rectangle ($(p1) + (0.25,1.75)$);
\end{tikzpicture}}
\end{minipage}
\begin{minipage}{0.25\textwidth}
\centering
\scalebox{\scalefactor}{
\begin{tikzpicture}[node distance = 1.5cm,>=stealth',bend angle=0,auto]
	\node [place,tokens=1] (start) [label=below:$p_i$] {};
	\node [transition] (t2) [right of=start,label=center:$a$] {}
	edge [pre] (start);
	\node [transition] (t1) [above of=t2,label=center:$a$] {}
	edge [pre] (start);
	\node [transition] (t3) [below of=t2,label=center:$a$] {}
	edge [pre] (start);
	\node [place] (p1) [right of=t2,label=below:$p_o$] {}
	edge [pre] (t1)
	edge [pre] (t2)
	edge [pre] (t3);
	\node at (0,1.15) {$W_2^{\name}$:};
	\draw[opacity=0] ($(start)-(0.25,1.75)$) rectangle ($(p1) + (0.25,1.75)$);
\end{tikzpicture}}
\end{minipage}
\begin{minipage}{0.425\textwidth}
\centering
\scalebox{\scalefactor}{
\begin{tikzpicture}[node distance = 1.5cm,>=stealth',bend angle=0,auto]
	\node [place,tokens=1] (start) [label=below:$p_i$] {};
	\node [transition] (t1) [right of=start,label=center:$a$] {}
	edge [pre] (start);
	\node [place] (p1) [right of=t1] {}
	edge [pre] (t1);
	\node [transition] (t2) [right of=p1,label=center:$b$] {}
	edge [pre] (p1);
	\node [place] (p2) [right of=t2,label=below:$p_o$] {}
	edge [pre] (t2);
	\node at (0,1.15) {$W_3^{\name}$:};
	\draw[opacity=0] ($(start)-(0.25,1.75)$) rectangle ($(p2) + (0.25,1.75)$);
\end{tikzpicture}}
\end{minipage}
\end{center}
\caption{Three workflow nets, $W_1^{\name}$, $W_2^{\name}$, $W_3^{\name}$, with $\C(W_2^{\name}) = 1$ and $\C(W_1^{\name}) = \C(W_3^{\name}) = 0$.}
\label{fig:sep-examples}
\end{figure}

\begin{description}
\propitemf{\propone}{\yes} For the two nets $W_1^{\name}$ and $W_2^{\name}$ of Figure~\ref{fig:sep-examples}, we get the complexity scores $\C(W_1^{\name}) = 0 \neq 1 = \C(W_2^{\name})$. 

\propitemf{\proptwo}{\no} 
For any $n \in \mathbb{N}$, the workflow net $W_{fin,n}^{\name}$ of Figure~\ref{fig:sep-fin}, we have $\C(W_{fin,n}^{\name}) = 0$.
Therefore, there are infinitely many workflow nets with this complexity score.
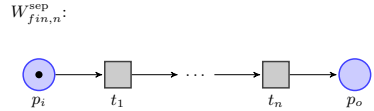
\begin{figure}[ht]
\begin{center}
\scalebox{\scalefactor}{
\begin{tikzpicture}[node distance = 1.5cm,>=stealth',bend angle=0,auto]
	\node [place,tokens=1] (start) [label=below:$p_i$] {};
	\node [transition] (t1) [right of=start,label=below:$t_1$] {}
	edge [pre] (start);
	\node (p1) [right of=t1] {$\dots$}
	edge [pre] (t1);
	\node [transition] (t2) [right of=p1,label=below:$t_n$] {}
	edge [pre] (p1);
	\node [place] (p2) [right of=t2,label=below:$p_o$] {}
	edge [pre] (t2);
	\node at (0,1.15) {$W_{fin,n}^{\name}$:};
\end{tikzpicture}}
\end{center}
\caption{A worklfow net with complexity $0$ and $n$ transitions.}
\label{fig:sep-fin}
\end{figure}

\propitemf{\propthree}{\yes} 
The workflow nets $W_1^{\name}$ and $W_3^{\name}$ of Figure~\ref{fig:sep-examples} are different in structure, but both get the same complexity score $\C(W_1^{\name}) = 0 = \C(W_3^{\name})$.

\propitemf{\propfour}{\yes} 
Take the workflow nets $W_1^{\name}$ and $W_2^{\name}$ of Figure~\ref{fig:sep-examples}.
Their languages are $L(W_1^{\name}) = \{\varepsilon, a\} = L(W_2^{\name})$, but their complexity scores are $\C(W_1^{\name}) = 0 \neq 1 = \C(W_2^{\name})$.

\propitemf{\propfive}{\no} 
The workflow net $W_2^{\name}$ of Figure~\ref{fig:sep-examples} contains no cut-vertices, so $\C(W_2^{\name}) = 1$.
But, since all operations of Definition~\ref{def:operations} except $\choiceop$ introduce at least one cut-vertex, we get:
\begin{itemize}
\item $\C(W_2^{\name} \seqop W_2^{\name}) = 1 - \frac{3}{9} = \frac{2}{3} < 1 = \C(W_2^{\name})$,
\item $\C(W_2^{\name} \parop W_2^{\name}) = 1 - \frac{2}{12} = \frac{5}{6} < 1 = \C(W_2^{\name})$ and
\item $\C(W_2^{\name} \loopop W_2^{\name}) = 1 - \frac{2}{20} = \frac{9}{10} < 1 = \C(W_2^{\name})$.
\end{itemize}
For the operator $\choiceop$, Figure~\ref{fig:sep-mon} shows two workflow nets, $W_5^{\name}$ and $W_6^{\name}$, where $\C(W_6^{\name} \choiceop W_5^{\name}) \approx 0.9286 < 1 = \C(W_6^{\name})$.
\begin{figure}[ht]
\begin{center}
\scalebox{\scalefactor}{
\begin{tikzpicture}[node distance = 1.5cm,>=stealth',bend angle=0,auto]
	\node [place,tokens=1] (start) [] {};
	\node [transition] (t1) [above right of=start,yshift=0.5cm] {}
	edge [pre] (start);
	\node [transition] (t2) [below right of=start,yshift=-0.5cm] {}
	edge [pre] (start);
	\node [place] (p1) [right of=t1] {}
	edge [pre] (t1);
	\node [transition] (t3) [right of=p1] {}
	edge [pre] (p1);
	\node [red place] (p3) [right of=t3] {}
	edge [pre] (t3);
	\node [transition] (t6) [above of=p3] {}
	edge [pre, bend right=30] (p3)
	edge [post, bend left=30] (p3);
	\node [transition] (t6) [right of=p3] {}
	edge [pre] (p3);
	\node [place] (p5) [right of=t6] {}
	edge [pre] (t6);
	\node [transition] (t7) [right of=p5] {}
	edge [pre] (p5);
	\node [place] (p6) [below right of=t7,yshift=-0.5cm] {}
	edge [pre] (t7);

	\node [transition] (t4) [below of=p3,yshift=-0.56cm] {};
	\node [place] (p2) [below left of=t4] {}
	edge [pre] (t2)
	edge [post] (t4);
	\node [transition] (t5) [below right of=p2] {}
	edge [pre] (p2);
	\node [place] (p4) [below right of=t4] {}
	edge [pre] (t4)
	edge [pre] (t5);
	\node [transition] (t8) [below left of=p6,yshift=-0.5cm] {}
	edge [pre] (p4)
	edge [post] (p6);
	
	\node at ($(start) + (0,2)$) {$W_4^{\name}$:};	
	
	\begin{pgfonlayer}{bg}
		\draw[rounded corners,draw=gray,fill=lightgray!50!white] ($(p1) - (0.75, 0.75)$) rectangle ($(p5) + (0.75, 2)$);
		\node[gray] at ($(p1) + (-0.4,-0.5)$) {$W_5^{\name}$};
		\draw[rounded corners,draw=gray,fill=lightgray!50!white] ($(p2) - (0.75, 1.5)$) rectangle ($(p4) + (0.75, 1.5)$);
		\node[gray] at ($(p2) + (-0.4,-1.25)$) {$W_6^{\name}$};
	\end{pgfonlayer}
\end{tikzpicture}}
\end{center}
\caption{A workflow net $W_4^{\name}$ with $W_4^{\name} = \choiceop(W_5^{\name}, W_6^{\name})$. Computing the separability metric yields $\C(W_4^{\name}) = 1 - \frac{1}{14} \approx 0.9286$, $\C(W_5^{\name}) = 1 - \frac{1}{4} = 0.75$ and $\C(W_6^{\name}) = 1 - \frac{0}{2} = 1$.}
\label{fig:sep-mon}
\end{figure}

\propitemf{\propsix}{\yes} 
Figure~\ref{fig:sep-comp} shows workflow nets $W_7^{\name}$ and $W_8^{\name}$ with the same complexity score $\C(W_7^{\name}) = 0.25 = \C(W_8^{\name})$.
But if we combine these nets with $\netminwf$ of Figure~\ref{fig:netminwf}, we get:
\begin{itemize}
\item $\C(W_7^{\name} \seqop \netminwf) = 1 - \frac{7}{8} = \frac{1}{8} \neq \frac{1}{6} = 1 - \frac{10}{12} = \C(W_8^{\name} \seqop \netminwf)$,
\item $\C(W_7^{\name} \parop \netminwf) = 1 - \frac{3}{11} = \frac{8}{11} \neq \frac{13}{15} = 1 - \frac{2}{15} = \C(W_8^{\name} \parop \netminwf)$,
\item $\C(W_7^{\name} \choiceop \netminwf) = 1 - \frac{1}{13} = \frac{12}{13} \neq 1 = 1 - \frac{0}{17} = \C(W_8^{\name} \choiceop \netminwf)$,
\item $\C(W_7^{\name} \loopop \netminwf) = 1 - \frac{5}{17} = \frac{12}{17} \neq \frac{17}{21} = 1 - \frac{4}{21} = \C(W_8^{\name} \loopop \netminwf)$.
\end{itemize}
\begin{figure}[ht]
\begin{center}
\scalebox{\scalefactor}{
\begin{tikzpicture}[node distance = 1.5cm,>=stealth',bend angle=0,auto]
	\node [place,tokens=1] (start) [label=below:$p_i$] {};
	\node [transition,draw=red,fill=red!20] (t1) [right of=start] {}
	edge [pre] (start);
	\node [red place] (p1) [right of=t1] {}
	edge [pre] (t1);
	\node [transition,draw=red,fill=red!20] (t2) [right of=p1] {}
	edge [pre] (p1);
	\node [place] (end) [right of=t2,label=below:$p_o$] {}
	edge [pre] (t2);
	\node [transition] (t3) [above of=p1] {}
	edge [pre, bend left=20] (p1)
	edge [post, bend right=20] (p1);
	\node at (0,1.15) {$W_7^{\name}$:};
\end{tikzpicture}} \\
\ \\
\scalebox{\scalefactor}{
\begin{tikzpicture}[node distance = 1.5cm,>=stealth',bend angle=0,auto]
	\node [place,tokens=1] (start) [label=below:$p_i$] {};
	\node [transition,draw=red,fill=red!20] (t1) [right of=start] {}
	edge [pre] (start);
	\node [red place] (p1) [right of=t1] {}
	edge [pre] (t1);
	\node [transition,draw=red,fill=red!20] (t2) [right of=p1] {}
	edge [pre] (p1);
	\node [red place] (p2) [right of=t2] {}
	edge [pre] (t2);
	\node [transition] (t3) [above right of=p2] {}
	edge [pre] (p2);
	\node [transition] (t4) [below right of=p2] {}
	edge [pre] (p2);
	\node [red place] (p3) [below right of=t3] {}
	edge [pre] (t3)
	edge [pre] (t4);
	\node [transition,draw=red,fill=red!20] (t5) [right of=p3] {}
	edge [pre] (p3);
	\node [place] (end) [right of=t5] {}
	edge [pre] (t5);
	\node at (0,1.15) {$W_8^{\name}$:};
\end{tikzpicture}}
\end{center}
\caption{Two workflow nets, $W_7^{\name}$, $W_8^{\name}$, with $\C(W_7^{\name}) = \C(W_8^{\name}) = 0.25$.}
\label{fig:sep-comp}
\end{figure}

\propitemf{\propseven}{\yes} 
Figure~\ref{fig:sep-perm} shows two workflow nets, $W_9^{\name}$ and $W_{10}^{\name}$ that are permutations of each other, but $\C(W_9^{\name}) = 0.25 \neq 0.5 = \C(W_{10}^{\name})$.
\begin{figure}[ht]
\begin{center}
\begin{minipage}{0.45\textwidth}
\centering
\scalebox{\scalefactor}{
\begin{tikzpicture}[node distance = 1.5cm,>=stealth',bend angle=0,auto]
	\node [place,tokens=1] (start) [label=below:$p_i$] {};
	\node [transition] (t1) [right of=start,label=below:$t_1$] {}
	edge [pre] (start);
	\node [place] (p1) [right of=t1,label=below:$p_1$] {}
	edge [pre] (t1);
	\node [transition] (t2) [right of=p1,label=below:$t_3$] {}
	edge [pre] (p1);
	\node [place] (end) [right of=t2,label=below:$p_o$] {}
	edge [pre] (t2);
	\node [transition] (t3) [above of=p1,label=above:$t_2$] {}
	edge [pre, bend left=20] (p1)
	edge [post, bend right=20] (p1);
	\node at (0,1.15) {$W_9^{\name}$:};
	\draw[opacity=0] ($(start)-(0.25,1.75)$) rectangle ($(end) + (0.25,2.25)$);
\end{tikzpicture}}
\end{minipage}
\begin{minipage}{0.5\textwidth}
\centering
\scalebox{\scalefactor}{
\begin{tikzpicture}[node distance = 1.5cm,>=stealth',bend angle=0,auto]
	\node [place,tokens=1] (start) [label=below:$p_i$] {};
	\node [transition] (t1) [right of=start,label=below:$t_1$] {}
	edge [pre] (start);
	\node [place] (p2) [right of=t1,label=below:$p_2$] {}
	edge [pre] (t1);
	\node [transition] (t3) [above right of=p2,label=below:$t_2$] {}
	edge [pre] (p2);
	\node [transition] (t4) [below right of=p2,label=below:$t_3$] {}
	edge [pre] (p2);
	\node [place] (end) [below right of=t3,label=below:$p_o$] {}
	edge [pre] (t3)
	edge [pre] (t4);
	\node at (0,1.15) {$W_{10}^{\name}$:};
	\draw[opacity=0] ($(start)-(0.25,1.75)$) rectangle ($(end) + (0.25,2.25)$);
\end{tikzpicture}}
\end{minipage}
\end{center}
\caption{Two workflow nets with $W_9^{\name} \in Perm(W_{10}^{\name})$ and complexity scores $\C(W_9^{\name}) = 0.25 \neq 0.5 = \C(W_{10}^{\name})$.}
\label{fig:sep-perm}
\end{figure}

\propitemf{\propeight}{\yes} 
The number of cut-vertices is independent of the labeling.
Uniformly changing the labeling therefore has no impact on the complexity score of $\C$.

\propitemf{\propnine}{\no} 
$\C$ fulfills this property only for the operations $\parop$, $\choiceop$ and $\loopop$:
Take the workflow net $\netminwf$ of Figure~\ref{fig:netminwf}. We get:
\begin{itemize}
\item $\C(\netminwf \parop \netminwf) = 1 - \frac{2}{8} = 0.75 > 0 + 0 = \C(\netminwf) + \C(\netminwf)$,
\item $\C(\netminwf \choiceop \netminwf) = 1 - \frac{0}{10} = 1 > 0 + 0 = \C(\netminwf) + \C(\netminwf)$,
\item $\C(\netminwf \loopop \netminwf) = 1 - \frac{4}{14} \approx 0.7143 > 0 + 0 = \C(\netminwf) + \C(\netminwf)$.
\end{itemize}
For the operation $\seqop$, however, $\C$ is subadditive.
We show this by first proving the following Theorem
\begin{theorem}
\label{thm:cut-vertices-seqop}
Let $M_1, M_2 \in \mathcal{M}$ be two workflow nets. 
Further, let $p_o^{M_1}$ be the output place of $M_1$, $p_i^{M_2}$ be the input place of $M_2$ and $t_{\seqop}$ be the newly introduced transition connecting $p_o^{M_1}$ and $p_i^{M_2}$.
Then, the set of cut-vertices in the net $M_1 \seqop M_2$ is $\mathcal{V}(\overline{M_1 \seqop M_2}) = \mathcal{V}(\overline{M_1}) \cup \mathcal{V}(\overline{M_2}) \cup \{p_o^{M_1}, p_i^{M_2}, t_{\seqop}\}$.
\end{theorem}
\begin{proof}
It is obvious by construction that $p_o^{M_1}$, $p_i^{M_2}$ and $t_{\seqop}$ are cut-vertices in $M_1 \seqop M_2$.
Therefore, let $v$ be a cut-vertex in $M_1$ or $M_2$. 
Without loss of generality, let $v$ be a cut-vertex of $M_1$.
Then, removing $v$ leads to more than one connected component in $M_1$. 
Because the $\seqop$-operation only introduces the edges $(p_o^{M_1}, t_{\seqop})$ and $(t_{\seqop} p_i^{M_2})$, removing $v$ in $M_1 \seqop M_2$ also leads to more than one connected component. 
Therefore, $v$ is also a cut-vertex in $M_1 \seqop M_2$.

Wlog let $\overline{v}$ be a vertex in $M_1$ that is not a cut-vertex.
Then, for each pair of vertices, there is a path connecting them that avoids $\overline{v}$. 
Since we don't remove any arcs with the $\seqop$-operation, this path still exists in $M_1 \seqop M_2$, so $\overline{v}$ is also not a cut-vertex in this net. \hfill $\square$
\end{proof}
Because of Theorem~\ref{thm:cut-vertices-seqop}, we know that for any nets $M_1, M_2 \in \mathcal{M}$, the set of non-cut-vertices in $M_1 \seqop M_2$ is the union of the set of non-cut-vertices in $M_1$ and $M_2$.
Let $M_1, M_2 \in \mathcal{M}$, $x_1$ be the number of non-cut-vertices in $M_1$, $x_2$ the number of non-cut-vertices in $M_2$ and $n_1, n_2$ be the number of nodes in $M_1, M_2$ respectively. We then get:
\begin{align*}
\C(M_1) + \C(M_2) &= \frac{x_1}{n_1 - 2} + \frac{x_2}{n_2 - 2} > \frac{x_1}{n_1 + n_2 - 2} + \frac{x_2}{n_1 + n_2 - 2} \\
&= \frac{x_1 + x_2}{n_1 + n_2 - 2} = \C(M_1 \seqop M_2)
\end{align*}
because $n_1, n_2 \geq 3$. Therefore, $\C$ fulfills this property for $\seqop$.

\propitemf{\propdef}{\yes} 
The separability measure is defined for all workflow nets, since the set of cut-vertices is defined for all graphs.
The metric returns only non-negative values, since the input and output place can never be cut-vertices, so we get $|\mathcal{V}(\overline{W})| \leq |P| + |T| - 2$.

\propitemf{\propmin}{\yes} 
As shown for property {\propdef}, $\C$ returns no value below $0$. 
In the workflow net $\netminwf$ of Figure~\ref{fig:netminwf}, every vertex except for the input and output place are cut-vertices, so $\C(\netminwf) = 0$ is the minimum complexity score.

\propitemf{\propinf}{\yes} 
Figure~\ref{fig:sep-inf} shows how to construct workflow nets with $3n + 2m + 2$ nodes for any $n, m \in \mathbb{N}$, such that $\C(W_{n,m}^{\name}) = 1 - \frac{n + 2}{3n + 2m}$. 
Therefore, $|\{c \mid \exists M \in \mathcal{M}: \C(M) = c\}| \geq |\{r \in \mathbb{R} \mid \exists n,m \in \mathbb{N}: r = 1 - \frac{n + 2}{3n + 2m}\}|$ $= \infty$, so $\C$ fulfills this property.
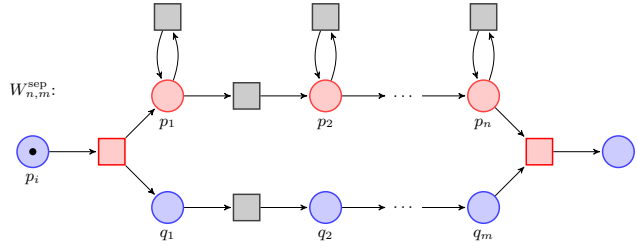
\begin{figure}[ht]
\begin{center}
\scalebox{\scalefactor}{
\begin{tikzpicture}[node distance = 1.5cm,>=stealth',bend angle=0,auto]
	\node [place,tokens=1] (start) [label=below:$p_i$] {};
	\node [transition, draw=red, fill=red!20] (t1) [right of=start] {}
	edge [pre] (start);
	\node [red place] (p1) [above right of=t1, label=below:$p_1$] {}
	edge [pre] (t1);
	\node [transition] (t2) [right of=p1] {}
	edge [pre] (p1);
	\node [red place] (p2) [right of=t2, label=below:$p_2$] {}
	edge [pre] (t2);
	\node (dots) [right of=p2] {$\dots$}
	edge [pre] (p2);
	\node [red place] (pn) [right of=dots, label=below:$p_n$] {}
	edge [pre] (dots);
	\node [place] (q1) [below right of=t1, label=below:$q_1$] {}
	edge [pre] (t1);
	\node [transition] (t3) [right of=q1] {}
	edge [pre] (q1);
	\node [place] (q2) [right of=t3, label=below:$q_2$] {}
	edge [pre] (t3);
	\node (dots) [right of=q2] {$\dots$}
	edge [pre] (q2);
	\node [place] (qm) [right of=dots, label=below:$q_m$] {}
	edge [pre] (dots);
	\node [transition, draw=red, fill=red!20] (t4) [below right of=pn] {}
	edge [pre] (pn)
	edge [pre] (qm);
	\node [place] (end) [right of=t4] {}
	edge [pre] (t4);
	\node [transition] (tp1) [above of=p1] {}
	edge [pre, bend left=20] (p1)
	edge [post, bend right=20] (p1);
	\node [transition] (tp2) [above of=p2] {}
	edge [pre, bend left=20] (p2)
	edge [post, bend right=20] (p2);
	\node [transition] (tpn) [above of=pn] {}
	edge [pre, bend left=20] (pn)
	edge [post, bend right=20] (pn);
	\node at (0,1.15) {$W_{n,m}^{\name}$:};
\end{tikzpicture}}
\end{center}
\caption{A worklfow net with $3n + 2m + 2$ nodes. All nodes highlighted in red are cut vertices, so its complexity is $\C(W_{n,m}^{\name}) = 1 - \frac{n + 2}{3n + 2m}$.}
\label{fig:sep-inf}
\end{figure}

\propitemf{\propnotsup}{\yes} The counter-examples for property {\propfive} show that $\C$ is not superadditive for any of the operations. 

\propitemf{\propadd}{\no} The counter-examples for property {\propfive} show that $\C$ is not additive for any of the operations. 
\end{description}

\newpage
\def\name{\controlflowname}
\def\C{\controlflow}
\subsubsection{Control Flow Complexity}
Similar to token-split complexity, the control flow complexity~\cite{Car05} aims to quantify how many different possible control flows the user needs to keep track of to fully understand the net.
In contrast to the token split metric, however, control flow complexity considers not only parallel splits, but all connector types of the net.
Cardoso~\cite{Car05} argues that parallel splits in BPMN models don't increase the number of possible control flows, while an exclusive choice with $k$ outgoing edges creates $k$ alternative control flows in the net.
Since connector nodes are always harder to understand than pure sequences, parallel splits give a small penalty to the complexity score.
An inclusive choice is even worse, since it allows for any combination of transitions in the postset, as long as at least one of them fires.
For an inclusive choice with $k$ out-going edges, the control flow complexity score increases by $2^k - 1$.

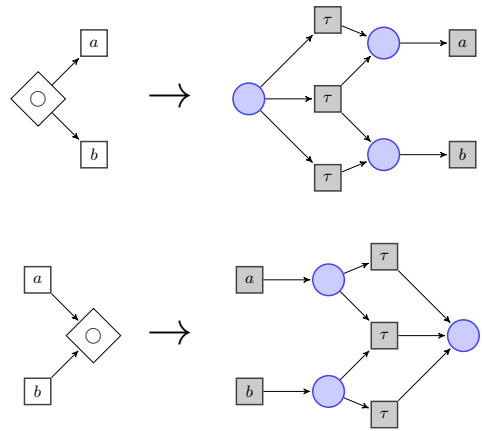
\begin{figure}[ht]
\begin{center}
\scalebox{\scalefactor}{
\begin{tikzpicture}[node distance = 1.5cm,>=stealth',bend angle=0,auto]
	\node[draw,diamond] (or) at (-3,0) {$\bigcirc$};
	\node[transition,fill=white] (bpmn-a) [above right of=or] {$a$}
	edge [pre] (or);
	\node[transition,fill=white] (bpmn-b) [below right of=or] {$b$}
	edge [pre] (or);
	
	\node at (-0.5,0) {\Huge $\rightarrow$};
	
	\node [place] (start) at (1,0) {};
	\node [transition] (t2) [right of=start,label=center:$\tau$] {}
	edge [pre] (start);
	\node [transition] (t1) [above of=t2,label=center:$\tau$] {}
	edge [pre] (start);
	\node [transition] (t3) [below of=t2,label=center:$\tau$] {}
	edge [pre] (start);
	\node [place] (p1) [above right of=t2] {}
	edge [pre] (t1)
	edge [pre] (t2);
	\node [place] (p2) [below right of=t2] {}
	edge [pre] (t2)
	edge [pre] (t3);
	\node [transition] (a) [right of=p1,label=center:$a$] {}
	edge [pre] (p1);
	\node [transition] (b) [right of=p2,label=center:$b$] {}
	edge [pre] (p2);
	
	\node[draw,diamond] (or) at (-1.95,-4.5) {$\bigcirc$};
	\node[transition,fill=white] (bpmn-a) [above left of=or] {$a$}
	edge [post] (or);
	\node[transition,fill=white] (bpmn-b) [below left of=or] {$b$}
	edge [post] (or);
	
	\node at (-0.5,-4.5) {\Huge $\rightarrow$};
	
	\node [place] (start) at (5.075,-4.5) {};
	\node [transition] (t2) [left of=start,label=center:$\tau$] {}
	edge [post] (start);
	\node [transition] (t1) [above of=t2,label=center:$\tau$] {}
	edge [post] (start);
	\node [transition] (t3) [below of=t2,label=center:$\tau$] {}
	edge [post] (start);
	\node [place] (p1) [above left of=t2] {}
	edge [post] (t1)
	edge [post] (t2);
	\node [place] (p2) [below left of=t2] {}
	edge [post] (t2)
	edge [post] (t3);
	\node [transition] (a) [left of=p1,label=center:$a$] {}
	edge [post] (p1);
	\node [transition] (b) [left of=p2,label=center:$b$] {}
	edge [post] (p2);
\end{tikzpicture}}
\end{center}
\caption{Inclusive choice connectors in BPMN notation (left) and their translation to workflow nets (right).}
\label{fig:inclusive-or}
\end{figure}
In workflow nets, we don't have single nodes that initiate an inclusive choice. 
Instead, we model such a choice like shown in Figure~\ref{fig:inclusive-or}.
That way, we can translate any or-split in BPMN notation into such a structure.
We need to make sure that this translation would not get a smaller complexity score than an inclusive choice node for BPMN.
But this is not the case, since we need to enable all combinations of transitions in the inclusive choice. 
Therefore, the place labeled $p$ in Figure~\ref{fig:inclusive-or} alone increases the complexity score by $\sum_{i = 1}^k \binom{k}{i} = 2^k - 1$ if $k$ transitions form an inclusive choice.
Apart from $p$, we have also $2^k - 1 - k$ transitions in the or-split-structure that are parallel splits, as well as $k$ xor-joins in the or-join-structure, further increasing the complexity score. 
For simplicity, we accept that the penalty for inclusive choices in workflow nets is higher than in BPMN and define: 
\begin{equation}
\label{eq:cfc}
\C(W) = |\mathcal{S}_{\text{and}}^W| + \sum_{p \in \mathcal{S}_{\text{xor}}^W} |\post{p}|.
\end{equation}
Cardoso~\cite{Car05} already analyzed which of Weyuker's properties the control flow complexity measure fulfills.
Since we translated the definition to workflow nets, we revisit these analyses and check that they didn't change.
Figure~\ref{fig:cfc-examples} shows three example nets and their respective complexity scores.
\begin{figure}[ht]
\begin{center}
\begin{minipage}{0.23\textwidth}
\centering
\scalebox{\scalefactor}{
\begin{tikzpicture}[node distance = 1.5cm,>=stealth',bend angle=0,auto]
	\node [place,tokens=1] (start) [label=below:$p_i$] {};
	\node [transition] (t1) [right of=start,label=center:$a$] {}
	edge [pre] (start);
	\node [place] (p1) [right of=t1,label=below:$p_o$] {}
	edge [pre] (t1);
	\node at (0,1.15) {$W_1^{\name}$:};
	\draw[opacity=0] ($(start)-(0.25,1.5)$) rectangle ($(p1) + (0.25,1.5)$);
\end{tikzpicture}}
\end{minipage}
\begin{minipage}{0.2\textwidth}
\centering
\scalebox{\scalefactor}{
\begin{tikzpicture}[node distance = 1.5cm,>=stealth',bend angle=0,auto]
	\node [place,tokens=1] (start) [label=below:$p_i$] {};
	\node [transition] (t1) [above right of=start,label=center:$a$] {}
	edge [pre] (start);
	\node [transition] (t2) [below right of=start,label=center:$a$] {}
	edge [pre] (start);
	\node [place] (p1) [below right of=t1,label=below:$p_o$] {}
	edge [pre] (t1)
	edge [pre] (t2);
	\node at (0,1.15) {$W_2^{\name}$:};
	\draw[opacity=0] ($(start)-(0.25,1.5)$) rectangle ($(p1) + (0.25,1.5)$);
\end{tikzpicture}}
\end{minipage}
\begin{minipage}{0.55\textwidth}
\centering
\scalebox{\scalefactor}{
\begin{tikzpicture}[node distance = 1.5cm,>=stealth',bend angle=0,auto]
	\node [place,tokens=1] (start) [label=below:$p_i$] {};
	\node [transition] (t1) [right of=start,label=center:$a$] {}
	edge [pre] (start);
	\node [place] (p1) [above right of=t1] {}
	edge [pre] (t1);
	\node [place] (p2) [below right of=t1] {}
	edge [pre] (t1);
	\node [transition] (t2) [right of=p1,label=center:$b$] {}
	edge [pre] (p1);
	\node [transition] (t3) [right of=p2,label=center:$c$] {}
	edge [pre] (p2);
	\node [place] (p3) [right of=t2] {}
	edge [pre] (t2);
	\node [place] (p5) [right of=t3] {}
	edge [pre] (t3);
	\node [place] (p4) at ($0.5*(p3) + 0.5*(p5) - (0.33,0)$) {}
	edge [pre] (t3);
	\node [transition] (t4) [below right of=p3,label=center:$d$] {}
	edge [pre] (p3)
	edge [pre] (p4)
	edge [pre] (p5);
	\node [place] (end) [right of=t4,label=below:$p_o$] {}
	edge [pre] (t4);
	\node at (0,1.15) {$W_3^{\name}$:};
	\draw[opacity=0] ($(start)-(0.25,1.5)$) rectangle ($(end) + (0.25,1.5)$);
\end{tikzpicture}}
\end{minipage}
\end{center}
\caption{Three workflow nets, $W_1^{\name}, W_2^{\name}, W_3^{\name}$, with $\C(W_1^{\name}) = 0$ and $\C(W_2^{\name}) = \C(W_3^{\name}) = 2$.}
\label{fig:cfc-examples}
\end{figure}

\begin{description}
\propitemc{\propone}{\yes}{Car05} 
For the two nets $W_1^{\name}$ and $W_2^{\name}$ of Figure~\ref{fig:cfc-examples}, we get the complexity scores $\C(W_1^{\name}) = 0 \neq 2 = \C(W_2^{\name})$.

\propitemc{\proptwo}{\no}{Car05} 
Let $c \in \mathbb{N}_0$ be an arbitrary integer. 
Figure~\ref{fig:CFC-fin} shows how to construct infinitely many workflow nets with complexity score $c$ according to the complexity measure $\C$.
\begin{figure}[ht]
\begin{center}
\scalebox{\scalefactor}{
\begin{tikzpicture}[node distance = 1.5cm,>=stealth',bend angle=0,auto]
	\node [place,tokens=1] (start) [label=below:$p_i$] {};
	\node [transition] (tk) [right of=start,label=below:$t_{c-1}$] {}
	edge [pre] (start);
	\node [place] (p1) [above right of=tk,label=below:$p^1$] {}
	edge [pre] (tk);
	\node [place] (pc) [below right of=tk,label=below:$p^{c+1}$] {}
	edge [pre] (tk);
	\node at ($0.5*(p1) + 0.5*(pc)$) {$\vdots$};
	\node [transition] (t1) [right of=p1,label=below:$t^1$] {}
	edge [pre] (p1);
	\node [transition] (tc) [right of=pc,label=below:$t^{k}$] {}
	edge [pre] (pc);
	\node at ($0.5*(t1) + 0.5*(tc)$) {$\vdots$};
	\node [place] (q1) [right of=t1,label=below:$q^1$] {}
	edge [pre] (t1);
	\node [place] (qc) [right of=tc,label=below:$q^{c+1}$] {}
	edge [pre] (tc);
	\node at ($0.5*(q1) + 0.5*(qc)$) {$\vdots$};
	\node [transition] (t) [below right of=q1,label=below:$t_c$] {}
	edge [pre] (q1)
	edge [pre] (qc);
	\node [place] (end) [right of=t,label=below:$p_o$] {}
	edge [pre] (t);
	\node [transition] (dc2) [below of=tc,label=below:$t_{c-2}$] {}
	edge [pre,bend left=20] (start)
	edge [post,bend right=20] (end);
	\node [transition] (d0) [below of=dc2,label=below:$t_0$] {}
	edge [pre,bend left=20] (start)
	edge [post,bend right=20] (end);
	\node at ($0.5*(dc2) + 0.5*(d0)$) {$\vdots$};
	\node at (0,1.15) {$W_{c,k}^{\name}$:};
\end{tikzpicture}}
\end{center}
\caption{A construction plan for workflow nets $W_{c,k}^{\name}$ with $k + c$ transitions and $\C(W_{c,k}^{\name}) = c$ for any $c,k \in \mathbb{N}_0$.}
\label{fig:CFC-fin}
\end{figure}

\propitemc{\propthree}{\yes}{Car05} 
The workflow nets $W_2^{\name}$ and $W_3^{\name}$ of Figure~\ref{fig:cfc-examples} are different in structure, but both have complexity $\C(W_2^{\name}) = 2 = \C(W_3^{\name})$.

\propitemc{\propfour}{\yes}{Car05} 
The two workflow nets $W_1^{\name}$ and $W_2^{\name}$ of Figure~\ref{fig:cfc-examples} have the same language $L(W_1^{\name}) = \{\varepsilon, a\} = L(W_2^{\name})$, but receive complexity scores $\C(W_1^{\name}) = 0 \neq 2 = \C(W_2^{\name})$.

\propitemc{\propfive}{\yes}{Car05} 
To show that $\C$ is monotone, we prove the following:
\begin{theorem}
\label{thm:cfc-mon}
Let $n \geq 2$. For any workflow nets $M_1, \dots, M_n \in \mathcal{M}$ we get:
\begin{itemize}
\item $\C(\seqop(M_1, \dots, M_n)) = \C(M_1) + \dots + \C(M_n)$,
\item $\C(\parop(M_1, \dots, M_n)) = \C(M_1) + \dots + \C(M_n) + 1$,
\item $\C(\choiceop(M_1, \dots, M_n)) = \C(M_1) + \dots + \C(M_n) + n$ and
\item $\C(\loopop(M_1, \dots, M_n)) = \C(M_1) + \dots + \C(M_n) + n$.
\end{itemize}
\end{theorem}
\begin{proof}
None of the operations changes the pre- or postset of existing places or transitions in $M_1, \dots, M_n$, except for the input and output place.
Furthermore, the operation $\seqop$ doesn't introduce new and- or xor-splits, so $\C(\seqop(M_1, \dots, M_n)) = \C(M_1) + \dots + \C(M_n)$.
The operation $\parop$ adds one and-split, which increases the complexity by $1$. 
Operations $\choiceop$ and $\loopop$ each introduce one new xor-split that has $n$ outgoing edges, thus increasing the complexity score by $n$. \hfill $\square$
\end{proof}
Thus, according to Theorem~\ref{thm:cfc-mon}, $\C(\oplus(M_1, \dots, M_n)) \geq \C(M_i)$ for any workflow nets $M_1, \dots, M_n \in \mathcal{M}$, any operation $\oplus \in \{\seqop, \parop, \choiceop, \loopop\}$ and a model $M_i \in \{M_1, \dots, M_n\}$.

\propitemc{\propsix}{\no}{Car05} 
Let $M_1, M_2, M_3 \in \mathcal{M}$ with $\C(M_1) = \C(M_2)$. Theorem~\ref{thm:cfc-mon} gives us the following equalities:
\begin{itemize}
\item $\C(M_1 \seqop M_3) = \C(M_1) + \C(M_3)$ \\
\phantom{$\C(M_1 \seqop M_3)$ }$= \C(M_2) + \C(M_3) = \C(M_1 \seqop M_3)$,
\item $\C(M_1 \parop M_3) = \C(M_1) + \C(M_3) + 1$ \\
\phantom{$\C(M_1 \parop M_3)$ }$= \C(M_2) + \C(M_3) + 1 = \C(M_1 \parop M_3)$,
\item $\C(M_1 \choiceop M_3) = \C(M_1) + \C(M_3) + n$ \\
\phantom{$\C(M_1 \choiceop M_3)$ }$= \C(M_2) + \C(M_3) + n = \C(M_1 \choiceop M_3)$,
\item $\C(M_1 \loopop M_3) = \C(M_1) + \C(M_3) + n$ \\
\phantom{$\C(M_1 \loopop M_3)$ }$= \C(M_2) + \C(M_3) + n = \C(M_1 \loopop M_3)$.
\end{itemize}

\propitemc{\propseven}{\yes}{Car05} 
Figure~\ref{fig:cfc-perm} shows two workflow nets, $W_4^{\name}$ and $W_5^{\name}$, that are permutations of each other, but $\C(W_4^{\name}) = 2 \neq 4 = \C(W_5^{\name})$.
\begin{figure}[ht]
\begin{center}
\begin{minipage}{0.55\textwidth}
\centering
\scalebox{\scalefactor}{
\begin{tikzpicture}[node distance = 1.5cm,>=stealth',bend angle=0,auto]
	\node [place,tokens=1] (start) [label=below:$p_i$] {};
	\node [transition] (t1) [right of=start,label=below:$t_1$] {}
	edge [pre] (start);
	\node [place] (p1) [above right of=t1,label=below:$p_1$] {}
	edge [pre] (t1);
	\node [place] (p2) [below right of=t1,label=below:$p_2$] {}
	edge [pre] (t1);
	\node [transition] (t2) [right of=p1,label=below:$t_2$] {}
	edge [pre] (p1);
	\node [transition] (t3) [right of=p2,label=below:$t_3$] {}
	edge [pre] (p2);
	\node [place] (p3) [right of=t2,label=below:$p_3$] {}
	edge [pre] (t2);
	\node [place] (p4) [right of=t3,label=below:$p_4$] {}
	edge [pre] (t3);
	\node [transition] (t4) [below right of=p3,label=below:$t_4$] {}
	edge [pre] (p3)
	edge [pre] (p4);
	\node [place] (end) [right of=t4,label=below:$p_o$] {}
	edge [pre] (t4);
	\node at (0,1.15) {$W_4^{\name}$:};
	\draw[opacity=0] ($(start)-(0.25,1.75)$) rectangle ($(end) + (0.25,1.75)$);
\end{tikzpicture}}
\end{minipage}
\begin{minipage}{0.4\textwidth}
\centering
\scalebox{\scalefactor}{
\begin{tikzpicture}[node distance = 1.5cm,>=stealth',bend angle=0,auto]
	\node [place,tokens=1] (start) [label=below:$p_i$] {};
	\node [transition] (t1) [above right of=start,label=below:$t_1$] {}
	edge [pre] (start);
	\node [transition] (t2) [below right of=start,label=below:$t_2$] {}
	edge [pre] (start);
	\node [place] (p2) [right of=t1,yshift=-0.5cm,label=below:$p_2$] {}
	edge [pre] (t1);
	\node [place] (p1) at ($(p2) + (0,1)$) [label=below:$p_1$] {}
	edge [pre] (t1);
	\node [place] (p3) [right of=t2,yshift=0.5cm,label=below:$p_3$] {}
	edge [pre] (t2);
	\node [place] (p4) at ($(p3) - (0,1)$) [label=below:$p_4$] {}
	edge [pre] (t2);
	\node [transition] (t3) [right of=p2,yshift=0.5cm,label=below:$t_3$] {}
	edge [pre] (p1)
	edge [pre] (p2);
	\node [transition] (t4) [right of=p3,yshift=-0.5cm,label=below:$t_4$] {}
	edge [pre] (p3)
	edge [pre] (p4);
	\node [place] (end) [below right of=t3,label=below:$p_o$] {}
	edge [pre] (t3)
	edge [pre] (t4);
	\node at (0,1.15) {$W_5^{\name}$:};
	\draw[opacity=0] ($(start)-(0.25,2.25)$) rectangle ($(end) + (0.25,2.25)$);
\end{tikzpicture}}
\end{minipage}
\end{center}
\caption{Two workflow nets, $W_4^{\name}$ and $W_5^{\name}$, with $W_5^{\name} \in Perm(W_4^{\name})$, $\C(W_4^{\name}) = 2$ and $\C(W_5^{\name}) = 4$.}
\label{fig:cfc-perm}
\end{figure}

\propitemc{\propeight}{\yes}{Car05} 
A uniform relabelling of the transitions has no effect on the set of split nodes or their out-degree. 
Therefore, uniformly changing the transition labels has no impact on the complexity score assigned by $\C$.

\propitemc{\propnine}{\no}{Car05} 
Theorem~\ref{thm:cfc-mon} implies $\C(M_1 \oplus M_2) > \C(M_1) + \C(M_2)$ for all $M_1, M_2 \in \mathcal{M}$ and $\oplus \in \{\parop, \choiceop, \loopop\}$. 
For the operation $\seqop$ however, $\C$ fulfills only a weaker property: $\C(M_1 \seqop M_2) \geq \C(M_1) + \C(M_2)$ for all $M_1, M_2 \in \mathcal{M}$.

\propitemf{\propdef}{\yes} 
$S_{\text{and}}^W$ and $S_{\text{xor}}^W$ are defined for each workflow net, so $\C$ is also defined for each workflow net.
Furthermore, $S_{\text{and}}^W$, $S_{\text{xor}}^W$ and $\post{p}$ for any $p \in P$ are sets, so their cardinality can't get below $0$.
Thus, $\C(W) \geq 0$ for any workflow net $W$.

\propitemf{\propmin}{\yes} 
The smallest possible value $\C$ can return is $0$.
The workflow net $\netminwf$ of Figure~\ref{fig:netminwf} is one example of a workflow net that gets this complexity score, since it doesn't have any connector nodes.

\propitemf{\propinf}{\yes} 
Let $c \in \mathbb{N}_0$. Figure~\ref{fig:CFC-fin} shows how to construct workflow nets with complexity $c$.
Therefore, $|\{c \in \mathbb{R} \mid \exists M \in \mathcal{M}: \C(M) = c\}| \geq |\mathbb{N}_0| = \infty$.

\propitemf{\propnotsup}{\no} Theorem~\ref{thm:cfc-mon} implies that $\C$ does not fulfill this property for any of the operations $\oplus \in \{\seqop, \parop, \choiceop, \loopop\}$.

\propitemf{\propadd}{\no} Theorem~\ref{thm:cfc-mon} implies that $\C$ fulfills this property only for the operation $\seqop$, but not for the other operations of Definition~\ref{def:operations}.
\end{description}

\subsection{Node IO Complexity}
The second complexity dimension found by Lieben et al.~\cite{LieDJJ18} is labeled \textsc{Node IO Complexity}.
Measures of this dimension take the number of incoming and outgoing sequence flows into account.
Apart from the complexity measures we discuss in this section, the Coefficient of Network Connectivity~\cite[p.120]{Men08} also belongs to this dimension.
This measure has a special role, as it belongs to two of the four dimensions.
We therefore analyze this measure within the section of the fourth dimension to keep the dimensions equally large.

\def\name{\maxconnname}
\def\C{\maxconn}
\newcommand{\compmaxconn}{C_{\name}}
\newcommand{\wfmax}[1]{W_{#1}^{\name}}
\subsubsection{Maximum Connector Degree}
The maximum connector degree is the maximum number of arcs that enter or leave a single connector node in the net~\cite[p121]{Men08}.
Connectors with a high degree introduce complexity, since they either increase the number of possible control flows or the amount of tokens to keep track of.
\begin{equation}
\label{eq:max-conn-deg}
\compmaxconn(W) = \max\{\ndeg(x) \mid x \in \mathcal{C}^W\}
\end{equation}
Figure~\ref{fig:mcd-examples} shows three example nets and their respective complexity scores.
\begin{figure}
\centering
\begin{minipage}{0.25\textwidth}
\centering
\scalebox{\scalefactor}{
\begin{tikzpicture}[node distance = 1.5cm,>=stealth',bend angle=0,auto]
    \node [place,tokens=1] (start) [label=below:$p_i$] {};
    \node [transition] (t2) [right of=start,label=center:$a$] {}
    edge [pre] (start);
    \node [transition] (t1) [above of=t2,label=center:$a$] {}
    edge [pre] (start);
    \node [place] (p1) [right of=t2,label=below:$p_o$] {}
    edge [pre] (t1)
    edge [pre] (t2);
    \node at (0,1.15) {$\wfmax{1}$:};
    \draw[opacity=0] ($(start)-(0.25,1.75)$) rectangle ($(p1) + (0.25,1.75)$);
\end{tikzpicture}}
\end{minipage}
\begin{minipage}{0.25\textwidth}
\centering
\scalebox{\scalefactor}{
\begin{tikzpicture}[node distance = 1.5cm,>=stealth',bend angle=0,auto]
    \node [place,tokens=1] (start) [label=below:$p_i$] {};
    \node [transition] (t2) [right of=start,label=center:$a$] {}
    edge [pre] (start);
    \node [transition] (t1) [above of=t2,label=center:$a$] {}
    edge [pre] (start);
    \node [transition] (t3) [below of=t2,label=center:$a$] {}
    edge [pre] (start);
    \node [place] (p1) [right of=t2,label=below:$p_o$] {}
    edge [pre] (t1)
    edge [pre] (t2)
    edge [pre] (t3);
    \node at (0,1.15) {$\wfmax{2}$:};
    \draw[opacity=0] ($(start)-(0.25,1.75)$) rectangle ($(p1) + (0.25,1.75)$);
\end{tikzpicture}}
\end{minipage}
\begin{minipage}{0.4\textwidth}
\centering
\scalebox{\scalefactor}{
\begin{tikzpicture}[node distance = 1.5cm,>=stealth',bend angle=0,auto]
    \node [place,tokens=1] (start) [label=below:$p_i$] {};
    \node [transition] (t1) [right of=start,label=center:$a$] {}
    edge [pre] (start);
    \node [place] (p1) [above right of=t1] {}
    edge [pre] (t1);
    \node [place] (p2) [below right of=t1] {}
    edge [pre] (t1);
    \node [transition] (t2) [below right of=p1,label=center:$b$] {}
    edge [pre] (p1)
    edge [pre] (p2);
    \node [place] (po) [right of=t2,label=below:$p_o$] {}
    edge [pre] (t2);
    \node at (0,1.15) {$\wfmax{3}$:};
    \draw[opacity=0] ($(start)-(0.25,1.75)$) rectangle ($(po) + (0.25,1.75)$);
\end{tikzpicture}}
\end{minipage}
\caption{Example nets $\wfmax{1}$, $\wfmax{2}$ and $\wfmax{3}$, with $\C(\wfmax{1}) = 2$ and $\C(\wfmax{2}) = \C(\wfmax{3}) = 3$.}
\label{fig:mcd-examples}
\end{figure}

\begin{description}
\propitemf{\propone}{\yes} 
For the two nets $\wfmax{1}$ and $\wfmax{2}$ of Figure~\ref{fig:mcd-examples}, we get the complexity scores $\C(\wfmax{1}) = 2 \neq 3 = \C(\wfmax{2})$.

\propitemf{\proptwo}{\no} Follows directly by Theorem~\ref{thm:conn-inf}.
Figure~\ref{fig:mcd-fin} shows how to construct such a sequence of places and transitions.
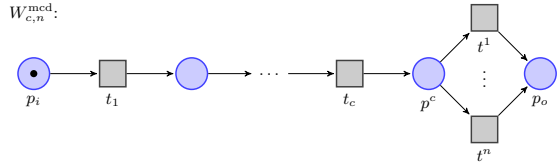
\begin{figure}[ht]
\begin{center}
\scalebox{\scalefactor}{
\begin{tikzpicture}[node distance = 1.5cm,>=stealth',bend angle=0,auto]
    \node [place,tokens=1] (start) [label=below:$p_i$] {};
    \node [transition] (t1) [right of=start,label=below:$t_1$] {}
    edge [pre] (start);
    \node [place] (p1) [right of=t1] {}
    edge [pre] (t1);
    \node (dots) [right of=p1] {$\dots$}
    edge [pre] (p1);
    \node [transition] (tn) [right of=dots,label=below:$t_c$] {}
    edge [pre] (dots);
    \node [place] (pk) [right of=tn,label=below:$p^{c}$] {}
    edge [pre] (tn);
    \node [transition] (t1) [above right of=pk,label=below:$t^1$] {}
    edge [pre] (pk);
    \node [transition] (tk) [below right of=pk,label=below:$t^{n}$] {}
    edge [pre] (pk);
    \node [place] (end) [below right of=t1,label=below:$p_o$] {}
    edge [pre] (t1)
    edge [pre] (tk);
    \node at ($0.5*(t1) + 0.5*(tk)$) {$\vdots$};
    \node at (0,1.15) {$W_{c,n}^{\name}$:};
\end{tikzpicture}}
\end{center}
\caption{A workflow net $W_{c,n}^{\name}$ with $c + n$ transitions, $c + 2$ places and complexity $\C(W_{c,n}^{\name}) = n + 1$.}
\label{fig:mcd-fin}
\end{figure}

\propitemf{\propthree}{\yes} 
The workflow nets $\wfmax{2}$ and $\wfmax{3}$ of Figure~\ref{fig:mcd-examples} are different in structure, but have equal complexity: $\compmaxconn(\wfmax{2}) = 3 = \compmaxconn(\wfmax{3})$.

\propitemf{\propfour}{\yes}
The workflow nets $\wfmax{1}$ and $\wfmax{2}$ of Figure~\ref{fig:mcd-examples} have the same language $L(\wfmax{1}) = \{\varepsilon, a\} = L(\wfmax{2})$, but receive different complexity scores $\C(\wfmax{1}) = 2 \neq 3 = \C(\wfmax{3})$.

\propitemf{\propfive}{\yes} 
Let $M_1, \dots, M_n \in \mathcal{M}$ be workflow nets and let $\oplus \in \{\seqop, \parop, \choiceop, \loopop\}$.
Further, let $M = \oplus(M_1, \dots, M_n)$.
Then, there are only two cases:
\begin{enumerate}
\item $M$ contains the connector with the highest degree. 
Then, obviously $\C(M) = \C(\oplus(M_1, \dots, M_n)) \geq \C(M_i)$ for any of the composed models $M_i \in \{M_1, \dots, M_n\}$.
\item A model $M_j \in \{M_1, \dots, M_n\}$ contains the connector with the highest degree.
Let $\C(M_j) = d$ and $v$ be a connector with this degree in $M_j$.
Then, in $M$, the node $v$ has either degree $d$ or degree $d + 1$, if a new arc to $v$ was added with the operation. 
Since by definition $d \geq \C(M_i)$ for all $M_i \in \{M_1, \dots, M_n\}$, we get 
\[\C(\oplus(M_1, \dots, M_n)) \geq d = \C(M_j) \geq \C(M_i).\]
\end{enumerate}

\propitemf{\propsix}{\yes}
Take the workflow nets $\wfmax{2}$ and $\wfmax{3}$ of Figure~\ref{fig:mcd-examples}, as well as $\netminwf$ of Figure~\ref{fig:netminwf}.
We know that $\C(\wfmax{2}) = 3 = \C(\wfmax{3})$, but regardless of which operation $\oplus \in \{\seqop, \parop, \loopop, \choiceop\}$ we use, we get the result $\C(\wfmax{2} \oplus \netminwf) = 4 \neq 3 = \C(\wfmax{3} \oplus \netminwf)$.

\propitemf{\propseven}{\yes} 
Figure~\ref{fig:mcd-perm} shows two workflow nets, $\wfmax{4}$ and $\wfmax{5}$, that are permutations of each other, but $\C(\wfmax{4}) = 3 \neq 4 = \C(\wfmax{5})$.
\begin{figure}
\centering
\begin{minipage}{0.7\textwidth}
\centering
\scalebox{\scalefactor}{
\begin{tikzpicture}[node distance = 1.5cm,>=stealth',bend angle=0,auto]
    \node [place,tokens=1] (start) [label=below:$p_i$] {};
    \node [transition] (t2) [right of=start,label=center:$b$] {}
    edge [pre] (start);
    \node [transition] (t1) [above of=t2,label=center:$a$] {}
    edge [pre] (start);
	\node [place] (p1) [right of=t2,label=below:$p_1$] {}
    edge [pre] (t1)
    edge [pre] (t2);
	\node [transition] (t5) [right of=p1, label=center:$c$] {}
	edge [pre] (p1);
	\node [place] (p2) [right of=t5,label=below:$p_2$] {}
	edge [pre] (t5);
    \node [transition] (t4) [right of=p2,label=center:$e$] {}
    edge [pre] (p2);
    \node [transition] (t3) [above of=t4,label=center:$d$] {}
    edge [pre] (p2);
    \node [place] (po) [right of=t4,label=below:$p_o$] {}
    edge [pre] (t3)
    edge [pre] (t4);
    \node at (0,1.15) {$\wfmax{4}$:};
    \draw[opacity=0] ($(start)-(0.25,1.75)$) rectangle ($(po) + (0.25,1.75)$);
\end{tikzpicture}}
\end{minipage}
\\
\begin{minipage}{0.7\textwidth}
\centering
\scalebox{\scalefactor}{
\begin{tikzpicture}[node distance = 1.5cm,>=stealth',bend angle=0,auto]
    \node [place,tokens=1] (start) [label=below:$p_i$] {};
    \node [transition] (t1) [right of=start,label=center:$a$] {}
    edge [pre] (start);
	\node [place] (p1) [right of=t1,label=below:$p_1$] {}
    edge [pre] (t1);
	\node [transition] (t5) [right of=p1,label=center:$b$] {}
    edge [pre] (p1);
	\node [place] (p2) [right of=t5,label=below:$p_2$] {}
    edge [pre] (t5);
    \node [transition] (t4) [right of=p2,label=center:$d$] {}
    edge [pre] (p2);
    \node [transition] (t3) [above of=t4,label=center:$c$] {}
    edge [pre] (p2);
    \node [transition] (t2) [below of=t4,label=center:$e$] {}
    edge [pre] (p2);
    \node [place] (po) [right of=t4,label=below:$p_o$] {}
	edge [pre] (t2)
    edge [pre] (t3)
    edge [pre] (t4);
    \node at (0,1.15) {$\wfmax{5}$:};
    \draw[opacity=0] ($(start)-(0.25,1.75)$) rectangle ($(po) + (0.25,1.75)$);
\end{tikzpicture}}
\end{minipage}
\caption{Two workflow nets, $\wfmax{4}$ and $\wfmax{5}$, with $\wfmax{4} \in Perm(\wfmax{5})$, but $\C(\wfmax{4}) = 3$ and $\C(\wfmax{5}) = 4$.}
\label{fig:mcd-perm}
\end{figure}

\propitemf{\propeight}{\yes} 
The $\compmaxconn$ measure is robust against relabellings, as can be seen straighforwardly by its insensitivity to transition labels generally. 

\propitemf{\propnine}{\no} 
If we consider only input workflow nets where  $\compmaxconn$ is defined, and so can assume that all input workflow nets have at least one connector, $\compmaxconn$ is subadditive for the operations $\seqop$, $\parop$, $\choiceop$ and $\loopop$. 

Let $M_1, M_2 \in \mathcal{M}$ be two workflow nets. 
We know that $\C(M_1) \geq 2$ and $\C(M_2) \geq 2$, since a node only qualifies as a connector if it has at least two outgoing or incoming arcs.
Let $c_1 := \C(M_1)$ and $c_2 := \C(M_2)$. 
Then, the maximum connector degree in $M_1 \oplus M_2$ for some $\oplus \in \{\seqop, \parop, \choiceop, \loopop\}$ can be $\max{c_1, c_2} + 1$ if the connector with maximum degree is an initial or final node in its net.
If both $M_1$ and $M_2$ have complexity 2, we get $\C(M_1 \oplus M_2) = 3 < 4 = \C(M_1) + \C(M_2)$.
Otherwise, we get $\C(M_1 \oplus M_2) \leq \max\{c_1,c_2\} + 1 < 2 \cdot \max\{c_1, c_2\} = \C(M_1) + \C(M_2)$.

If the metric is redefined to give nets with no connectors a complexity of zero, subadditivity is contradicted for $\oplus \in \{\parop, \choiceop, \loopop\}$:
For the net $\netminwf$ of Figure~\ref{fig:netminwf}, we get $\C(\netminwf \oplus \netminwf) \geq 2 > 0 = \C(\netminwf) + \C(\netminwf$) for these operations.
For the operation $\seqop$, subadditivity still holds in this case.

\propitemf{\propdef}{\no} 
The maximum connector degree is undefined for workflow nets without any connectors. 
However, we can add a special case where $\compmaxconn$ returns the complexity score $0$ if there are no connectors in the net. 
With this adaptation, $\compmaxconn$ returns only non-negative values, since the degree of a node is always non-negative.

\propitemf{\propmin}{\yes} 
With the adaptation shown for in the analysis of {\propdef}, the smallest workflow net, $\netminwf$ of Figure~\ref{fig:netminwf}, has two places and one transition. 
None of them are connectors, making the minimum complexity zero.
For workflow nets that contain at least one connector node, the maximum connector degree is $\geq 2$, since a connector node needs at least two incoming or outgoing arcs.

\propitemf{\propinf}{\yes} 
Let $n \in \mathbb{N}$ with $n \geq 2$. 
Figure~\ref{fig:mcd-fin} shows how to construct a workflow net with complexity score $n + 1$, so $\C$ can return infinitely many values.
We get $|\{c \in \mathbb{R} \mid \exists M \in \mathcal{M}: \C(M) = c\}| \geq \{n \in \mathbb{N} \mid n \geq 3\} = \infty$.

\propitemf{\propnotsup}{\yes} The discussion for \propnine{} shows that $\C$ is not superadditive.

\propitemf{\propadd}{\no} The discussion for \propnine{} shows $\C$ is not additive.
\end{description}

\def\name{\sequentialityname}
\def\C{\sequentiality}
\newcommand{\compseq}{C_{\sequentialityname}}
\newcommand{\wfseq}[1]{W_{#1}^{\sequentialityname}}
\newcommand{\wfseqinf}{W_{c,k}^{\sequentialityname}}
\subsubsection{Sequentiality}
Naturally, sequences of transitions are easier to understand than concurrent or exclusive execution. 
Therefore, Mendling~\cite[p123]{Men08} defines a complexity measure for sequentiality, which is the number of arcs between non-connector nodes divided by the total number of arcs.
The higher the score of this metric, the lower the amount of direct sequences in the net.
In our version of this complexity metric, we therefore subtract this fraction from $1$ to ensure comparability with the other complexity measures.
\begin{equation}
\label{eq:seq}
\compseq(W) = 1 - \frac{|\{(x,y) \in F \mid x,y \not\in \mathcal{C}^W\}|}{|F|}
\end{equation}
Figure~\ref{fig:seq-examples} shows two example nets and their respective complexity scores.
\begin{figure}
\centering
\begin{minipage}{0.4\textwidth}
\centering
\scalebox{\scalefactor}{
\begin{tikzpicture}[node distance = 1.5cm,>=stealth',bend angle=0,auto]
    \node [place,tokens=1] (start) [label=below:$p_i$] {};
    \node [transition] (t1) [right of=start,label=center:$a$] {}
    edge [pre] (start);
    \node [place] (p1) [right of=t1] {}
    edge [pre] (t1);
    \node [transition] (t2) [above right of=p1,label=center:$b$] {}
    edge [pre] (p1);
    \node [transition] (t3) [below right of=p1,label=center:$c$] {}
    edge [pre] (p1);
    \node [place] (po) [below right of=t2,label=below:$p_o$] {}
	edge [pre] (t2)
    edge [pre] (t3);
    \node at (0,1.15) {$\wfseq{1}$:};
    \draw[opacity=0] ($(start)-(0.25,1.75)$) rectangle ($(po) + (0.25,1.75)$);
\end{tikzpicture}}
\end{minipage}
\begin{minipage}{0.4\textwidth}
\centering
\scalebox{\scalefactor}{
\begin{tikzpicture}[node distance = 1.5cm,>=stealth',bend angle=0,auto]
    \node [place,tokens=1] (start) [label=below:$p_i$] {};
    \node [transition] (t1) [above right of=start,label=center:$a$] {}
    edge [pre] (start);
    \node [place] (p1) [right of=t1] {}
    edge [pre] (t1);
    \node [transition] (t2) [right of=p1,label=center:$b$] {}
    edge [pre] (p1);
    \node [transition] (t3) [below right of=start,label=center:$a$] {}
    edge [pre] (start);
    \node [place] (p2) [right of=t3] {}
    edge [pre] (t3);
    \node [transition] (t4) [right of=p2,label=center:$c$] {}
    edge [pre] (p2);
    \node [place] (po) [below right of=t2,label=below:$p_o$] {}
	edge [pre] (t2)
    edge [pre] (t4);
    \node at (0,1.15) {$\wfseq{2}$:};
    \draw[opacity=0] ($(start)-(0.25,1.75)$) rectangle ($(po) + (0.25,1.75)$);
\end{tikzpicture}}
\end{minipage}
\caption{Example nets with $\C(\wfseq{1}) = \frac{5}{6}$ and $\C(\wfseq{2}) = \frac{4}{8} = \frac{1}{2}$.}
\label{fig:seq-examples}
\end{figure}

\begin{description}
\propitemf{\propone}{\yes} 
For the two nets $\wfseq{1}$ and $\wfseq{2}$ of Figure~\ref{fig:seq-examples}, we get the complexity scores $\C(\wfseq{1}) = \frac{5}{6} \neq \frac{1}{2} = \C(\wfseq{2})$.

\propitemf{\proptwo}{\no} 
The net $\wfseqinf$ in Figure~\ref{fig:seq-fin} is a choice between $c$ sequences containing $2k - 1$ arcs. 
In total, it has $|F| = 2ck + c$ arcs.
Of those, $2c$ have at least one connector as an endpoint, so $\C(\wfseqinf) = \frac{2c}{2ck+c} = \frac{2}{2k + 1}$. 
As this complexity score is independent of $c$, we have found infinitely many workflow nets with this complexity score.
\begin{figure}[ht]
\begin{center}
\scalebox{\scalefactor}{
\begin{tikzpicture}[node distance = 1.5cm,>=stealth',bend angle=0,auto]
    \node [place,tokens=1] (start) [label=below:$p_i$] {};
    \node [transition] (t1) [above right of=start,label=below:$t_1^1$] {}
    edge [pre] (start);
    \node [transition] (tc) [below right of=start,label=below:$t_c^1$] {}
    edge [pre] (start);
    \node [place] (p1) [right of=t1,label=below:$p_1^1$] {}
    edge [pre] (t1);
    \node (dots1) [right of=p1] {$\dots$}
    edge [pre] (p1);
    \node [place] (pc) [right of=tc,label=below:$p_c^1$] {}
    edge [pre] (tc);
    \node (dotsc) [right of=pc] {$\dots$}
    edge [pre] (pc);
    \node [transition] (tk1) [right of=dots1,label=below:$t_1^k$] {}
    edge [pre] (dots1);
    \node [transition] (tkc) [right of=dotsc,label=below:$t_c^k$] {}
    edge [pre] (dotsc);
    \node [place] (end) [below right of=tk1,label=below:$p_o$] {}
    edge [pre] (tk1)
    edge [pre] (tkc);
    \node at ($0.5*(t1) + 0.5*(tc)$) {$\vdots$};
    \node at ($0.5*(tk1) + 0.5*(tkc)$) {$\vdots$};
    \node at (0,1.15) {$\wfseqinf$:};
\end{tikzpicture}}
\end{center}
\caption{A workflow net $\wfseqinf$ with complexity $\C(\wfseqinf) = \frac{2c}{2ck+c} = \frac{2}{2k + 1}$.}
\label{fig:seq-fin}
\end{figure}

\propitemf{\propthree}{\yes} 
The workflow nets $W_{2,3}^{\sequentialityname}$ and $W_{3,3}^{\sequentialityname}$ of Figure~\ref{fig:seq-fin} are different in structure, but have equal complexity: $\C(W_{2,3}^{\sequentialityname}) = \frac{2}{7} = \C(W_{3,3}^{\sequentialityname})$.

\propitemf{\propfour}{\yes} 
The workflow nets $\wfseq{1}$ and $\wfseq{2}$ of Figure~\ref{fig:seq-examples} have the same language $L(\wfseq{1}) = \{\varepsilon, a, ab, ac\} = L(\wfseq{2})$, but receive different complexity scores $\compseq(\wfseq{1}) = \frac{5}{6} \neq \frac{1}{2} = \compseq(\wfseq{2})$.

\propitemf{\propfive}{\no} 
Take the workflow nets $\wfseq{4}$ and $\netminwf$ of Figure~\ref{fig:seq-mon}.
The same Figure already shows that $\C$ is not monotone for $\parop$, but using the same nets we also get:
\begin{itemize}
\item $\C(\wfseq{4} \seqop \netminwf) = \frac{1}{2} < 1 = \C(\wfseq{4})$,
\item $\C(\wfseq{4} \parop \netminwf) = \frac{5}{6} < 1 = \C(\wfseq{4})$,
\item $\C(\wfseq{4} \choiceop \netminwf) = \frac{5}{7} < 1 = \C(\wfseq{4})$,
\item $\C(\wfseq{4} \loopop \netminwf) = \frac{2}{3} < 1 = \C(\wfseq{4})$.
\end{itemize}
\begin{figure}
    \centering
\scalebox{\scalefactor}{
\begin{tikzpicture}[node distance = 1.5cm,>=stealth',bend angle=0,auto]
    \node[place] (start) at (0,0) [] {};
    \node[transition, right of=start] (t1) [] {}
    edge [pre] (start);
    \node[place] (p1) [above right of=t1] {}
    edge [pre] (t1);
    \node[place] (p2) [below right of=t1] {}
    edge [pre] (t1);
    \node[transition, right of=p1] (t2) {}
    edge [pre] (p1);
    \node[transition, above right of=p2,xshift=0.44cm] (t3) {}
    edge [pre] (p2);
    \node[transition, below right of=p2,xshift=0.44cm] (t4) {}
    edge [pre] (p2);
    \node[place] (p3) [right of=t2] {}
    edge [pre] (t2);
    \node[place] (p4) [below right of=t3,xshift=0.44cm] {}
    edge [pre] (t3)
    edge [pre] (t4);
    \node[transition, below right of=p3] (t5) {}
    edge [pre] (p3)
    edge [pre] (p4);
    \node[place] (end) [right of=t5] {}
    edge [pre] (t5);
    
	\node at (0,1.5) {$\wfseq{3}$:};    
    
    \begin{pgfonlayer}{bg}
		\draw[rounded corners,draw=gray,fill=lightgray!50!white] ($(p1) - (0.75, 0.5)$) rectangle ($(p3) + (0.75, 1)$);
		\node[gray] at ($(p1) + (-0.4,0.75)$) {$\netminwf$};
		\draw[rounded corners,draw=gray,fill=lightgray!50!white] ($(p2) - (0.75, 1.5)$) rectangle ($(p4) + (0.75, 1.5)$);
		\node[gray] at ($(p2) + (-0.4,-1.25)$) {$\wfseq{4}$};
	\end{pgfonlayer}
\end{tikzpicture}}
\caption{A workflow net $\wfseq{3} = \wfseq{4} \parop \netminwf$ with $\compseq(\wfseq{3}) = 1 - \frac{2}{12} = \frac{5}{6}$, $\compseq(\netminwf) = 1 - \frac{2}{2} = 0$ and $\compseq(\wfseq{4}) = 1 - \frac{0}{4} = 1$.}
\label{fig:seq-mon}
\end{figure}

\propitemf{\propsix}{\yes} 
Sequentiality is sensitive to compositions.
As an example, Figure~\ref{fig:sequence-comp} shows nets $\wfseq{4}$ and $\wfseq{5}$ with $\C(\wfseq{4}) = 1 = \C(\wfseq{5})$, but for which we get:
\begin{itemize}
\item $\C(\wfseq{4} \seqop \netminwf) = \frac{5}{8} \neq \frac{9}{13} = \C(\wfseq{5} \seqop \netminwf)$,
\item $\C(\wfseq{4} \parop \netminwf) = \frac{5}{6} \neq \frac{15}{17} = \C(\wfseq{5} \parop \netminwf)$,
\item $\C(\wfseq{4} \choiceop \netminwf) = \frac{5}{7} \neq \frac{14}{19} = \C(\wfseq{5} \choiceop \netminwf)$,
\item $\C(\wfseq{4} \loopop \netminwf) = \frac{2}{3} \neq \frac{16}{23} = \C(\wfseq{5} \seqop \netminwf)$.
\end{itemize}
\begin{figure}
\centering
\begin{minipage}{0.8\textwidth}
\centering
\scalebox{\scalefactor}{
\begin{tikzpicture}[node distance = 1.5cm,>=stealth',bend angle=0,auto]
    \node[place] (start) at (3,0) [] {};
    \node[transition, above right of=start] (t1) [] {}
    edge [pre] (start);
    \node[transition, below right of=start] (t2) [] {}
    edge [pre] (start);
    \node[place] (p1end) [below right of=t1] {}
    edge [pre] (t1)
	edge [pre] (t2);
	%
	\node[transition, right of=p1end] (tc1) {}
	edge [pre] (p1end);
	%
    \node[place] (pmstart) [right of=tc1] {}
	edge [pre] (tc1);
	\node[transition, right of=pmstart] (tm1) {}
	edge [pre] (pmstart);
    \node[place] (pmend) [right of=tm1] {}
	edge [pre] (tm1);
    
	\node at (0,1.8) {$\wfseq{4} \seqop \netminwf$:};    
    
    \begin{pgfonlayer}{bg}
		\draw[rounded corners,draw=gray,fill=lightgray!50!white] ($(start) - (0.75, 1.5)$) rectangle ($(p1end) + (0.75, 1.5)$);
		\node[gray] at ($(start) + (-0.4,-1.25)$) {$\wfseq{4}$};
		\draw[rounded corners,draw=gray,fill=lightgray!50!white] ($(pmstart) - (0.75, 1.5)$) rectangle ($(pmend) + (0.75, 1.5)$);
		\node[gray] at ($(pmstart) + (-0.4,-1.25)$) {$\netminwf$};
	\end{pgfonlayer}
\end{tikzpicture}}
\end{minipage}
\\
\begin{minipage}{0.8\textwidth}
\centering
\scalebox{\scalefactor}{
\begin{tikzpicture}[node distance = 1.5cm,>=stealth',bend angle=0,auto]
    \node[place] (start) at (0,0) [] {};
    \node[transition, above right of=start] (t1) [] {}
    edge [pre] (start);
    \node[transition, below right of=start] (t2) [] {}
    edge [pre] (start);
    \node[place] (p1) [right of=t1] {}
    edge [pre] (t1)
	edge [pre] (t2);
    \node[place] (p2) [right of=t2] {}
    edge [pre] (t1)
	edge [pre] (t2);
    \node[transition, below right of=p1] (t3) {}
	edge [pre] (p1)
    edge [pre] (p2);
    \node[place] (p1end) [right of=t3] {}
    edge [pre] (t3);
	%
	\node[transition, right of=p1end] (tc1) {}
	edge [pre] (p1end);
	%
    \node[place] (pmstart) [right of=tc1] {}
	edge [pre] (tc1);
	\node[transition, right of=pmstart] (tm1) {}
	edge [pre] (pmstart);
    \node[place] (pmend) [right of=tm1] {}
	edge [pre] (tm1);
    
	\node at (0,1.8) {$\wfseq{5} \seqop \netminwf$:};    
    
    \begin{pgfonlayer}{bg}
		\draw[rounded corners,draw=gray,fill=lightgray!50!white] ($(start) - (0.75, 1.5)$) rectangle ($(p1end) + (0.75, 1.5)$);
		\node[gray] at ($(start) + (-0.4,-1.25)$) {$\wfseq{5}$};
		\draw[rounded corners,draw=gray,fill=lightgray!50!white] ($(pmstart) - (0.75, 1.5)$) rectangle ($(pmend) + (0.75, 1.5)$);
		\node[gray] at ($(pmstart) + (-0.4,-1.25)$) {$\netminwf$};
	\end{pgfonlayer}
\end{tikzpicture}}
\end{minipage}
\caption{Example nets $\wfseq{4}$ and $\wfseq{5}$, each composed sequentially with $\netminwf$.}
\label{fig:sequence-comp}
\end{figure}

\propitemf{\propseven}{\yes} 
Figure~\ref{fig:seq-perm} shows two workflow nets, $\wfseq{6}$ and $\wfseq{7}$, that are permutations of each other, but $\C(\wfseq{6}) = 1 \neq \frac{7}{8} = \C(\wfseq{7})$.
\begin{figure}
\centering
\begin{minipage}{0.4\textwidth}
\centering
\scalebox{\scalefactor}{
\begin{tikzpicture}[node distance = 1.5cm,>=stealth',bend angle=0,auto]
    \node [place,tokens=1] (start) [label=below:$p_i$] {};
    \node [transition] (t2) [right of=start,label=center:$b$] {}
    edge [pre] (start);
    \node [transition] (t1) [above of=t2,label=center:$a$] {}
    edge [pre] (start);
	\node [place] (p1) [right of=t2,label=below:$p_1$] {}
    edge [pre] (t1)
    edge [pre] (t2);
    \node [transition] (t4) [right of=p1,label=center:$d$] {}
    edge [pre] (p1);
    \node [transition] (t3) [above of=t4,label=center:$c$] {}
    edge [pre] (p1);
    \node [place] (po) [right of=t4,label=below:$p_o$] {}
    edge [pre] (t3)
    edge [pre] (t4);
    \node at (0,1.15) {$\wfseq{6}$:};
    \draw[opacity=0] ($(start)-(0.25,1.75)$) rectangle ($(po) + (0.25,1.75)$);
\end{tikzpicture}}
\end{minipage}
\begin{minipage}{0.4\textwidth}
\centering
\scalebox{\scalefactor}{
\begin{tikzpicture}[node distance = 1.5cm,>=stealth',bend angle=0,auto]
    \node [place,tokens=1] (start) [label=below:$p_i$] {};
    \node [transition] (t2) [right of=start,label=center:$b$] {}
    edge [pre] (start);
    \node [transition] (t1) [right of=start,label=center:$a$] {}
    edge [pre] (start);
	\node [place] (p1) [right of=t2,label=below:$p_1$] {}
    edge [pre] (t1)
    edge [pre] (t2);
    \node [transition] (t4) [right of=p1,label=center:$c$] {}
    edge [pre] (p1);
    \node [transition] (t3) [below of=t4,label=center:$d$] {}
    edge [pre] (p1);
    \node [transition] (t2) [above of=t4, label=center:$b$] {}
    edge [pre] (p1);
    \node [place] (po) [right of=t4,label=below:$p_o$] {}
	edge [pre] (t2)
    edge [pre] (t3)
    edge [pre] (t4);
    \node at (0,1.15) {$\wfseq{7}$:};
    \draw[opacity=0] ($(start)-(0.25,1.75)$) rectangle ($(po) + (0.25,1.75)$);
\end{tikzpicture}}
\end{minipage}
\caption{Two workflow nets, $\wfseq{6}$ and $\wfseq{7}$, with $\wfseq{6} \in Perm(\wfseq{7})$, but $\C(\wfseq{6}) = 1$ and $\C(\wfseq{7}) = \frac{7}{8}$.}
\label{fig:seq-perm}
\end{figure}

\propitemf{\propeight}{\yes} 
The $\compseq$ measure is robust against relabellings, as can be seen straighforwardly by its insensitivity to transition labels generally. 

\propitemf{\propnine}{\yes} 
The smallest possible workflow net is $\netminwf$, and $\compseq(\netminwf) = 0$ as it has no connectors. 
We can see that $\compseq(\netminwf, \oplus \netminwf) > 0 = \compseq(\netminwf) + \compseq(\netminwf)$ for $\oplus \in \{\parop, \choiceop, \loopop\}$, since these operations introduce new connector nodes and thereby elevate the complexity of the net above $0$.

However, $\C$ is subadditive for the $\seqop$ operator. 
Let $M_1, M_2 \in \mathcal{M}$ be two nets with $M_1 = (P_1, T_1, F_1, \ell_1, p_i^1, p_o^1)$ and $M_2 = (P_2, T_2, F_2, \ell_2, p_i^2, p_o^2)$.
$\seqop$ introduces two new arcs, which can both have connector-endpoints, if $p_i^1$ and $p_i^2$ are both connectors.
Therefore, if we take $s_1$ as the number of sequential arcs in $M_1$ and $s_2$ the number of sequential arcs in $M_2$, we get:
\begin{align*}
\compseq(M_1 \seqop M_2)) \leq \frac{s_1 + s_2 + 2}{|F_1| + |F_2| + 2} \leq \frac{s_1 + s_2}{|F_1| + |F_2|} \leq \frac{s_1}{|F_1|} + \frac{s_2}{|F_2|}
\end{align*}
since $s_1 \leq |F_1|$ and $s_2 \leq |F_2|$. 
Thus, for the sequence operator, we get $\compseq(M_1 \seqop M_2) \leq \compseq(M_1) + \compseq(M_2)$.
The associativity of the $\seqop$ operation  means the claim also holds if we increase the number of workflow nets composed in sequence.

\propitemf{\propdef}{\yes} 
Workflow nets always have at least two arcs, so the denominator of $\compseq$ can never be $0$. 
Furthermore, $\compseq(W) \geq 0$ for all workflow nets $W$, as obviously $|\{(x,y) \in F \mid x,y \not\in C^W\}| \leq |F|$.

\propitemf{\propmin}{\yes} 
The minimum score of $\compseq$ is $0$, which can be reached for the smallest possible workflow net $\netminwf$.
In this net, there are no connector nodes, so all arcs go from a non-connector node to another non-connector node.

\propitemf{\propinf}{\yes} 
Figure~\ref{fig:seq-fin} shows how to construct infinitely many workflow nets of different complexity.

\propitemf{\propnotsup}{\yes}
The counterexamples for the property {\propfive} also show that $\C$ is not superadditive with respect to any of the operations.

\propitemf{\propadd}{\no}
The counterexamples for the property {\propfive} also show that $\C$ is not additive with respect to any of the operations.
\end{description}

\def\name{\avgconnname}
\def\C{\avgconn}
\newcommand{\compavgconn}{C_{\avgconnname}}
\newcommand{\wfavg}[1]{W_{#1}^{\avgconnname}}
\newcommand{\wfavginf}{\wfavg{4}}
\subsubsection{Average Connector Degree}
The average connector degree is the number of arcs entering or leaving a connector node divided by the total number of connectors~\cite[pp.120-121]{Men08}.
If this amount is high for a workflow net, it is less understandable than a workflow net with low average connector degree.
\begin{equation}
\label{eq:avg-conn-deg}
		\compavgconn(W) = 
			\frac{1}{|\mathcal{C}^W|} \cdot \sum_{x \in \mathcal{C}^W} \ndeg(x) 
\end{equation}
Figure~\ref{fig:acd-examples} shows three example nets and their respective complexity scores.
\begin{figure}
\centering
\begin{minipage}{0.25\textwidth}
\centering
\scalebox{\scalefactor}{
\begin{tikzpicture}[node distance = 1.5cm,>=stealth',bend angle=0,auto]
    \node [place,tokens=1] (start) [label=below:$p_i$] {};
    \node [transition] (t2) [right of=start,label=center:$a$] {}
    edge [pre] (start);
    \node [transition] (t1) [above of=t2,label=center:$a$] {}
    edge [pre] (start);
    \node [place] (p1) [right of=t2,label=below:$p_o$] {}
    edge [pre] (t1)
    edge [pre] (t2);
    \node at (0,1.15) {$\wfavg{1}$:};
    \draw[opacity=0] ($(start)-(0.25,1.75)$) rectangle ($(p1) + (0.25,1.75)$);
\end{tikzpicture}}
\end{minipage}
\begin{minipage}{0.25\textwidth}
\centering
\scalebox{\scalefactor}{
\begin{tikzpicture}[node distance = 1.5cm,>=stealth',bend angle=0,auto]
    \node [place,tokens=1] (start) [label=below:$p_i$] {};
    \node [transition] (t2) [right of=start,label=center:$a$] {}
    edge [pre] (start);
    \node [transition] (t1) [above of=t2,label=center:$a$] {}
    edge [pre] (start);
    \node [transition] (t3) [below of=t2,label=center:$a$] {}
    edge [pre] (start);
    \node [place] (p1) [right of=t2,label=below:$p_o$] {}
    edge [pre] (t1)
    edge [pre] (t2)
    edge [pre] (t3);
    \node at (0,1.15) {$\wfavg{2}$:};
    \draw[opacity=0] ($(start)-(0.25,1.75)$) rectangle ($(p1) + (0.25,1.75)$);
\end{tikzpicture}}
\end{minipage}
\begin{minipage}{0.35\textwidth}
\centering
\scalebox{\scalefactor}{
\begin{tikzpicture}[node distance = 1.5cm,>=stealth',bend angle=0,auto]
    \node [place,tokens=1] (start) [label=below:$p_i$] {};
    \node [transition] (t1) [right of=start,label=center:$a$] {}
    edge [pre] (start);
	\node [place] (p1) [above right of=t1] {}
	edge [pre] (t1);
	\node [place] (p2) [below right of=t1] {}
	edge [pre] (t1);
    \node [transition] (t2) [below right of=p1,label=center:$b$] {}
    edge [pre] (p1)
	edge [pre] (p2);
    \node [place] (po) [right of=t2,label=below:$p_o$] {}
    edge [pre] (t2);
    \node at (0,1.15) {$\wfavg{3}$:};
    \draw[opacity=0] ($(start)-(0.25,1.75)$) rectangle ($(po) + (0.25,1.75)$);
\end{tikzpicture}}
\end{minipage}
\caption{Three workflow nets, $\wfavg{1}, \wfavg{2}, \wfavg{3}$, with $\C(\wfavg{1}) = 2$ and $\C(\wfavg{2}) = \C(\wfavg{3}) = 3$.}
\label{fig:acd-examples}
\end{figure}

\begin{description}
\propitemf{\propone}{\yes} 
For the two nets $\wfavg{1}$ and $\wfavg{2}$ of Figure~\ref{fig:acd-examples}, we get the complexity scores $\C(\wfavg{1}) = 2 \neq 3 = \C(\wfavg{2})$.

\propitemf{\proptwo}{\no} 
Follows directly by Theorem~\ref{thm:conn-inf}. 
Figure~\ref{fig:acd-fin} shows how to construct such a sequence of places and transitions.
\begin{figure}[ht]
\begin{center}
\scalebox{\scalefactor}{
\begin{tikzpicture}[node distance = 1.5cm,>=stealth',bend angle=0,auto]
    \node [place,tokens=1] (start) [label=below:$p_i$] {};
    \node [transition] (t1) [right of=start,label=below:$t_1$] {}
    edge [pre] (start);
    \node [place] (p1) [right of=t1] {}
    edge [pre] (t1);
    \node (dots) [right of=p1] {$\dots$}
    edge [pre] (p1);
    \node [transition] (tk) [right of=dots,label=below:$t_k$] {}
    edge [pre] (dots);
    \node [place] (p1) [above right of=tk,label=below:$p^1$] {}
    edge [pre] (tk);
    \node [place] (pk) [below right of=tk,label=below:$p^{c}$] {}
    edge [pre] (tk);
    \node [transition] (t1) [right of=p1,label=below:$t^1$] {}
    edge [pre] (p1);
    \node [transition] (tk) [right of=pk,label=below:$t^{c}$] {}
    edge [pre] (pk);
    \node [place] (end) [below right of=t1,label=below:$p_o$] {}
    edge [pre] (t1)
    edge [pre] (tk);
    \node at ($0.5*(p1) + 0.5*(pk)$) {$\vdots$};
    \node at ($0.5*(t1) + 0.5*(tk)$) {$\vdots$};
    \node at (0,1.15) {$W_{c,k}^{\name}$:};
\end{tikzpicture}}
\end{center}
\caption{A workflow net $W_{c,k}^{\name}$ with $k + c$ transitions and $k + c + 1$ places for any $c,k \in \mathbb{N}$ with $c \geq 2$. Its complexity score is $\C(W_{c,k}^{\name}) = \frac{2c + 1}{2} = c + \frac{1}{2}$.}
\label{fig:acd-fin}
\end{figure}
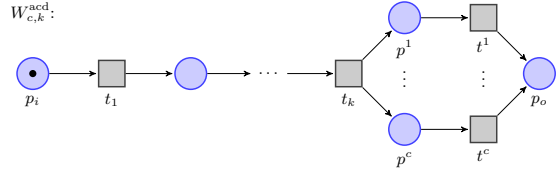

\propitemf{\propthree}{\yes} 
The workflow nets $\wfavg{2}$ and $\wfavg{3}$ of Figure~\ref{fig:acd-examples} are different in structure, but have equal complexity: $\compavgconn(\wfavg{2}) = \compavgconn(\wfavg{3}) = 3$.

\propitemf{\propfour}{\yes} 
The two workflow nets $\wfavg{1}$ and $\wfavg{2}$ of Figure~\ref{fig:acd-examples} have the same language $L(\wfavg{1}) = \{\varepsilon, a\} = L(\wfavg{2})$, but receive different complexity scores $\C(\wfavg{1}) = 2 \neq 3 = \C(\wfavg{2})$.

\propitemf{\propfive}{\no} 
Consider the workflow nets $\wfavg{5}$ and $\wfavg{6}$ of Figure~\ref{fig:acd-mon}. 
We get:
\begin{itemize}
\item $\C(\wfavg{6} \seqop \wfavg{5}) = \frac{14}{4} = 3.5 > 2 = \C(\wfavg{6})$,
\item $\C(\wfavg{6} \parop \wfavg{5}) = \frac{22}{6} = 3.\overline{6} > 2 = \C(\wfavg{6})$,
\item $\C(\wfavg{6} \choiceop \wfavg{5}) = \frac{20}{6} = 3.\overline{3} > 2 = \C(\wfavg{6})$,
\item $\C(\wfavg{6} \loopop \wfavg{5}) = \frac{22}{6} = 3.\overline{6} > 2 = \C(\wfavg{6})$.
\end{itemize}
\begin{figure}
	\centering
\scalebox{\scalefactor}{
\begin{tikzpicture}[node distance = 1.5cm,>=stealth',bend angle=0,auto]
	\node[place] (start) at (0,0) [label=below:$p_i$] {};
	\node[transition, right of=start] (t1) [] {}
	edge [pre] (start);
	\node[place] (p1) [above right of=t1,yshift=0.75cm] {}
	edge [pre] (t1);
	\node[transition, above right of=p1,yshift=0.5cm] (t2) {}
	edge [pre] (p1);
	\node[transition, above right of=p1,yshift=-0.35cm] (t3) {}
	edge [pre] (p1);
	\node[transition, below right of=p1,yshift=0.35cm] (t4) {}
	edge [pre] (p1);
	\node[transition, below right of=p1,yshift=-0.5cm] (t5) {}
	edge [pre] (p1);
	\node[place] (p2) [below right of=t3,yshift=0.35cm] {}
	edge [pre] (t2)
	edge [pre] (t3)
	edge [pre] (t4)
	edge [pre] (t5);
	\node[transition, below right of=p2,yshift=-0.75cm] (t6) {}
	edge [pre] (p2);
	
	\node[place, below right of=t1] (p3) [yshift=-0.75cm] {}
	edge [pre] (t1);
	\node[transition, above right of=p3] (t7) {}
	edge [pre] (p3);
	\node[transition, below right of=p3] (t8) {}
	edge [pre] (p3);
	\node[place, below right of=t7] (p4) [] {}
	edge [pre] (t7)
	edge [pre] (t8)
	edge [post] (t6);
	\node[place, right of=t6] (end) [] {}
	edge [pre] (t6);
	
	\node at ($(start) + (0,2)$) {$\wfavg{4}$:};	
	
	\begin{pgfonlayer}{bg}
		\draw[rounded corners,draw=gray,fill=lightgray!50!white] ($(p1) - (0.75, 2)$) rectangle ($(p2) + (0.75, 2)$);
		\node[gray] at ($(p1) + (-0.4,1.75)$) {$\wfavg{5}$};
		\draw[rounded corners,draw=gray,fill=lightgray!50!white] ($(p3) - (0.75, 1.5)$) rectangle ($(p4) + (0.75, 1.5)$);
		\node[gray] at ($(p3) + (-0.4,-1.25)$) {$\wfavg{6}$};
	\end{pgfonlayer}
\end{tikzpicture}}
\caption{A workflow net $\wfavg{4} = \wfavg{6} \parop \wfavg{5}$ with $\compavgconn(\wfavg{4}) = \frac{22}{6} = 3.\overline{6}$, $\compavgconn(\wfavg{5}) = 4$ and $\compavgconn(\wfavg{6}) = 2$.}
\label{fig:acd-mon}
\end{figure}

\propitemf{\propsix}{\yes} 
The workflow nets $\wfavg{2}$ and $\wfavg{3}$ of Figure~\ref{fig:acd-examples} have the same complexity score $\C(\wfavg{2}) = 3 = \C(\wfavg{3})$.
We get:
\begin{itemize}
\item $\C(\wfavg{2} \seqop \wfavg{2}) = \frac{14}{4} = 3.5 \neq 3.25 = \frac{13}{4} = \C(\wfavg{3} \seqop \wfavg{2})$,
\item $\C(\wfavg{2} \parop \wfavg{2}) = \frac{22}{6} = 3.\overline{6} \neq 3.\overline{3} = \frac{20}{6} = \C(\wfavg{3} \parop \wfavg{2})$,
\item $\C(\wfavg{2} \choiceop \wfavg{2}) = \frac{20}{6} = 3.\overline{3} \neq 3 = \frac{18}{6} = \C(\wfavg{3} \choiceop \wfavg{2})$,
\item $\C(\wfavg{2} \loopop \wfavg{2}) = \frac{22}{6} = 3.\overline{6} \neq 3 = \frac{18}{6} = \C(\wfavg{3} \loopop \wfavg{2})$.
\end{itemize}

\propitemf{\propseven}{\yes}
Figure~\ref{fig:acd-perm} shows two workflow nets, $\wfavg{7}$ and $\wfavg{8}$, that are permutations of each other, but $\C(\wfavg{7}) = 2.\overline{6} \neq 3.5 = \C(\wfavg{8})$.
\begin{figure}
\centering
\begin{minipage}{0.4\textwidth}
\centering
\scalebox{\scalefactor}{
\begin{tikzpicture}[node distance = 1.5cm,>=stealth',bend angle=0,auto]
    \node [place,tokens=1] (start) [label=below:$p_i$] {};
    \node [transition] (t2) [right of=start,label=center:$b$] {}
    edge [pre] (start);
    \node [transition] (t1) [above of=t2,label=center:$a$] {}
    edge [pre] (start);
	\node [place] (p1) [right of=t2,label=below:$p_1$] {}
    edge [pre] (t1)
    edge [pre] (t2);
    \node [transition] (t4) [right of=p1,label=center:$d$] {}
    edge [pre] (p1);
    \node [transition] (t3) [above of=t4,label=center:$c$] {}
    edge [pre] (p1);
    \node [place] (po) [right of=t4,label=below:$p_o$] {}
    edge [pre] (t3)
    edge [pre] (t4);
    \node at (0,1.15) {$\wfavg{7}$:};
    \draw[opacity=0] ($(start)-(0.25,1.75)$) rectangle ($(po) + (0.25,1.75)$);
\end{tikzpicture}}
\end{minipage}
\begin{minipage}{0.4\textwidth}
\centering
\scalebox{\scalefactor}{
\begin{tikzpicture}[node distance = 1.5cm,>=stealth',bend angle=0,auto]
    \node [place,tokens=1] (start) [label=below:$p_i$] {};
    \node [transition] (t1) [right of=start,label=center:$a$] {}
    edge [pre] (start);
	\node [place] (p1) [right of=t1,label=below:$p_1$] {}
    edge [pre] (t1);
    \node [transition] (t4) [right of=p1,label=center:$c$] {}
    edge [pre] (p1);
    \node [transition] (t3) [above of=t4,label=center:$b$] {}
    edge [pre] (p1);
    \node [transition] (t2) [below of=t4,label=center:$d$] {}
    edge [pre] (p1);
    \node [place] (po) [right of=t4,label=below:$p_o$] {}
	edge [pre] (t2)
    edge [pre] (t3)
    edge [pre] (t4);
    \node at (0,1.15) {$\wfavg{8}$:};
    \draw[opacity=0] ($(start)-(0.25,1.75)$) rectangle ($(po) + (0.25,1.75)$);
\end{tikzpicture}}
\end{minipage}
\caption{Two workflow nets, $\wfavg{7}$ and $\wfavg{8}$, with $\wfavg{8} \in Perm(\wfavg{7})$, $\C(\wfavg{7}) = 2.\overline{6}$ and $\C(\wfavg{8}) = 3.5$.}
\label{fig:acd-perm}
\end{figure}

\propitemf{\propeight}{\yes} 
The $\compavgconn$ measure is robust against relabellings, as can be seen by its general insensitivity to transition labels.

\propitemf{\propnine}{\no} 
If we consider only inputs workflow nets where  $\compavgconn$ is defined, and thus can assume that all input workflow nets have at least one connector, $\compavgconn$ is subadditive for the operations $\seqop$, $\parop$, $\choiceop$ and $\loopop$.
\begin{proof}
Let $M_1 = (P^1,T^1,F^1,p_i^1,p_o^1)$ and $M_2 = (P^2,T^2,F^2,p_i^2,p_o^2)$ be two workflow nets. 
The sequential composition $M_1 \seqop M_2$ introduces only one new transition with one incoming and one outgoing arc.
By definition, this transition can't be a connector node, since neither its number of incoming nor the number of outgoing arcs exceeds $1$.
The degrees of $p_o^1$ and $p_i^2$ increase by $1$ with this operation.
This might impact the measure when $p_o^1$ or $p_i^2$ are connector nodes.
With $\compavgconn(M_1) = \frac{d_1}{|\mathcal{C}_1^M|}$ and $\compavgconn(M_2) = \frac{d_2}{|\mathcal{C}_2^M|}$, we have:
\[\compavgconn(M_1 \seqop M_2) \leq \frac{d_1 + d_2 + 2}{|\mathcal{C}_1^M| + |\mathcal{C}_2^M|}.\]
Since $\frac{d_1 + d_2 +2}{|\mathcal{C}_1^M| + |\mathcal{C}_2^M|} \leq \frac{|\mathcal{C}_2^M| \cdot d_1 + |\mathcal{C}_1^M| \cdot d_2}{|\mathcal{C}_1^M||\mathcal{C}_2^M|}$, for $d_1, d_2 \geq 2$ and $|\mathcal{C}_1^M|, |\mathcal{C}_2^M| \geq 2$:
\begin{align*}
\compavgconn(M_1 \seqop M_2) & \leq \frac{d_1}{|\mathcal{C}_1^M|} + \frac{d_2}{|\mathcal{C}_2^M|} \\
\compavgconn(M_1 \seqop M_2) & \leq \compavgconn(M_1) + \compavgconn(M_2).
\end{align*}
As sequential composition is associative, this claim also holds if we increase the amount of workflow nets composed by the $\seqop$-operator.

The operators $\parop$, $\choiceop$ and $\loopop$ each create two new connectors of degree at least $n$ for their $n$ parameters $M_1, \dots, M_n$.
They may also increase the node degree of the initial and final places of these nets by one, if they are connectors.
With $\compavgconn(M_i) = \frac{d_i}{|\mathcal{C}_i^{M_i}|}$ for $i \in \{1,\dots, n\}$, we get:
\[\compavgconn(\oplus(M_1, \dots, M_n)) \leq \frac{4(n+1) + \sum_{i = 1}^n d_i}{2 + \sum_{i = 1}^n |\mathcal{C}_i^{M_i}|}.\]
When there are two nets, each with the minimum possible average connector degree $2$, this bound is $\frac{20}{6} = 3.\overline{3}$, which is less than the sum of the average degrees, $4$.
The sum of $\compavgconn(M_i)$ diverges as $n$, $d_i$ and $|C_i^{M_i}|$ increase.
Therefore, we have
\[\compavgconn(\oplus(M_1, \dots, M_n)) \leq \compavgconn(M_1) + \dots + \compavgconn(M_2)\] \qed
\end{proof}

\propitemf{\propdef}{\no} 
When there are no connectors in a workflow net, we would divide by $0$, leading to an undefined score.
Figure~\ref{fig:netminwf} shows such a workflow net without connectors.
If we define $\compavgconn$ to be $0$ in this case, the measure only returns non-negative values, since $\ndeg(x) \geq 0$ for each node $x$ in a workflow net.

\propitemf{\propmin}{\yes} 
The smallest possible value of $\compavgconn$ (without introducing the special case above) is $2$, since a connector needs at least $2$ outgoing or $2$ incoming arcs to classify as a connector.
The workflow net $\wfavg{1}$ of Figure~\ref{fig:acd-examples} is the workflow net with the smallest average connector degree and receives this complexity score.

\propitemf{\propinf}{\yes} 
The net $W_{c,k}^{\name}$ in Figure~\ref{fig:acd-fin} shows how to construct nets with an average connector degree $\frac{2c+1}{2}$ where $c \in \mathbb{N}$ with $c \geq 2$.
Therefore, we get $|\{c \in \mathbb{R} \mid \exists M \in \mathcal{M}: \compavgconn(M) = c\}| \geq \{c \in \naturals \mid c \geq 2\}| = \infty$.

\propitemf{\propnotsup}{\yes} 
Take the workflow nets $\wfavg{1}$ and $\wfavg{2}$ of Figure~\ref{fig:acd-examples}.
We get:
\begin{itemize}
\item $\C(\wfavg{1} \seqop \wfavg{2}) = 3 < 5 = 2 + 3 = \C(\wfavg{1}) + \C(\wfavg{2})$,
\item $\C(\wfavg{1} \parop \wfavg{2}) = 3.\overline{3} < 5 = 2 + 3 = \C(\wfavg{1}) + \C(\wfavg{2})$,
\item $\C(\wfavg{1} \choiceop \wfavg{2}) = 3 < 5 = 2 + 3 = \C(\wfavg{1}) + \C(\wfavg{2})$,
\item $\C(\wfavg{1} \loopop \wfavg{2}) = 3.\overline{3} < 5 = 2 + 3 = \C(\wfavg{1}) + \C(\wfavg{2})$.
\end{itemize}

\propitemf{\propadd}{\no} The counter-examples for property \propnotsup{} also show that $\C$ is not additive.
\end{description}

\subsection{Path Complexity}
Complexity measures of the \textsc{Path Complexity} dimension take the complexity of paths in the model into account.
The model is interpreted as a directed graph and paths are defined as in graph-theory, ignoring all semantics of choice or concurrent connectors.
Thus, these measures ignore the semantics of the model and rate the complexity based on how understandable the underlying graph is.

\def\name{\depthname}
\def\C{\depth}
\subsubsection{Depth}
For a workflow net $W = (P,T,F,\ell,p_i,p_o)$, the depth~\cite[pp.124-125]{Men08} of a node $v \in P \cup T$ is defined by its in-depth and its out-depth.
The in-depth counts how many split nodes precede $v$ that were not closed by a join before $v$, while the out-depth counts how many join nodes succeed $v$ that aren't opened by a split after $v$.
For the in-depth, we inspect all paths from $p_i$ to $v$.
Let $\mathcal{P}_{p_i,v}$ be the set of all such paths that don't revisit any nodes. 
Further, let $\mathcal{S}^W := \mathcal{S}_{\text{and}}^W \cup \mathcal{S}_{\text{xor}}^W$ the set of all split nodes in $W$ and $\mathcal{J}^W := \mathcal{J}_{\text{and}}^W \cup \mathcal{J}_{\text{xor}}^W$ the set of all join nodes in $W$.
For $p = (v_1, \dots, v_n) \in \mathcal{P}_{p_i, v}$, we define inductively:
\begin{align}
\lambda_{W}(v_1) &= \lambda_{W}(p_i) := 0\\
\lambda_{p}(v_n) &:=
\begin{cases}
\lambda_{W}(v_{n-1}) + 1 &\text{if } v_{n-1} \in \mathcal{S}^W \land v_n \not \in \mathcal{J}^W\\ 
\lambda_{W}(v_{n-1}) &\text{if } v_{n-1} \in \mathcal{S}^W \land v_n \in \mathcal{J}^W\\ 
\lambda_{W}(v_{n-1}) &\text{if } v_{n-1} \not \in \mathcal{S}^W \land v_n \not \in \mathcal{J}^W\\ 
\lambda_{W}(v_{n-1}) - 1 &\text{if } v_{n-1} \not \in \mathcal{S}^W \land v_n \in \mathcal{J}^W\\ 
\end{cases}
\ \\
\lambda_{W}(v) &:= \max\left\{0,\max_{p \in \mathcal{P}_{p_i,v}} \lambda_{p}(v)\right\} \quad (\text{for any } v \neq p_i)
\end{align}
For the out-depth, we set $\overleftarrow{W} := (P,T,\overleftarrow{F},\ell,p_o,p_i)$ with $\overleftarrow{F} := \{(u,v) \mid (v,u) \in F\}$ as the workflow net with all arcs reversed.
Now, we can reuse the definition of the in-depth to define the out-depth.
We define the depth of a workflow net $W$ as the maximum depth of any node in $W$.
\begin{equation}
\label{eq:depth}
\depth(W) := \max\{\min\{\lambda_{W}(v), \lambda_{\overleftarrow{W}}(v)\} \mid v \in P \cup T\}
\end{equation}
Figure~\ref{fig:depth-examples} shows three example nets and their respective complexity scores.
\begin{figure}[ht]
\begin{center}
\begin{minipage}{0.3\textwidth}
\centering
\scalebox{\scalefactor}{
\begin{tikzpicture}[node distance = 1.5cm,>=stealth',bend angle=0,auto]
	\node [place,tokens=1] (start) [label=below:$p_i$] {};
	\node [transition] (t1) [above right of=start,label=center:$a$] {}
	edge [pre] (start);
	\node [transition] (t2) [below right of=start,label=center:$b$] {}
	edge [pre] (start);
	\node [place] (p1) [below right of=t1,label=below:$p_o$] {}
	edge [pre] (t1)
	edge [pre] (t2);
	\node at (0,1.15) {$W_1^{\name}$:};
\end{tikzpicture}}
\end{minipage}
\begin{minipage}{0.3\textwidth}
\centering
\scalebox{\scalefactor}{
\begin{tikzpicture}[node distance = 1.5cm,>=stealth',bend angle=0,auto]
	\node [place,tokens=1] (start) [label=below:$p_i$] {};
	\node [transition] (t2) [right of=start,label=center:$b$] {}
	edge [pre] (start);
	\node [transition] (t1) [above of=t2,label=center:$a$] {}
	edge [pre] (start);
	\node [transition] (t3) [below of=t2,label=center:$c$] {}
	edge [pre] (start);
	\node [place] (p1) [right of=t2,label=below:$p_o$] {}
	edge [pre] (t1)
	edge [pre] (t2)
	edge [pre] (t3);
	\node at (0,1.15) {$W_2^{\name}$:};
\end{tikzpicture}}
\end{minipage}
\begin{minipage}{0.6\textwidth}
\centering
\scalebox{\scalefactor}{
\begin{tikzpicture}[node distance = 1.5cm,>=stealth',bend angle=0,auto]
	\node [place,tokens=1] (start) [label=below:$p_i$] {};
	\node [transition] (t1) [right of=start,label=center:$\tau$] {}
	edge [pre] (start);
	\node [place] (p1) [right of=t1] {}
	edge [pre] (t1);
	\node [transition] (t3) [above right of=p1,label=center:$b$] {}
	edge [pre] (p1);
	\node [transition] (t4) [below right of=p1,label=center:$c$] {}
	edge [pre] (p1);
	\node [place] (p2) [below right of=t3] {}
	edge [pre] (t3)
	edge [pre] (t4);
	\node [transition] (t2) [above of=t3,label=center:$a$] {}
	edge [pre] (start);
	\node [transition] (t5) [right of=p2,label=center:$\tau$] {}
	edge [pre] (p2);
	\node [place] (p3) [right of=t5,label=below:$p_o$] {}
	edge [pre] (t2)
	edge [pre] (t5);
	\node at (0,1.15) {$W_3^{\name}$:};
\end{tikzpicture}}
\end{minipage}
\end{center}
\caption{Three workflow nets, $W_1^{\name}$, $W_2^{\name}$, $W_3^{\name}$, with $\C(W_3^{\name}) = 2$ and $\C(W_1^{\name}) = \C(W_2^{\name}) = 1$.}
\label{fig:depth-examples}
\end{figure}

\begin{description}
\propitemf{\propone}{\yes} 
For the two nets $W_2^{\name}$ and $W_3^{\name}$ of Figure~\ref{fig:depth-examples}, we get the complexity scores $\C(W_2^{\name}) = 1 \neq 2 = \C(W_3^{\name})$.

\propitemf{\proptwo}{\no} 
There are infinitely many workflow nets with complexity $1$:
Take the workflow net $W_2^{\name}$ of Figure~\ref{fig:depth-examples}.
If we add more transitions $t$ with $\pre{t} = \{p_i\}$ and $\post{t} = \{p_o\}$ to this net, we don't affect the depth of the net.
Therefore, all such workflow nets get the same complexity score, $1$.

\propitemf{\propthree}{\yes} 
The workflow nets $W_1^{\name}$ and $W_2^{\name}$ are different in structure, but both get the complexity score $\C(W_1^{\name}) = 1 = \C(W_2^{\name})$.

\propitemf{\propfour}{\yes} 
Take the workflow nets $W_2^{\name}$ and $W_3^{\name}$ of Figure~\ref{fig:depth-examples}.
Their languages are $L(W_2^{depth}) = \{\varepsilon,a,b,c\} = L(W_3^{depth})$, but their complexity scores are $\C(W_2^{\name}) = 1 \neq 2 = \C(W_3^{\name})$.

\propitemf{\propfive}{\yes}
Let $M_1, \dots, M_n \in \mathcal{M}$ be workflow nets.
By definition, the depth of composed nets can be computed as follows:
\begin{itemize}
\item $\C(\parop(M_1, \dots, M_n)) = 1 + \max \{\C(M_1), \dots, \C(M_n)\}$,
\item $\C(\choiceop(M_1, \dots, M_n)) = 1 + \max \{\C(M_1), \dots, \C(M_n)\}$,
\item $\C(\loopop(M_1, \dots, M_n))$ \newline
$= \max \{\C(M_1), \C(M_2) + 1, \dots, \C(M_n) + 1\}$.
\end{itemize}
Therefore, $\C(\oplus(M_1, \dots, M_n)) \geq \C(M_i)$ for any $\oplus \in \{\parop, \choiceop, \loopop\}$ and $M_i \in \{M_1, \dots, M_n\}$.
Next, consider the sequential operator $\seqop$. 
Take the net $M = \seqop(M_1, \dots, M_n)$ and nodes $v_1, \dots, v_n$ such that $\min\{\lambda_{M_j}(v_j), \lambda_{\overleftarrow{M_j}}(v_j)\}$ for any $j \in \{1, \dots, n\}$, i.e. $v_j$ is the deepest node of the net $M_j$.
Then, we have $\lambda_M(v_j) \geq \lambda_{M_j}(v_j)$ and $\lambda_{\overleftarrow{M}}(v_j) \geq \lambda_{\overleftarrow{M_j}}(v_j)$.
We deduce:
\begin{eqnarray*}
\C(M)
&=& \max_{v}\{\min\{\lambda_{M}(v), \lambda_{\overleftarrow{M}}(v)\}\}\\
&=& \min\{\max\{\lambda_{M}(v_j) \mid 1 \leq j \leq n\}, \max\{\lambda_{\overleftarrow{M}}(v_j) \mid 1 \leq j \leq n\}\}\\
&\ge& \max\{\min\{\lambda_{M_j}(v_j), \lambda_{\overleftarrow{M_j}}(v_j) \mid 1 \leq j \leq n\}\}\\
&=& \max \{\C(M_1), \dots, \C(M_n)\}\\
&\ge& \C(M_i) \text{ for any } M_i \in \{M_1, \dots, M_n\}
\end{eqnarray*}
Therefore, $\C$ is monotone for all operations of Definition~\ref{def:operations}.

\propitemf{\propsix}{\no}
The property does not hold for operations $\oplus\in\{\parop,\choiceop,\loopop\}$.
Let $M_1, M_2, M_3 \in \mathcal{M}$ be workflow nets with $\C(M_1) = \C(M_2)$. 
We get:
\begin{itemize}
\item $\C(M_1 \parop M_3) = 1 + \max \{\C(M_1),\C(M_3)\}$\\
\phantom{$\C(M_1 \parop M_3)$ }$= 1 + \max \{\C(M_2),\C(M_3)\}$\\
\phantom{$\C(M_1 \parop M_3)$ }$= \C(M_2 \parop M_3)$,
\item $\C(M_1 \choiceop M_3) = 1 + \max \{\C(M_1),\C(M_3)\}$\\
\phantom{$\C(M_1 \choiceop M_3)$ }$= 1 + \max \{\C(M_2),\C(M_3)\}$\\ 
\phantom{$\C(M_1 \choiceop M_3)$ }$= \C(M_2 \choiceop M_3)$,
\item $\C(M_1 \loopop M_3) = \max \{\C(M_1),\C(M_3) + 1\}$\\
\phantom{$\C(M_1 \loopop M_3)$ }$= \max \{\C(M_2),\C(M_3) + 1\}$\\ 
\phantom{$\C(M_1 \loopop M_3)$ }$= \C(M_2 \parop M_3)$.
\end{itemize}
For the operator $\seqop$, consider the net $W_4^{\name}$ of Figure~\ref{fig:depth-comp} and the net $\overleftarrow{W_4^{\name}}$ obtained by reversing the arcs of $W_4^{\name}$.
Both workflow nets have complexity $\C(W_4^{\name}) = 1 = \C(\overleftarrow{W_4^{\name}})$, but composing these nets with $W_4^{\name}$ gives $\C (\overleftarrow{W_4^{\name}} \seqop W_4^{\name}) = 2 \neq 1 = \C(W_4^{\name} \oplus W_4^{\name})$.
\begin{figure}[ht]
\begin{center}
\begin{minipage}{0.55\textwidth}
\centering
\scalebox{\scalefactor}{
\begin{tikzpicture}[node distance = 1.5cm,>=stealth',bend angle=0,auto]
	\node [place,tokens=1] (start) [label=below:$p_i$] {};
	\node [transition] (t2) [above right of=start,label=center:$b$] {}
	edge [pre] (start);
	\node [transition] (t1) [above of=t2,label=center:$a$] {}
	edge [pre] (start);
	\node [transition] (t3) [below right of=start,label=center:$c$] {}
	edge [pre] (start);
	\node [transition] (t4) [below of=t3,label=center:$d$] {}
	edge [pre] (start);

	\node [place] (p1) [below right of=t1] {}
	edge [pre] (t1)
	edge [pre] (t2);
	\node [place] (p2) [above right of=t4] {}
	edge [pre] (t3)
	edge [pre] (t4);

	\node [transition] (t5) [right of=p1,label=center:$\tau$] {}
	edge [pre] (p1);
	\node [transition] (t6) [right of=p2,label=center:$\tau$] {}
	edge [pre] (p2);

	\node (m1) [right of=start] {};
	\node (m2) [right of=m1] {};
	\node [place] (end) [right of=m2, label=below:$p_o$] {}
	edge [pre] (t5)
	edge [pre] (t6);

	\node at (0,2.85) {$W_4^{\name}$:};
\end{tikzpicture}}
\end{minipage}
\end{center}
\caption{A workflow net $W_4^{\name}$ with $\C(W_4^{\name}) = 1 = \C(\protect\overleftarrow{W_4^{\name}})$.}
\label{fig:depth-comp}
\end{figure}

\propitemf{\propseven}{\yes}
Figure~\ref{fig:depth-perm} shows two workflow nets, $W_5^{\name}$ and $W_6^{\name}$, that are permutations of each other, but $\C(W_5^{\name}) = 2 \neq 1 = \C(W_6^{\name})$.
\begin{figure}[ht]
\begin{center}
\begin{minipage}{0.55\textwidth}
\centering
\scalebox{\scalefactor}{
\begin{tikzpicture}[node distance = 1.5cm,>=stealth',bend angle=0,auto]
	\node [place,tokens=1] (start) [label=below:$p_i$] {};
	\node [transition] (t1) [right of=start,label=center:$\tau$] {}
	edge [pre] (start);
	\node [place] (p1) [right of=t1] {}
	edge [pre] (t1);
	\node [transition] (t3) [above right of=p1,label=center:$b$] {}
	edge [pre] (p1);
	\node [transition] (t4) [below right of=p1,label=center:$c$] {}
	edge [pre] (p1);
	\node [place] (p2) [below right of=t3] {}
	edge [pre] (t3)
	edge [pre] (t4);
	\node [transition] (t2) [above of=t3,label=center:$a$] {}
	edge [pre] (start);
	\node [transition] (t5) [right of=p2,label=center:$\tau$] {}
	edge [pre] (p2);
	\node [place] (p3) [right of=t5,label=below:$p_o$] {}
	edge [pre] (t2)
	edge [pre] (t5);
	\node at (0,1.15) {$W_5^{\name}$:};
\end{tikzpicture}}
\end{minipage}
\ \\
\ \\
\ \\
\begin{minipage}{0.55\textwidth}
\centering
\scalebox{\scalefactor}{
\begin{tikzpicture}[node distance = 1.5cm,>=stealth',bend angle=0,auto]
	\node [place,tokens=1] (start) [label=below:$p_i$] {};
	\node [transition] (t1) [right of=start,label=center:$\tau$] {}
	edge [pre] (start);
	\node [place] (p1) [right of=t1] {}
	edge [pre] (t1);
	\node [transition] (t3) [above right of=p1,label=center:$b$] {}
	edge [pre] (p1);
	\node [transition] (t4) [below right of=p1,label=center:$c$] {}
	edge [pre] (p1);
	\node [place] (p2) [below right of=t3] {}
	edge [pre] (t3)
	edge [pre] (t4);
	\node [transition] (t2) [above of=t3,label=center:$a$] {}
	edge [pre] (p1)
	edge [post] (p2);
	\node [transition] (t5) [right of=p2,label=center:$\tau$] {}
	edge [pre] (p2);
	\node [place] (p3) [right of=t5,label=below:$p_o$] {}
	edge [pre] (t5);
	\node at (0,1.15) {$W_6^{\name}$:};
\end{tikzpicture}}
\end{minipage}
\end{center}
\caption{Two workflow nets with $W_5^{\name} \in Perm(W_6^{\name})$ and complexity scores $\C(W_5^{\name}) = 2 \neq 1 = \C(W_6^{\name})$.}
\label{fig:depth-perm}
\end{figure}

\propitemf{\propeight}{\yes}
A relabelling of the transitions does not affect $\C$.

\propitemf{\propnine}{\yes} 
Consider the workflow net $\netminwf$ of Figure~\ref{fig:netminwf}.
For all operators except $\seqop$, we get:
\begin{itemize}
\item $\C(\netminwf \parop \netminwf) = 1 > 0 = 0 + 0 = \C(\netminwf) + \C(\netminwf)$,
\item $\C(\netminwf \choiceop \netminwf) = 1 > 0 = 0 + 0 = \C(\netminwf) + \C(\netminwf)$,
\item $\C(\netminwf \loopop \netminwf) = 1 > 0 = 0 + 0 = \C(\netminwf) + \C(\netminwf)$.
\end{itemize}
Next, we consider the operator $\seqop$.
For each $n\in\mathbb{N}$ there is a workflow net $N := W_{sub,n}^{\depthname}$ with $\max_v\{\lambda_{N}(v)\} = 1$, $\max_v\{\lambda_{\overleftarrow{N}}(v)\} = n$ and, for all $v \in P \cup T$, $\lambda_{N}(v) = 1 \Leftrightarrow \lambda_{\overleftarrow{N}}(v) =n$.
Figure \ref{fig:depth:NSUB} shows a possible construction for $n=3$.
For any $n>2$, we get that $\C(N) = \C(\overleftarrow{N})=1$ and $\C (\overleftarrow{N} \seqop N) = n > 2 = \C(N) +  \C(\overleftarrow{N})$.
\begin{figure}[ht]
\begin{center}
\centering
\scalebox{\scalefactor}{
\begin{tikzpicture}[node distance = 1.5cm,>=stealth',bend angle=0,auto]
	\node [place,tokens=1] (start) [label=below:$p_i$] {};
	\node [transition] (t3) [right of=start,label=center:$c$] {}
	edge [pre] (start);
	\node [transition] (t4) [below of=t3,label=center:$d$] {}
	edge [pre] (start);
	\node [transition] (t2) [above of=t3,label=center:$b$] {}
	edge [pre] (start);
	\node [transition] (t1) [above of=t2,label=center:$a$] {}
	edge [pre] (start);

	\node [place] (p1) [right of=t3] {}
	edge [pre] (t3)
	edge [pre] (t4);
	\node [transition] (t5) [right of=p1,label=center:$\tau$] {}
	edge [pre] (p1);
	\node [place] (p2) [right of=t5] {}
	edge [pre] (t5)
	edge [pre] (t2);
	\node [transition] (t6) [right of=p2,label=center:$\tau$] {}
	edge [pre] (p2);
	\node [place] (end) [right of=t6, label=below:$p_o$] {}
	edge [pre] (t1)
	edge [pre] (t6);

	\node at (0,2.85) {$W_{sub,n}^{\depthname}$:};
\end{tikzpicture}}
\end{center}
\caption{A workflow net $W_{sub,n}^{\depthname}$ with $\C(W_{sub,n}^{\depthname}) = 1 = \C(\protect\overleftarrow{W_{sub,n}^{\depthname}})$,
but where $\C(\protect\overleftarrow{W_{sub,n}^{\depthname}} \seqop W_{sub,n}^{\depthname}))=3$.}
\label{fig:depth:NSUB}
\end{figure}

\propitemf{\propdef}{\yes}
$\C$ is well-defined and deterministic, since $\mathcal{P}_{p_i,v}$ is finite for each node $v \in P \cup T$ and because each node lies on a path from $p_i$ to $p_o$ in a workflow net.
By definition of $\lambda_{W}(v)$, $\C$ returns only non-negative values.

\propitemf{\propmin}{\yes} 
The smallest possible workflow net $\netminwf$ of Figure~\ref{fig:netminwf} has no splits and joins. 
Thus, $\C(\netminwf) = 0$, which is the minimum of this complexity measure.

\propitemf{\propinf}{\yes} 
Let $n \in \mathbb{N}$ with $n \geq 3$. 
We can easily construct a workflow net with complexity $n$ by creating $n$ subsequent split connectors followed by $n$ subsequent join connectors on a path from $p_i$ to $p_o$.
Therefore, we get $|\{c \in \mathbb{R} \mid \exists M \in \mathcal{M}: \C(M) = c\}| \geq |\{n \in \mathbb{N} \mid n \geq 3\}| = \infty$.

\propitemf{\propnotsup}{\yes} 
For the workflow nets $W_1^{\depthname}$ of Figure~\ref{fig:depth-examples}, we can easily see that $\C(W_1^{\depthname} \seqop W_1^{\depthname}) = 1 < 2 = \C(W_1^{\depthname}) + \C(W_1^{\depthname})$.
For all other operations, we use our analyses for the property {\propfive} and deduce:
\begin{itemize}
\item $\C(M_1 \parop M_2) = 1 + \max\{\C(M_1), \C(M_2)\}$ \\
\phantom{$\C(M_1 \parop M_2)$ }$< \C(M_1) + \C(M_2)$,
\item $\C(M_1 \choiceop M_2) = 1 + \max\{\C(M_1), \C(M_2)\}$ \\
\phantom{$\C(M_1 \choiceop M_2)$ }$< \C(M_1) + \C(M_2)$,
\item $\C(M_1 \loopop M_2) = \max\{\C(M_1), \C(M_2) + 1\}$ \\
\phantom{$\C(M_1 \loopop M_2)$ }$< \C(M_1) + \C(M_2)$
\end{itemize}
for workflow nets $M_1, M_2 \in \mathcal{M}$ with $\C(M_1), \C(M_2) \geq 2$.

\propitemf{\propadd}{\no}
The counter-examples for property {\propnine} also show that $\C$ is not additive.
\end{description}

\def\name{\diametername}
\def\C{\diameter}
\subsubsection{Diameter}
The diameter~\cite[p.119]{Men08} of a workflow net $W = (P,T,F,\ell,p_i,p_o)$ is the maximal length of a way from $p_i$ to $p_o$ where no arc is used more than once.
Note that, in contrast to paths, ways are allowed to revisit nodes.
The length of a way $w = v_1, \dots, v_k$ it the number of visited nodes, where revisited nodes can be counted more than once.
For $w = v_1, \dots, v_k$, we have $|w| = k$.
Let $\mathcal{W}_{p_i,p_o}$ be the set of all ways as defined above. 
We define:
\begin{equation}
\label{eq:diam}
\diameter(W) = \max\{|p| \mid p \in \mathcal{W}_{p_i, p_o}\}.
\end{equation}
Figure~\ref{fig:diam-examples} shows three example nets and their respective complexity scores.
\begin{figure}[ht]
\begin{center}
\begin{minipage}{0.275\textwidth}
\centering
\scalebox{\scalefactor}{
\begin{tikzpicture}[node distance = 1.5cm,>=stealth',bend angle=0,auto]
	\node [place,tokens=1] (start) [label=below:$p_i$] {};
	\node [transition] (t1) [above right of=start,label=center:$a$] {}
	edge [pre] (start);
	\node [transition] (t2) [below right of=start,label=center:$b$] {}
	edge [pre] (start);
	\node [place] (p1) [below right of=t1,label=below:$p_o$] {}
	edge [pre] (t1)
	edge [pre] (t2);
	\node at (0,1.15) {$W_1^{\name}$:};
\end{tikzpicture}}
\end{minipage}
\begin{minipage}{0.275\textwidth}
\centering
\scalebox{\scalefactor}{
\begin{tikzpicture}[node distance = 1.5cm,>=stealth',bend angle=0,auto]
	\node [place,tokens=1] (start) [label=below:$p_i$] {};
	\node [transition] (t1) [right of=start,label=center:$a$] {}
	edge [pre] (start);
	\node [place] (p1) [right of=t1,label=below:$p_o$] {}
	edge [pre] (t1);
	\node at (0,1.15) {$W_2^{\name}$:};
\end{tikzpicture}}
\end{minipage}
\begin{minipage}{0.425\textwidth}
\centering
\scalebox{\scalefactor}{
\begin{tikzpicture}[node distance = 1.5cm,>=stealth',bend angle=0,auto]
	\node [place,tokens=1] (start) [label=below:$p_i$] {};
	\node [transition] (t1) [right of=start,label=center:$a$] {}
	edge [pre] (start);
	\node [place] (p1) [right of=t1] {}
	edge [pre] (t1);
	\node [transition] (t2) [right of=p1,label=center:$\tau$] {}
	edge [pre] (p1);
	\node [place] (p2) [right of=t2,label=below:$p_o$] {}
	edge [pre] (t2);
	\node at (0,1.15) {$W_3^{\name}$:};
\end{tikzpicture}}
\end{minipage}
\end{center}
\caption{Three workflow nets, $W_1^{\name}$, $W_2^{\name}$, $W_3^{\name}$, with $\C(W_3^{\name}) = 5$ and $\C(W_1^{\name}) = \C(W_2^{\name}) = 3$.}
\label{fig:diam-examples}
\end{figure}
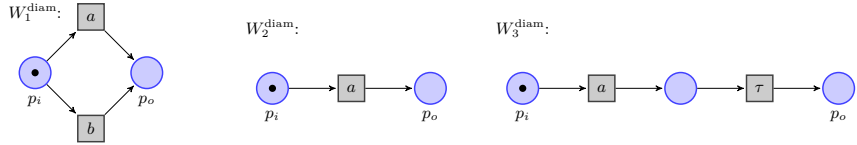

\begin{description}
\propitemf{\propone}{\yes} 
For the two nets $W_2^{\name}$ and $W_3^{\name}$ of Figure~\ref{fig:diam-examples}, we get the complexity scores $\C(W_2^{\name}) = 3 \neq 5 = \C(W_3^{\name})$.

\propitemf{\proptwo}{\no} 
There are infinitely many workflow nets with complexity $3$:
Take the workflow net $W_1^{\name}$ of Figure~\ref{fig:diam-examples}. 
If we add more transitions $t$ with $\pre{t} = \{p_i\}$ and $\post{t} = \{p_o\}$ to this net, we don't affect the diameter of the net.
Therefore, all such workflow nets get the same complexity score $3$.

\propitemf{\propthree}{\yes} 
The workflow nets $W_1^{\name}$ and $W_2^{\name}$ of Figure~\ref{fig:diam-examples} are different in structure, but both have the same diameter and therefore get the same complexity score $\C(W_1^{\name}) = 3 = \C(W_2^{\name})$. 

\propitemf{\propfour}{\yes} 
Take the workflow nets $W_2^{\name}$ and $W_3^{\name}$ of Figure~\ref{fig:diam-examples}. 
Their languages are $L(W_2^{\name}) = \{\varepsilon, a\} = L(W_3^{\name})$, but their complexity scores are $\C(W_2^{\name}) = 3 \neq 5 = \C(W_3^{\name})$.

\propitemf{\propfive}{\yes}
We use the following Theorem:
\begin{theorem}
\label{thm:diam-mon}
Let $M_1, \dots, M_n \in \mathcal{M}$ be workflow nets. Then,
\begin{itemize}
\item $\C(\seqop(M_1, \dots, M_n)) = n - 1 + \C(M_1) + \dots + \C(M_n)$,
\item $\C(\parop(M_1, \dots, M_n)) = 4 + \max\{\C(M_1), \dots \C(M_n)\}$,
\item $\C(\choiceop(M_1, \dots, M_n) = 4 + \max\{\C(M_1), \dots, \C(M_n)\}$,
\item $\C(\loopop(M_1, \dots, M_n) = 8 + \C(M_1)$.
\end{itemize}
\end{theorem}
\begin{proof}
We show these claims for each operation individually.
Let $w_j$ be the longest way in $M_j$ for any $j \in \{1, \dots, n\}$.
\begin{itemize}
\item Let $t_1^*, \dots, t_{n-1}^*$ be the transitions introduced by the $\seqop$ operation.
The longest way in $M = \hspace*{1mm}\seqop\hspace*{-1mm}(M_1, \dots, M_n)$ then is $w_1, t_1^*, \dots, t_{n-1}^* w_n$. 
Its length is $|w_1| + \dots + |w_n| + n - 1$.
\item Let $p_i^*, t_i^*, t_o^*, p_o^*$ be the places and transitions introduced by the $\parop$ operation.
The longest way in $M = \parop(M_1, \dots M_n)$ then is $p_i^*, t_i^*, w, t_o^*, p_o^*$, where $w$ is the longest of the paths $w_1, \dots, w_n$.
The length of this new path is $\max\{|w_1|, \dots, |w_n|\} + 4$.
\item Let $p_i^*, t_1^*, \dots, t_n^*, s_1^*, \dots, s_n^*, p_o^*$ be the places and transitions introduced by the $\choiceop$ operation.
The longest way in $M = \choiceop(M_1, \dots, M_n)$ then is $p_i^*, t_m^*, w_m, s_m^*, p_o^*$, where $m$ is the index of the longest of the paths $w_1, \dots, w_n$.
Its length is $\max\{|w_1|, \dots, |w_n|\} + 4$.
\item Let $p_i^*, t^*, p^*, t_1^*, \dots, t_n^*, s_1^*, \dots, s_n^*, q^*, s^*, p_o^*$ be the places and transitions introduced by the $\loopop$ operation.
The longest way in $\loopop(M_1, \dots, M_n)$ is the way $p_i^*, t^*, p^*, t_1^*, w_1, s_1^*, q^*, s^*, p_o^*$, since we can't visit any nodes of $M_2, \dots, M_n$ without needing to reuse arcs of $M_1$. 
The length of this path is $|w_1| + 8$.
\end{itemize}
\qed
\end{proof}
With Theorem~\ref{thm:diam-mon}, we get for any $M_1, \dots, M_n \in \mathcal{M}$ and $M_i \in \{M_1, \dots, M_n\}$:
\begin{itemize}
\item $\C(\seqop(M_1, \dots, M_n)) = 1 + \C(M_1) + \dots + \C(M_n) > \C(M_i)$,
\item $\C(\parop(M_1, \dots, M_n)) = 4 + \max\{\C(M_1), \dots \C(M_n)\} > \C(M_i)$,
\item $\C(\choiceop(M_1, \dots, M_n)) = 4 + \max\{C(M_1, \dots, \C(M_n)\} > \C(M_i)$.
\end{itemize}
For the operator $\loopop$, on the other hand, monotonicity not necessarily applies, as we can find models $M_1, M_2$ with $\C(M_1) + 8 < \C(M_2)$, contradicting monotonicity for the $\loopop$ operator.

\propitemf{\propsix}{\no}
Let $M_1, M_2, M_3 \in \mathcal{M}$ be workflow nets with complexity scores $\C(M_1) = \C(M_2)$.
Theorem~\ref{thm:diam-mon} gives:
\begin{itemize}
\item $\C(M_1 \seqop M_3) = 1 + \C(M_1) + \C(M_3)$\\
\phantom{$\C(M_1 \seqop M_3)$ }$= 1 + \C(M_2) + \C(M_3) = \C(M_2 \seqop M_3)$,
\item $\C(M_1 \parop M_3) = 4 + \max\{\C(M_1), \C(M_3)\}$\\
\phantom{$\C(M_1 \parop M_3)$ }$= 4 + \max\{\C(M_2),\C(M_3)\} = \C(M_2 \parop M_3)$,
\item $\C(M_1 \choiceop M_3) = 4 + \max\{C(M_1, \C(M_3)\}$\\
\phantom{$\C(M_1 \choiceop M_3)$ }$= 4 + \max\{\C(M_2),\C(M_3)\} = \C(M_2 \choiceop M_3)$,
\item $\C(M_1 \loopop M_3) = 8 + \C(M_1)$\\
\phantom{$\C(M_1 \loopop M_3)$ }$= 8 + \C(M_2) = \C(M_2 \loopop M_3)$.
\end{itemize}

\propitemf{\propseven}{\yes}
Figure \ref{fig:diam:perm} shows two workflow nets, $W_4^{\name}$ and $W_5^{\name}$, that are permutations of each other, but $\C(W_4^{\name}) = 7 \neq 9 = \C(W_5^{\name})$.
\begin{figure}[ht]
\begin{center}
\begin{minipage}{0.55\textwidth}
\centering
\scalebox{\scalefactor}{
\begin{tikzpicture}[node distance = 1.5cm,>=stealth',bend angle=0,auto]
	\node [place,tokens=1] (start) [label=below:$p_i$] {};
	\node [transition] (t1) [right of=start,label=center:$\tau$] {}
	edge [pre] (start);
	\node [place] (p1) [above right of=t1] {}
	edge [pre] (t1);
	\node [place] (p2) [below right of=t1] {}
	edge [pre] (t1);

	\node [transition] (t3) [right of=p1,label=center:$b$] {}
	edge [pre] (p1);
	\node [transition] (t4) [right of=p2,label=center:$c$] {}
	edge [pre] (p2);
	
	\node [place] (p3) [right of=t3] {}
	edge [pre] (t3);
	\node [place] (p4) [right of=t4] {}
	edge [pre] (t4);
	\node [transition] (t5) [below right of=p3,label=center:$\tau$] {}
	edge [pre] (p3)
	edge [pre] (p4);
	\node [place] (p5) [right of=t5,label=below:$p_o$] {}
	edge [pre] (t5);
	\node at (0,1.15) {$W_4^{\name}$:};
\end{tikzpicture}}
\end{minipage}
\ \\
\ \\
\ \\
\begin{minipage}{0.55\textwidth}
\centering
\scalebox{\scalefactor}{
\begin{tikzpicture}[node distance = 1.5cm,>=stealth',bend angle=0,auto]
	\node [place,tokens=1] (start) [label=below:$p_i$] {};
	\node [transition] (t1) [right of=start,label=center:$\tau$] {}
	edge [pre] (start);
	\node [place] (p1) [above right of=t1] {};
	\node [place] (p2) [below right of=t1] {}
	edge [pre] (t1);

	\node [transition] (t3) [right of=p1,label=center:$b$] {}
	edge [pre] (p1);
	\node [transition] (t4) [right of=p2,label=center:$c$] {}
	edge [pre] (p2)
	edge [post] (p1);
	
	\node [place] (p3) [right of=t3] {}
	edge [pre] (t3);
	\node [place] (p4) [right of=t4] {}
	edge [pre] (t4);
	\node [transition] (t5) [below right of=p3,label=center:$\tau$] {}
	edge [pre] (p3)
	edge [pre] (p4);
	\node [place] (p5) [right of=t5,label=below:$p_o$] {}
	edge [pre] (t5);
	\node at (0,1.15) {$W_5^{\name}$:};
\end{tikzpicture}}
\end{minipage}
\end{center}
\caption{Two workflow nets with $W_4^{\name} \in Perm(W_5^{\name})$ and complexity scores $\C(W_4^{\name}) = 7 \neq 9 = \C(W_5^{\name})$.}
\label{fig:diam:perm}
\end{figure}

\propitemf{\propeight}{\yes}
A relabelling of the transitions does not affect $\C$.

\propitemf{\propnine}{\yes} 
Take the workflow net $\netminwf$ of Figure~\ref{fig:netminwf}.
We get:
\begin{itemize}
\item $\C(M_1 \seqop M_2) > \C(M_1) + \C(M_2)$\\ 
for all $M_1, M_2 \in \mathcal{M}$ due to Theorem~\ref{thm:diam-mon},
\item $\C(\netminwf \parop \netminwf) = 7 > 6 = 3 + 3 = \C(\netminwf) + \C(\netminwf)$,
\item $\C(\netminwf \choiceop \netminwf) = 7 > 6 = 3 + 3 = \C(\netminwf) + \C(\netminwf)$,
\item $\C(\netminwf \loopop \netminwf) = 11 > 6 = 3 + 3 = \C(\netminwf) + \C(\netminwf)$
\end{itemize}

\propitemf{\propdef}{\yes} 
Since the output place must be reachable from the input place, there is at least one path in each workflow net, so the maximum always exists.
This maximum is non-negative, since the length of a path can't be negative.

\propitemf{\propmin}{\yes} 
The minimum possible diameter is that of the smallest possible workflow net $\netminwf$. 
This workflow net has diameter $3$, using two arcs.
It is not possible to use less than two arcs for a path from $p_i$ to $p_o$ in a workflow net, as this would mean that $(p_i, p_o) \in F$, contradicting $F \subseteq (P \times T) \cup (T \times P)$.

\propitemf{\propinf}{\yes} 
Let $n \in \mathbb{N}$ with $n \geq 3$ and $n$ being an odd number.
Figure~\ref{fig:diam-inf} shows how to construct a workflow net of complexity $n$, so we immediately get $|\{c \in \mathbb{R} \mid \exists M \in \mathcal{M}: \C(M) = c\}| \geq |\{n \in \mathbb{N} \mid n \geq 3 \land \frac{n}{2} \not\in \mathbb{N}\}| = \infty$.
\begin{figure}[ht]
\begin{center}
\scalebox{\scalefactor}{
\begin{tikzpicture}[node distance = 1.5cm,>=stealth',bend angle=0,auto]
	\node [place,tokens=1] (start) [label=below:$p_i$] {};
	\node [transition] (t1) [right of=start,label=below:$t_1$] {}
	edge [pre] (start);
	\node [place] (p1) [right of=t1,label=below:$p_1$] {}
	edge [pre] (t1);
	\node [transition] (t2) [right of=p1,label=below:$t_2$] {}
	edge [pre] (p1);
	\node (dots) [right of=t2] {$\dots$}
	edge [pre] (t2);
	\node [place] (p2) [right of=dots,label=below:$p_{k-1}$] {}
	edge [pre] (dots);
	\node [transition] (t4) [right of=p2,label=below:$t_k$] {}
	edge [pre] (p2);
	\node [place] (end) [right of=t4,label=below:$p_o$] {}
	edge [pre] (t4);
	\node at (0,1) {$W_{inf,k}^{\name}$:};
\end{tikzpicture}}
\end{center}
\caption{A workflow net $W_{inf,k}^{\name}$ with $k$ transitions and $k+1$ places. Its complexity is $\C(W_{inf,k}^{\name}) = 2k + 1$.}
\label{fig:diam-inf}
\end{figure}
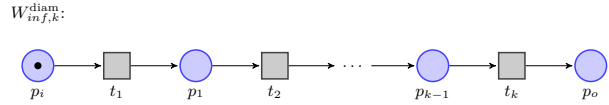

\propitemf{\propnotsup}{\no} 
Theorem~\ref{thm:diam-mon} directly implies that $\C$ is superadditive with respect to the operation $\seqop$, since for any workflow nets $M_1, M_2 \in \mathcal{M}$, we get $\C(M_1 \seqop M_2) = 1 + \C(M_1) + \C(M_2) > \C(M_1) + \C(M_2)$.
For the other operations, we find counter-examples for superadditivity:
Consider the net $W_{inf,4}^{\name}$ of Figure~\ref{fig:diam-inf}.
We get:
\begin{itemize}
\item $\C(W_{inf,4}^{\name} \parop W_{inf,4}^{\name}) = 13 < 18 = \C(W_{inf,4}^{\name}) + \C(W_{inf,4}^{\name})$,
\item $\C(W_{inf,4}^{\name} \choiceop W_{inf,4}^{\name}) = 13 < 18 = \C(W_{inf,4}^{\name}) + \C(W_{inf,4}^{\name})$,
\item $\C(W_{inf,4}^{\name} \loopop W_{inf,4}^{\name}) = 17 < 18 = \C(W_{inf,4}^{\name}) + \C(W_{inf,4}^{\name})$.
\end{itemize}

\propitemf{\propadd}{\no} 
The analysis of the property {\propnine} also show that $\C$ is not additive with respect to any of the operations.
\end{description}

\def\name{\cyclicityname}
\def\C{\cyclicity}
\subsubsection{Cyclicity}
Cycles in a workflow net are often perceived as complex structures.
Therefore, a simple process model aims for as many nodes outside of a cycle as possible~\cite[pp.127-128]{Men08}.
For a workflow net $W$, cyclicity is the number of nodes that lie on a cycle in $W$ divided by the total number of nodes in $W$.
\begin{equation}
\label{eq:cycle}
\C(W) := \frac{|\{x \in P \cup T \mid x \text{ lies on a cycle in } W\}|}{|P| + |T|}
\end{equation}
Figure~\ref{fig:cyc-examples} shows three example nets and their respective complexity scores.
\begin{figure}[ht]
\begin{center}
\begin{minipage}{0.4\textwidth}
\centering
\scalebox{\scalefactor}{
\begin{tikzpicture}[node distance = 1.5cm,>=stealth',bend angle=0,auto]
	\node [place,tokens=1] (start) [label=below:$p_i$] {};
	\node [transition] (t1) [right of=start,label=center:$a$] {}
	edge [pre] (start);
	\node [place] (p1) [right of=t1,label=below:$p_o$] {}
	edge [pre] (t1);
	\node at (0,1.15) {$W_1^{\name}$:};
\end{tikzpicture}}
\end{minipage}
\begin{minipage}{0.4\textwidth}
\centering
\scalebox{\scalefactor}{
\begin{tikzpicture}[node distance = 1.5cm,>=stealth',bend angle=0,auto]
	\node [place,tokens=1] (start) [label=below:$p_i$] {};
	\node [transition] (t1) [right of=start,label=center:$a$] {}
	edge [pre] (start);
	\node [place] (p1) [right of=t1] {}
	edge [pre] (t1);
	\node [transition] (t2) [right of=p1,label=center:$b$] {}
	edge [pre] (p1);
	\node [place] (p2) [right of=t2,label=below:$p_o$] {}
	edge [pre] (t2);
	\node at (0,1.15) {$W_2^{\name}$:};
\end{tikzpicture}}
\end{minipage}
\begin{minipage}{0.8\textwidth}
\centering
\scalebox{\scalefactor}{
\begin{tikzpicture}[node distance = 1.5cm,>=stealth',bend angle=0,auto]
	\node [place,tokens=1] (start) [label=below:$p_i$] {};
	\node [transition] (t1) [right of=start,label=center:$a$] {}
	edge [pre] (start);
	\node [place] (p1) [right of=t1] {}
	edge [pre] (t1);
	\node [transition] (t2) [right of=p1,label=center:$b$] {}
	edge [pre] (p1);
	\node [place] (p2) [right of=t2,label=below:$p_o$] {}
	edge [pre] (t2);
	\node [transition] (t3) [above of=p1,label=center:$\tau$] {}
	edge [pre,bend left=30] (p1)
	edge [post,bend right=30] (p1);
	\node at (0,1.15) {$W_3^{\name}$:};
\end{tikzpicture}}
\end{minipage}
\end{center}
\caption{Three workflow nets, $W_1^{\name}$, $W_2^{\name}$, $W_3^{\name}$, with $\C(W_3^{\name}) = \frac{1}{3}$ and $\C(W_1^{\name}) = \C(W_2^{\name}) = 0$.}
\label{fig:cyc-examples}
\end{figure}

\begin{description}
\propitemf{\propone}{\yes} 
For the two nets $W_2^{\name}$ and $W_3^{\name}$ of Figure~\ref{fig:cyc-examples}, we get the complexity scores $\C(W_2^{\name}) = 0 \neq \frac{1}{3} = \C(W_3^{\name})$.

\propitemf{\proptwo}{\no} 
Take the workflow net $W_1^{\name}$ of Figure~\ref{fig:cyc-examples}. 
If we add more transitions $t$ with $\pre{t} = \{p_i\}$ and $\post{t} = \{p_o\}$ to this net, we don't change its complexity.
Therefore, we can create infinitely many workflow nets with complexity $0$.

\propitemf{\propthree}{\yes} 
The workflow nets $W_1^{\name}$ and $W_2^{\name}$ of Figure~\ref{fig:cyc-examples} are different in structure, but both get the same complexity score $\C(W_1^{\name}) = 3 = \C(W_2^{\name})$.

\propitemf{\propfour}{\yes} 
Take the workflow nets $W_2^{\name}$ and $W_3^{\name}$ of Figure~\ref{fig:cfc-examples}.
Their languages are $L(W_2^{cyc}) = \{\varepsilon, a, ab\} = L(W_3^{cyc})$, but their complexity scores are $\C(W_2^{\name}) = 0 \neq \frac{1}{3} = \C(W_3^{\name})$.

\propitemf{\propfive}{\no}
Consider the workflow nets $W_1^{\name}$ and $W_3^{\name}$ of Figure~\ref{fig:cyc-examples}. 
For the operations $\seqop, \choiceop$ and $\loopop$, we get:
\begin{itemize}
\item $\C(W_3^{\name} \seqop W_1^{\name}) = \frac{2}{10} = 0.2 < 0.\overline{3} = \frac{1}{3} = \C(W_3^{\name})$,
\item $\C(W_3^{\name} \parop W_1^{\name}) = \frac{2}{13} \approx 0.1538 < 0.\overline{3} = \frac{1}{3} = \C(W_3^{\name})$,
\item $\C(W_3^{\name} \choiceop W_1^{\name}) = \frac{2}{15} = 0.1\overline{3} < 0.\overline{3} = \frac{1}{3} = \C(W_3^{\name})$.
\end{itemize}
For operation $\loopop$, the situation is a bit more involved.
In general, the property does not hold for iteration.
Consider a workflow net $M_n \in \mathcal{M}$ with $n$ nodes and a maximal portion
of nodes lying on a cycle.
This portion is $\frac{n-2}{n}$, because the start and end place cannot lie on a cycle by definition.
Now fix another workflow net $M \in \mathcal{M}$ with $k$ nodes.
We show $\C(M_n) > \C(M_n \loopop M)$ if $n$ is large enough.
In $M_n \loopop M$ exactly $4$ nodes that do not lie on a cycle.
That means:
\begin{eqnarray*}
&& \C(M_n) > \C(M_n \loopop M)\\
&\Leftrightarrow& \frac{n-2}{n} > \frac{n + k + 6}{n + k + 10}\\
&\Leftrightarrow& (n-2)(n + k +10) > n(n + k +6)\\
&\Leftrightarrow& n^2 + kn + 10n -2n - 2k - 20 > n^2 + kn + 6n\\
&\Leftrightarrow& 8n - 2k - 20 > 6n\\
&\Leftrightarrow& 2n > 20 + 2k\\
&\Leftrightarrow& n > 10 + k.
\end{eqnarray*}
Note that such a net $M_n$ is not a meaningful workflow net:
Let $t$ be a transition with $p_i \in \pre{t}$.
Since $t$ lies on a cycle, there is another place $p\in\pre{t}$ which is unmarked.
Therefore, $t$ is dead and there is no firing sequence from the start to the end place.

If we only consider sound workflow nets, then the property holds for iteration.
It is easy to see, that $\C(M) \le \frac{n-4}{n}$ if $M$ is a sound workflow net with $n$ nodes.
We deduce for arbitrary additional workflow nets $M_1', \dots, M_n' \in \mathcal{M}$ whose number of nodes sum to $k$:
\begin{eqnarray*}
&& \C(M) \le \C(\loopop(M, M_1', \dots, M_n'))\\
&\Leftarrow& \frac{n-4}{n} \le \frac{n + k + 6}{n + k + 10}\\
&\Leftrightarrow& (n-4)(n + k +10) > n(n + k +6)\\
&\Leftrightarrow& n^2 + kn + 10n -4n - 4k - 40 \le n^2 + kn + 6n\\
&\Leftrightarrow& 6n - 4k - 40 \le 6n\\
&\Leftrightarrow& 0 \le 40 + 4k
\end{eqnarray*}
The last inequation is obviously true.

\propitemf{\propsix}{\yes}
Take the workflow nets $W_1^{\name}, W_2^{\name}$ and $W_3^{\name}$ of Figure~\ref{fig:cyc-examples}.
We have $\C(W_1^{\name}) = 0 = \C(W_2^{\name})$, but:
\begin{itemize}
\item $\C(W_1^{\name} \seqop W_3^{\name}) = \frac{2}{10} = 0.2 \neq 0.1\overline{6} = \frac{2}{12} = \C(W_2^{\name} \seqop W_3^{\name})$,
\item $\C(W_1^{\name} \parop W_3^{\name}) = \frac{2}{13} \approx 0.1538 \neq 0.1\overline{3} = \frac{2}{15} = \C(W_2^{\name} \parop W_3^{\name})$,
\item $\C(W_1^{\name} \choiceop W_3^{\name}) = \frac{2}{15} = 0.1\overline{3} \neq 0.1176 \approx \frac{2}{17} = \C(W_2^{\name} \choiceop W_3^{\name})$,
\item $\C(W_1^{\name} \loopop W_3^{\name}) = \frac{15}{19} \approx 0.789 \neq 0.810 \approx \frac{17}{21} = \C(W_2^{\name} \loopop W_3^{\name})$.
\end{itemize}

\propitemf{\propseven}{\yes}
Figure~\ref{fig:cyc-perm} shows two workflow nets, $W_4^{\name}$ and $W_5^{\name}$, that are permutations of each other, but $\C(W_4^{\name}) = \frac{1}{3} \neq 0 = \C(W_5^{\name})$.
\begin{figure}[ht]
\begin{center}
\begin{minipage}{0.4\textwidth}
\centering
\scalebox{\scalefactor}{
\begin{tikzpicture}[node distance = 1.5cm,>=stealth',bend angle=0,auto]
	\node [place,tokens=1] (start) [label=below:$p_i$] {};
	\node [transition] (t1) [right of=start,label=center:$a$] {}
	edge [pre] (start);
	\node [place] (p1) [right of=t1] {}
	edge [pre] (t1);
	\node [transition] (t2) [right of=p1,label=center:$b$] {}
	edge [pre] (p1);
	\node [place] (p2) [right of=t2,label=below:$p_o$] {}
	edge [pre] (t2);
	\node [transition] (t3) [above of=p1,label=center:$\tau$] {}
	edge [pre,bend left=30] (p1)
	edge [post,bend right=30] (p1);
	\node at (0,1.15) {$W_4^{\name}$:};
\end{tikzpicture}}
\end{minipage}

\begin{minipage}{0.4\textwidth}
\centering
\scalebox{\scalefactor}{
\begin{tikzpicture}[node distance = 1.5cm,>=stealth',bend angle=0,auto]
	\node [place,tokens=1] (start) [label=below:$p_i$] {};
	\node [transition] (t1) [right of=start,label=center:$a$] {}
	edge [pre] (start);
	\node [place] (p1) [right of=t1] {}
	edge [pre] (t1);
	\node [transition] (t2) [right of=p1,label=center:$b$] {}
	edge [pre] (p1);
	\node [place] (p2) [right of=t2,label=below:$p_o$] {}
	edge [pre] (t2);
	\node [transition] (t3) [above of=p1,label=center:$\tau$] {}
	edge [pre] (p1)
	edge [post] (p2);
	\node at (0,1.15) {$W_5^{\name}$:};
\end{tikzpicture}}
\end{minipage}
\end{center}
\caption{Two workflow nets with $W_4^{\name} \in Perm(W_5^{\name})$ and complexity scores $\C(W_4^{\name}) = \frac{1}{3} \neq 0 = \C(W_5^{\name})$.}
\label{fig:cyc-perm}
\end{figure}

\propitemf{\propeight}{\yes}
A relabelling of the transitions does not affect $\C$.

\propitemf{\propnine}{\yes} 
This property does not hold for operations $\oplus \in \{\seqop, \parop, \choiceop\}$, as shown by the following Theorem:
\begin{theorem}
\label{thm:cyc-sub}
Let $M_1, M_2 \in \mathcal{M}$ be workflow nets and $\oplus \in \{\seqop, \parop, \choiceop\}$ one of the operations of Defintiion~\ref{def:operations}. 
Then, $\C(M_1 \oplus M_2) < \C(M_1) + \C(M_2)$.
\end{theorem}
\begin{proof}
Let $x_{\oplus}$ be the number of new nodes the operation $\oplus$ introduces in the net $M := M_1 \oplus M_2$.
Further, let $c_1$ be the number of nodes on a cycle in $M_1$ and $c_2$ be the number of nodes on a cycle in $M_2$.
Let $n_1$ be the total number of nodes in $M_1$ and $n_2$ be the total number of nodes in $M_2$.
Since $\oplus$ doesn't introduce a new cycle, all newly introduced nodes lie outside of a cycle and we have
\[\C(M) = \frac{c_1 + c_2}{n_1 + n_2 + x_{\oplus}} < \frac{c_1 + c_2}{n_1 + n_2} < \frac{c_1}{n_1} + \frac{c_2}{n_2} = \C(M_1) + \C(M_2)\]
\qed
\end{proof}
For the operation $\loopop$, however, we can find a counter-example for subadditivity:
Consider the net $\netminwf$ of Figure~\ref{fig:netminwf}.
If we use this net as input for $\loopop$, we get $\C(\netminwf \loopop \netminwf) = \frac{12}{16} = 0.75 > 0 + 0 = \C(\netminwf) + \C(\netminwf)$.

\propitemf{\propdef}{\yes} 
The smallest possible worfklow net, $\netminwf$, has $3$ nodes, so the denominator of $\C$ is always positive.
Furthermore, the number of nodes on a cycle cannot be negative so $\C(M) \geq 0$ for all workflow nets $M \in \mathcal{M}$.

\propitemf{\propmin}{\yes} 
The minimum possible number of nodes lying on a cycle is $0$, so $0$ is the minimum possible value for $\C$. 
The workflow net $\netminwf$ of Figure~\ref{fig:netminwf} receives this complexity score.

\propitemf{\propinf}{\yes} 
Let $n \in \mathbb{N}$ with $n \geq 2$.
Figure~\ref{fig:cyc-inf} shows how to construct a workflow net with complexity score $\frac{1}{n + 1}$ according to $\C$.
Therefore, we get that $|\{c \in \mathbb{R} \mid \exists M \in \mathcal{M}: \C(M) = c\}| \geq |\{\frac{1}{n+1} \mid n \in \mathbb{N} \land n \geq 2\}| = \infty$.
\begin{figure}[ht]
\begin{center}
\scalebox{\scalefactor}{
\begin{tikzpicture}[node distance = 1.5cm,>=stealth',bend angle=0,auto]
	\node [place,tokens=1] (start) [label=below:$p_i$] {};
	\node [transition] (t1) [right of=start,label=below:$t_1$] {}
	edge [pre] (start);
	\node [place] (p1) [right of=t1,label=below:$p_1$] {}
	edge [pre] (t1);
	\node [transition] (t2) [right of=p1,label=below:$t_2$] {}
	edge [pre] (p1);
	\node (dots) [right of=t2] {$\dots$}
	edge [pre] (t2);
	\node [place] (p2) [right of=dots,label=below:$p_{k-1}$] {}
	edge [pre] (dots);
	\node [transition] (t5) [above of=p2,label=above:$t_{k+1}$] {}
	edge [pre,bend left=20] (p2)
	edge [post,bend right=20] (p2);
	\node [transition] (t4) [right of=p2,label=below:$t_k$] {}
	edge [pre] (p2);
	\node [place] (end) [right of=t4,label=below:$p_o$] {}
	edge [pre] (t4);
	\node at (0,1) {$W_{inf,k}^{\name}$:};
\end{tikzpicture}}
\end{center}
\caption{A workflow net $W_{inf,k}^{\name}$ with $k+1$ transitions and $k+1$ places. One transition and one place lie on a cycle in thsi net. Its complexity is therefore $\C(W_{inf,k}^{\name}) = \frac{2}{2(k+1)} = \frac{1}{k+1}$.}
\label{fig:cyc-inf}
\end{figure}
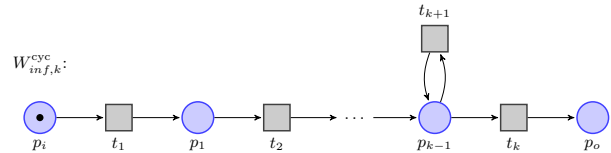

\propitemf{\propnotsup}{\yes}
Theorem~\ref{thm:cyc-sub} directly implies that $\C$ is not superadditive for the operations $\seqop, \parop, \choiceop$.
For the operation $\loopop$, consider the net $W_{\approx 1, k}^{\cyclicityname}$ of Figure~\ref{fig:cycle-unb}.
For $k \geq 2$, $\C(W_{\approx 1, k}^{\cyclicityname} \loopop W_{\approx 1, k}^{\cyclicityname}) < 1 < \frac{2k}{k+1} = \C(W_{\approx 1, k}^{\cyclicityname}) + \C(W_{\approx 1, k}^{\cyclicityname})$, so $\C$ is also not superadditive for the operator $\loopop$.
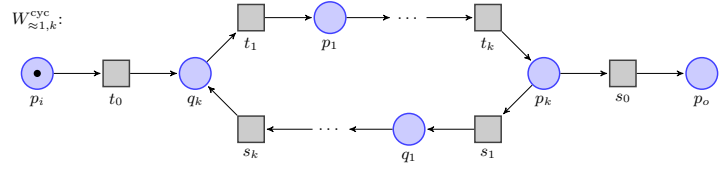
\begin{figure}[ht]
\begin{center}
\scalebox{\scalefactor}{
\begin{tikzpicture}[node distance = 1.5cm,>=stealth',bend angle=0,auto]
	\node [place,tokens=1] (start) [label=below:$p_i$] {};
	\node [transition] (t0) [right of=start,label=below:$t_0$] {}
	edge [pre] (start);
	\node [place] (p0) [right of=t0,label=below:$q_k$] {}
	edge [pre] (t0);
	\node [transition] (t1) [above right of=p0,label=below:$t_1$] {}
	edge [pre] (p0);
	\node [place] (p1) [right of=t1,label=below:$p_1$] {}
	edge [pre] (t1);
	\node (dots) [right of=p1] {$\dots$}
	edge [pre] (p1);
	\node [transition] (tk) [right of=dots,label=below:$t_k$] {}
	edge [pre] (dots);
	\node [place] (pk) [below right of=tk,label=below:$p_k$] {}
	edge [pre] (tk);
	\node [transition] (s1) [below left of=pk,label=below:$s_1$] {}
	edge [pre] (pk);
	\node [place] (q1) [left of=s1,label=below:$q_1$] {}
	edge [pre] (s1);
	\node (dots2) [left of=q1] {$\dots$}
	edge [pre] (q1);
	\node [transition] (sk) [left of=dots2,label=below:$s_k$] {}
	edge [pre] (dots2)
	edge [post] (p0);
	\node [transition] (tk) [right of=pk,label=below:$s_0$] {}
	edge [pre] (pk);
	\node [place] (end) [right of=tk,label=below:$p_o$] {}
	edge [pre] (tk);
	
	\node at (0,1) {$W_{\approx 1, k}^{\cyclicityname}$:};
\end{tikzpicture}}
\end{center}
\caption{A workflow net $W_{\approx 1, k}^{\cyclicityname}$ where $4k$ nodes lie on a cycle. In total, there are $4k + 4$ nodes, so $\C(W_{\approx 1, k}^{\cyclicityname}) = \frac{4k}{4k+4} = \frac{k}{k+1}$.}
\label{fig:cycle-unb}
\end{figure}

\propitemf{\propadd}{\no}
Our analyses of the property {\propnine} show that $\C$ is not additive with respect to any of the operations.
\end{description}

\newpage
\subsection{Degree of Connectedness.}
The fourth complexity dimension found by Lieben et al.~\cite{LieDJJ18} is labeled \textsc{Degree of Connectedness}. 
Complexity measures of this dimension capture how many sequence flows are possible in the model, since each sequence flow introduces complexity.
This dimension contains only two complexity measures, which are quite similar in their behavior.

\def\name{\netconnname}
\def\C{\netconn}
\subsubsection{Coefficient of Network Connectivity}
The coefficient of network connectivity~\cite[p120]{Men08} relates the number of edges in the model to its total number of vertices.
The idea is that workflow nets with a low ratio of edges to vertices are more sparse and therefore easier to understand.
The advantage of this metric in contrast to density is that it can be adapted to any type of graph, since we don't need to know the maximum possible number of edges in the graph.
\begin{equation}
\label{eq:CNC}
\C(N) = \frac{|F|}{|P| + |T|}
\end{equation}
Figure~\ref{fig:cnc-examples} shows example nets and their respective complexity scores.
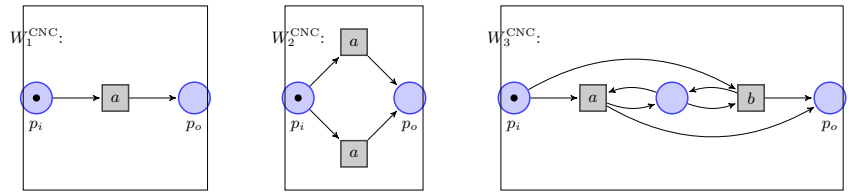
\begin{figure}[ht]
\begin{center}
\begin{minipage}{0.2525\textwidth}
\centering
\scalebox{\scalefactor}{
\begin{tikzpicture}[node distance = 1.5cm,>=stealth',bend angle=0,auto]
	\node [place,tokens=1] (start) [label=below:$p_i$] {};
	\node [transition] (t1) [right of=start,label=center:$a$] {}
	edge [pre] (start);
	\node [place] (p1) [right of=t1,label=below:$p_o$] {}
	edge [pre] (t1);
	\node at (0,1.15) {$W_1^{\name}$:};
	\draw[opacity=0] ($(start)-(0.25,1.75)$) rectangle ($(p1) + (0.25,1.75)$);
\end{tikzpicture}}
\end{minipage}
\begin{minipage}{0.25\textwidth}
\centering
\scalebox{\scalefactor}{
\begin{tikzpicture}[node distance = 1.5cm,>=stealth',bend angle=0,auto]
	\node [place,tokens=1] (start) [label=below:$p_i$] {};
	\node [transition] (t1) [above right of=start,label=center:$a$] {}
	edge [pre] (start);
	\node [transition] (t2) [below right of=start,label=center:$a$] {}
	edge [pre] (start);
	\node [place] (p1) [below right of=t1,label=below:$p_o$] {}
	edge [pre] (t1)
	edge [pre] (t2);
	\node at (0,1.15) {$W_2^{\name}$:};
	\draw[opacity=0] ($(start)-(0.25,1.75)$) rectangle ($(p1) + (0.25,1.75)$);
\end{tikzpicture}}
\end{minipage}
\begin{minipage}{0.425\textwidth}
\centering
\scalebox{\scalefactor}{
\begin{tikzpicture}[node distance = 1.5cm,>=stealth',bend angle=0,auto]
	\node [place,tokens=1] (start) [label=below:$p_i$] {};
	\node [transition] (t1) [right of=start,label=center:$a$] {}
	edge [pre] (start);
	\node [place] (p1) [right of=t1] {}
	edge [pre,bend left=20] (t1)
	edge [post,bend right=20] (t1);
	\node [transition] (t2) [right of=p1,label=center:$b$] {}
	edge [pre,bend left=20] (p1)
	edge [post,bend right=20] (p1)
	edge [pre,bend right=30] (start);
	\node [place] (p2) [right of=t2,label=below:$p_o$] {}
	edge [pre] (t2)
	edge [pre,bend left=30] (t1);
	\node at (0,1.15) {$W_3^{\name}$:};
	\draw[opacity=0] ($(start)-(0.25,1.75)$) rectangle ($(p2) + (0.25,1.75)$);
\end{tikzpicture}}
\end{minipage}
\end{center}
\caption{Three workflow nets, $W_1^{\name}$, $W_2^{\name}$, $W_3^{\name}$, with $\C(W_1^{\name}) = \frac{2}{3}$, $\C(W_2^{\name}) = 1$ and $\C(W_3^{\name}) = \frac{8}{5}$.}
\label{fig:cnc-examples}
\end{figure}

\begin{description}
\propitemf{\propone}{\yes}
For the two nets $W_1^{\name}$ and $W_2^{\name}$ of Figure~\ref{fig:cnc-examples}, we get the complexity scores $\C(W_1^{\name}) = \frac{2}{3} \neq 1 = \C(W_2^{\name})$.

\propitemf{\proptwo}{\no} 
Figure~\ref{fig:cnc-fin} shows how to construct infinitely many workflow nets with complexity $1$ according to $\C$, so $\C$ does not fulfill this property.
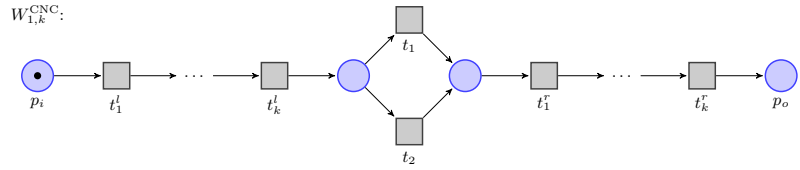
\begin{figure}[ht]
\begin{center}
\scalebox{\scalefactor}{
\begin{tikzpicture}[node distance = 1.5cm,>=stealth',bend angle=0,auto]
	\node [place,tokens=1] (start) [label=below:$p_i$] {};
	\node [transition] (t1l) [right of=start,label=below:$t_1^l$] {}
	edge [pre] (start);
	\node (dots1) [right of=t1l] {$\dots$}
	edge [pre] (t1l);
	\node [transition] (tkl) [right of=dots1,label=below:$t_k^l$] {}
	edge [pre] (dots1);
	\node [place] (p1) [right of=tkl] {}
	edge [pre] (tkl);
	\node [transition] (t1) [above right of=p1,label=below:$t_1$] {}
	edge [pre] (p1);
	\node [transition] (t2) [below right of=p1,label=below:$t_2$] {}
	edge [pre] (p1);
	\node [place] (p2) [below right of=t1] {}
	edge [pre] (t1)
	edge [pre] (t2);
	\node [transition] (t1r) [right of=p2,label=below:$t_1^r$] {}
	edge [pre] (p2);
	\node (dots2) [right of=t1r] {$\dots$}
	edge [pre] (t1r);
	\node [transition] (tkr) [right of=dots2,label=below:$t_k^r$] {}
	edge [pre] (dots2);
	\node [place] (end) [right of=tkr,label=below:$p_o$] {}
	edge [pre] (tkr);
	\node at (0,1.15) {$W_{1,k}^{\name}$:};
\end{tikzpicture}}
\end{center}
\caption{A construction plan for workflow nets with $4k + 4$ arcs, $2k + 2$ places and $2k + 2$ transitions. The complexity score of this net is $\C(W_{1,k}^{\name}) = \frac{4k + 4}{4k + 4} = 1$.}
\label{fig:cnc-fin}
\end{figure}

\propitemf{\propthree}{\yes} 
Figure~\ref{fig:cnc-fin} shows that there are at least two workflow nets with complexity $1$ according to $\C$. 

\propitemf{\propfour}{\yes} 
Take the the workflow nets $W_1^{\name}$ and $W_2^{\name}$ of Figure~\ref{fig:cnc-examples}.
Their languages are $L(W_1^{\name}) = \{\varepsilon, a\} = L(W_2^{\name})$, but their complexity scores are $\C(W_1^{\name}) = \frac{2}{3} \neq 1 = \C(W_2^{\name})$.

\propitemf{\propfive}{\no} 
Take the workflow nets $W_1^{\name}$ and $W_3^{\name}$ of Figure~\ref{fig:cnc-examples}. We get:
\begin{itemize}
\item $\C(W_3^{\name} \seqop W_1^{\name}) = \frac{12}{5 + 4} = \frac{12}{9} = 1.\overline{3} < 1.6 = \frac{8}{5} = \C(W_3^{\name})$,
\item $\C(W_3^{\name} \parop W_1^{\name}) = \frac{16}{7 + 5} = \frac{16}{12} = 1.\overline{3} < 1.6 = \frac{8}{5} = \C(W_3^{\name})$,
\item $\C(W_3^{\name} \choiceop W_1^{\name}) = \frac{18}{7 + 7} = \frac{9}{7} \approx 1.29 < 1.6 = \frac{8}{5} = \C(W_3^{\name})$,
\item $\C(W_3^{\name} \loopop W_1^{\name}) = \frac{22}{9 + 9} = \frac{11}{9} = 1.\overline{2} < 1.6 = \frac{8}{5} = \C(W_3^{\name})$.
\end{itemize}

\propitemf{\propsix}{\yes} 
Take the workflow nets $W_2^{\name}$ and $W_3^{\name}$ of Figure~\ref{fig:cnc-examples}, as well as the workflow net $W_{1,1}^{\name}$ of Figure~\ref{fig:cnc-fin}. 
The workflow nets $W_2^{\name}$ and $W_{1,1}^{\name}$ have the same complexity score $\C(W_2^{\name}) = \C(W_{1,1}^{\name}) = 1$. 
Combining them with the net $W_3^{\name}$, however, we get:
\begin{itemize}
\item $\C(W_2^{\name} \seqop W_3^{\name}) = \frac{14}{5 + 5} = 1.4 $ \\
\phantom{$\C(W_2^{\name} \seqop W_3^{\name})$}
$\neq 1.286 \approx \frac{18}{7 + 7} = \C(W_{1,1}^{\name} \seqop W_3^{\name})$,
\item $\C(W_2^{\name} \parop W_3^{\name}) = \frac{18}{7 + 6} \approx 1.385 $ \\
\phantom{$\C(W_2^{\name} \parop W_3^{\name})$}
$\neq 1.294 \approx \frac{22}{9 + 8} = \C(W_{1,1}^{\name} \parop W_3^{\name})$,
\item $\C(W_2^{\name} \choiceop W_3^{\name}) = \frac{20}{7 + 8} = 1.\overline{3} $ \\
\phantom{$\C(W_2^{\name} \choiceop W_3^{\name})$ }
$\neq 1.263 \approx \frac{24}{9 + 10} = \C(W_{1,1}^{\name} \choiceop W_3^{\name})$,
\item $\C(W_2^{\name} \loopop W_3^{\name}) = \frac{24}{9 + 10} \approx 1.263 $ \\
\phantom{$\C(W_2^{\name} \loopop W_3^{\name})$}
$\neq 1.217 \approx \frac{28}{11 + 12} = \C(W_{1,1}^{\name} \loopop W_3^{\name})$.
\end{itemize}

\propitemf{\propseven}{\yes} 
The workflow net $W_3^{\name}$ of Figure~\ref{fig:cnc-examples} is a permutation of the net $W_{<1,2}^{\name}$ of Figure~\ref{fig:cnc-inf}, but $\C(W_3^{\name}) = \frac{8}{5} \neq \frac{4}{5} = \C(W_{<1,2}^{\name})$.

\propitemf{\propeight}{\yes} 
$\C$ is only dependent on the number of arcs, places and transitions, so performing a relabeling on the transitions cannot have any impact on the complexity score of $\C$.

\propitemf{\propnine}{\no} 
Since we know how many nodes and edges we add with each operation, we can argue over general graphs instead of Petri nets to show that $\C$ is subadditive.
\begin{theorem}
\label{thm:cnc-sub-seq}
Let $G_1$ be a graph with $n_1$ nodes and $m_1 \geq n_1 - 1$ edges, $G_2$ be a graph with $n_2$ nodes and $m_2 \geq n_2 - 1$ edges, and $G$ be a graph with $n = n_1 + n_2 + 1$ nodes and $m = m_1 + m_2 + 2$ edges.
Then, $\frac{m_1}{n_1} + \frac{m_2}{n_2} > \frac{m}{n}$.
\end{theorem}
\begin{proof}
\begin{align*}
&\frac{m_1}{n_1} + \frac{m_2}{n_2} > \frac{m_1 + m_2 + 2}{n_1 + n_2 + 1} \\
\Leftrightarrow &(m_1 n_1 + m_2 n_1) \cdot (n_1 + n_2 + 1) > (m_1 + m_2 + 2) \cdot (n_1 n_2) \\
\Leftrightarrow &\underbrace{m_1 n_2^2}_{> n_2} + \underbrace{m_2 n_1^2}_{> n_1} + \underbrace{m_1}_{\geq n_1 - 1} n_2 + \underbrace{m_2}_{\geq n_2 - 1} n_1 > 2 n_1 n_2
\end{align*}
\qed
\end{proof}
Theorem~\ref{thm:cnc-sub-seq} directly implies that $\C$ is subadditive for $\seqop$.
\begin{theorem}
\label{thm:cnc-sub}
Let $G_1$ be a graph with $n_1 \geq 3$ nodes and $m_1 \geq n_1 - 1$ edges, $G_2$ be a graph with $n_2 \geq 3$ nodes and $m_2 \geq n_2 - 1$ edges, and $G$ be a graph with $n = n_1 + n_2 + k$ nodes and $m = m_1 + m_2 + k + 2$ edges for $k \in \{4, 6, 10\}$.
Then, $\frac{m_1}{n_1} + \frac{m_2}{n_2} > \frac{m}{n}$.
\end{theorem}
\begin{proof}
\begin{align*}
&\frac{m_1}{n_1} + \frac{m_2}{n_2} > \frac{m_1 + m_2 + k + 2}{n_1 + n_2 + k} \\
\Leftrightarrow &\underbrace{m_1 n_2^2}_{\geq 6 n_2 (*)} + k \underbrace{m_1}_{\geq n_1 - 1} n_2 + \underbrace{m_2 n_1^2}_{\geq 6 n_1 (*)} + k \underbrace{m_1}_{\geq n_2 - 1} n_1 > (k + 2) n_1 n_2 \\
\end{align*}
Where $(*)$ is true because $n_1, n_2 \geq 3$ and $m_1, m_2 \geq 2$.
The lower inequality is true since
\begin{align*}
&k(n_1 - 1)n_2 + 6n_1 + k(n_2 - 1)n_1 + 6n_2 \\
=\hspace*{1mm} &2k n_1 n_2 - (k - 6)n_1 - (k - 6)n_2 \\
>\hspace*{1mm} &(2k - 2(k - 6))n_1 n_2 \\
=\hspace*{1mm} &12 n_1 n_2 \\
\overset{k \leq 10}{\geq} &(k + 2)n_1 n_2 
\end{align*}
\qed
\end{proof}

\propitemf{\propdef}{\yes} 
The coefficient of network connectivity is defined for all workflow nets, since by definition $|P| \geq 2$ and $|T| \geq 1$ in every workflow net. 
Furthermore, $|F| \geq 0$, so the measure returns only non-negative values.

\propitemf{\propmin}{\yes} 
In a workflow net, every node needs to lie on a path from the input place to the output place, so $|F| \geq |P| + |T| - 1$ .
Since $\frac{n-1}{n}$ is strictly increasing with increasing $n$, the smallest possible value of $\C$ is $\frac{2}{3}$. 
This smallest possible complexity score is rewarded to $\netminwf$ of Figure~\ref{fig:netminwf}.

\propitemf{\propinf}{\yes} 
Figure~\ref{fig:cnc-inf} shows how to construct infinitely many workflow nets with different complexity scores.
Thus, $\C$ doesn't fulfill this property, since $|\{c \in \mathbb{R} \mid \exists M \in \mathcal{M}: \C(M) = c\}| \geq |\{\frac{2k}{2k + 1} \mid k \in \mathbb{N}\}| = |\mathbb{N}| = \infty$.
\begin{figure}[ht]
\begin{center}
\scalebox{\scalefactor}{
\begin{tikzpicture}[node distance = 1.5cm,>=stealth',bend angle=0,auto]
	\node [place,tokens=1] (start) [label=below:$p_i$] {};
	\node [transition] (t1) [right of=start,label=below:$t_1$] {}
	edge [pre] (start);
	\node [place] (p1) [right of=t1,label=below:$p_1$] {}
	edge [pre] (t1); 
	\node (dots1) [right of=p1] {$\dots$}
	edge [pre] (p1);
	\node [transition] (tk) [right of=dots1,label=below:$t_k$] {}
	edge [pre] (dots1);
	\node [place] (end) [right of=tk,label=below:$p_o$] {}
	edge [pre] (tk);
	\node at (0,1.15) {$W_{<1,k}^{\name}$:};
\end{tikzpicture}}
\end{center}
\caption{A construction plan for workflow nets with $2k$ arcs, $k + 1$ places and $k$ transitions. The complexity score of this net is $\C(W_{<1,k}^{\name}) = \frac{2k}{2k + 1}$.}
\label{fig:cnc-inf}
\end{figure}
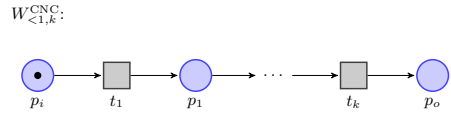

\propitemf{\propnotsup}{\yes} The analysis of the property {\propfive} shows that $\C$ is not superadditive.

\propitemf{\propadd}{\no} The analysis of the property {\propfive} shows that $\C$ is not additive.
\end{description}

\def\name{\densityname}
\def\C{\density}
\subsubsection{Density}
The density~\cite[p.120]{Men08} of a process model is the number of its arcs divided by the total number of possible arcs.
Since a workflow net is a bipartite graph, the maximum possible number of arcs is $|P| \cdot |T| - 2 \cdot |T|$. 
We therefore define this complexity measure as follows:
\begin{equation}
\label{eq:density}
\C(N) = \frac{|F|}{2 \cdot |P| \cdot |T| - 2 \cdot |T|} = \frac{|F|}{2 \cdot |T| \cdot (|P| - 1)}.
\end{equation}
Figure~\ref{fig:density-examples} shows two simple example nets and their respective complexity scores.
\begin{figure}[ht]
\begin{center}
\begin{minipage}{0.4\textwidth}
\centering
\scalebox{\scalefactor}{
\begin{tikzpicture}[node distance = 1.5cm,>=stealth',bend angle=0,auto]
	\node [place,tokens=1] (start) [label=below:$p_i$] {};
	\node [transition] (t1) [right of=start,label=center:$a$] {}
	edge [pre] (start);
	\node [place] (p1) [right of=t1,label=below:$p_o$] {}
	edge [pre] (t1);
	\node at (0,1.15) {$W_1^{\name}$:};
\end{tikzpicture}}
\end{minipage}
\begin{minipage}{0.55\textwidth}
\centering
\scalebox{\scalefactor}{
\begin{tikzpicture}[node distance = 1.5cm,>=stealth',bend angle=0,auto]
	\node [place,tokens=1] (start) [label=below:$p_i$] {};
	\node [transition] (t1) [right of=start,label=center:$\tau$] {}
	edge [pre] (start);
	\node [place] (p1) [right of=t1] {}
	edge [pre] (t1);
	\node [transition] (t2) [right of=p1,label=center:$a$] {}
	edge [pre] (p1);
	\node [place] (end) [right of=t2] {}
	edge [pre] (t2);
	\node at (0,1.15) {$W_2^{\name}$:};
\end{tikzpicture}}
\end{minipage}
\end{center}
\caption{Two workflow nets, $W_1^{\name}$ and $W_2^{\name}$, with $\C(W_1^{\name}) = 1$ and $\C(W_2^{\name}) = 0.5$.}
\label{fig:density-examples}
\end{figure}

\begin{description}
\propitemf{\propone}{\yes} 
For the two nets $W_1^{\name}$ and $W_2^{\name}$ of Figure~\ref{fig:density-examples}, we get the complexity scores $\C(W_1^{\name}) = 1 \neq 0.5 = \C(W_2^{\name})$.

\propitemf{\proptwo}{\no} 
Figure~\ref{fig:density-fin} shows how to construct infinitely many workflow nets with complexity $1$, so $\C$ does not fulfill this property.
\begin{figure}[ht]
\begin{center}
\scalebox{\scalefactor}{
\begin{tikzpicture}[node distance = 1.5cm,>=stealth',bend angle=0,auto]
	\node [place,tokens=1] (start) [label=below:$p_i$] {};
	\node (dots) [right of=start] {$\vdots$};
	\node [transition] (t1) [above of=dots,label=below:$t_1$] {}
	edge [pre] (start);
	\node [transition] (tk) [below of=dots,label=below:$t_k$] {}
	edge [pre] (start);
	\node [place] (end) [right of=dots,label=below:$p_o$] {}
	edge [pre] (t1)
	edge [pre] (tk);
	\node at (0,1.15) {$W_{1,k}^{\name}$:};
\end{tikzpicture}}
\end{center}
\caption{A workflow net $W_{1,k}^{\name}$ with $2$ places, $k$ transitions and $2k$ edges. For this net and some $k \geq 1$, we get $\C(W_{1,k}^{\name}) = \frac{2k}{2 \cdot k \cdot (2-1)} = 1$.}
\label{fig:density-fin}
\end{figure}
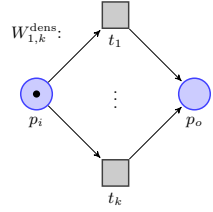

\propitemf{\propthree}{\yes} 
Figure~\ref{fig:density-fin} shows that two structurally different workflow nets can get the same complexity score from $\C$.

\propitemf{\propfour}{\yes} 
Take the workflow nets $W_1^{\name}$ and $W_2^{\name}$ of Figure~\ref{fig:density-examples}. 
Their languages are $L(W_1^{\name}) = \{\varepsilon, a\} = L(W_2^{\name})$, but their complexity scores are $\C(W_1^{\name}) = 1 \neq 0.5 = \C(W_2^{\name})$.

\propitemf{\propfive}{\no} Consider the net $W_1^{\name}$ of Figure~\ref{fig:density-examples}. We get:
\begin{itemize}
\item $\C(W_1^{\name} \seqop W_1^{\name}) = \frac{6}{2 \cdot 3 \cdot (4-1)} = \frac{1}{3} < 1 = \C(W_1^{\name})$,
\item $\C(W_1^{\name} \parop W_1^{\name}) = \frac{10}{2 \cdot 4 \cdot (6-1)} = \frac{1}{4} < 1 = \C(W_1^{\name})$,
\item $\C(W_1^{\name} \choiceop W_1^{\name}) = \frac{12}{2 \cdot 6 \cdot (6-1)} = \frac{1}{5} < 1 = \C(W_1^{\name})$,
\item $\C(W_1^{\name} \loopop W_1^{\name}) = \frac{16}{2 \cdot 8 \cdot (8-1)} = \frac{1}{7} < 1 = \C(W_1^{\name})$.
\end{itemize}
Therefore, $\C$ does not fulfill this property for any of these operations.

\propitemf{\propsix}{\yes} 
Consider the nets $W_1^{\name}$ and $W_2^{\name}$ of Figure~\ref{fig:density-examples} as well as the net $W_3^{\name}$ of Figure~\ref{fig:density-perm}.
We get:
\begin{itemize}
\item $\C(W_2^{\name} \seqop W_1^{\name}) = \frac{8}{2 \cdot 4 \cdot (5 - 1)} = \frac{1}{4}$ \\
\phantom{$\C(W_2^{\name} \seqop W_1^{\name})$ }
$\neq \frac{13}{50} = \frac{13}{2 \cdot 5 \cdot (6-1)} = \C(W_3^{\name} \seqop W_1^{\name})$,
\item $\C(W_2^{\name} \parop W_1^{\name}) = \frac{12}{2 \cdot 5 \cdot (7 - 1)} = \frac{1}{5}$ \\
\phantom{$\C(W_2^{\name} \parop W_1^{\name})$ }
$\neq \frac{17}{84} = \frac{17}{2 \cdot 6 \cdot (8-1)} = \C(W_3^{\name} \parop W_1^{\name})$,
\item $\C(W_2^{\name} \choiceop W_1^{\name}) = \frac{14}{2 \cdot 7 \cdot (7 - 1)} = \frac{1}{6}$ \\
\phantom{$\C(W_2^{\name} \choiceop W_1^{\name})$ }
$\neq \frac{19}{112} = \frac{19}{2 \cdot 8 \cdot (8-1)} = \C(W_3^{\name} \choiceop W_1^{\name})$,
\item $\C(W_2^{\name} \loopop W_1^{\name}) = \frac{18}{2 \cdot 9 \cdot (9 - 1)} = \frac{1}{8}$ \\
\phantom{$\C(W_2^{\name} \loopop W_1^{\name})$ }
$\neq \frac{23}{180} = \frac{23}{2 \cdot 10 \cdot (10-1)} = \C(W_3^{\name} \loopop W_1^{\name})$.
\end{itemize}
Therefore, $\C$ fulfills this property for all of these operations.
\begin{figure}[ht]
\begin{center}
\scalebox{\scalefactor}{
\begin{tikzpicture}[node distance = 1.5cm,>=stealth',bend angle=0,auto]
	\node [place,tokens=1] (start) [label=below:$p_i$] {};
	\node [transition] (t1) [right of=start] {}
	edge [pre] (start);
	\node [place] (p1) [right of=t1] {}
	edge [pre] (t1);
	\node [transition] (t2) [right of=p1] {}
	edge [pre,bend left=20] (p1)
	edge [post,bend right=20] (p1);
	\node [place] (p2) [right of=t2] {}
	edge [pre,bend left=20] (t2)
	edge [post,bend right=20] (t2)
	edge [pre,bend right=30] (t1);
	\node [transition] (t3) [right of=p2] {}
	edge [pre] (p2);
	\node [place] (end) [right of=t3] {}
	edge [pre] (t3);
	\node at (0,0.75) {$W_3^{\name}$:};
\end{tikzpicture}}
\scalebox{\scalefactor}{
\begin{tikzpicture}[node distance = 1.5cm,>=stealth',bend angle=0,auto]
	\node [place,tokens=1] (start) [label=below:$p_i$] {};
	\node [transition] (t1) [right of=start] {}
	edge [pre] (start);
	\node [place] (p1) [right of=t1] {}
	edge [pre] (t1);
	\node [transition] (t2) [right of=p1] {}
	edge [pre] (p1);
	\node [place] (p2) [right of=t2] {}
	edge [pre] (t2);
	\node [transition] (t3) [right of=p2] {}
	edge [pre] (p2);
	\node [place] (end) [right of=t3] {}
	edge [pre] (t3);
	\node at (0,0.75) {$W_4^{\name}$:};
\end{tikzpicture}}
\end{center}
\caption{Two workflow nets, $W_3^{\name}$ and $W_4^{\name}$, where $W_3^{\name} \in Perm(W_4^{\name})$ and $\C(W_3^{\name}) = \frac{9}{2 \cdot 3 \cdot (4-1)} = 0.5$ and $\C(W_4^{\name}) = \frac{6}{2 \cdot 3 \cdot (4-1)} = \frac{1}{3}$.}
\label{fig:density-perm}
\end{figure}

\propitemf{\propseven}{\yes} 
Figure~\ref{fig:density-perm} shows two workflow nets, $W_3^{\name}$ and $W_4^{\name}$, that are permutations of each other, but $\C(W_3^{\name}) = 0.5 \neq 0.\overline{3} = \C(W_4^{\name})$.

\propitemf{\propeight}{\yes} 
$\C$ depends only on the number of arcs, places and transitions. 
Therefore, the labeling has no impact on the values returned by $\C$.

\propitemf{\propnine}{\no} 
Let $M_1, M_2 \in \mathcal{M}$. 
For $i \in \{1,2\}$, let $f_i$ be the number of arcs in $M_i$, $t_i$ be the number of transitions in $M_i$ and $p_i$ be the number of places in $M_i$.
By arithmetic rules, we have: 
\begin{align*}
\C(M_1) + \C(M_2) &= \frac{f_1}{2t_1p_1 - 2t_1} + \frac{f_2}{2t_2p_2 - 2t_2} \\
&> \frac{f_1 + f_2}{2t_1p_1 - 2t_1 + 2t_2p_2 - 2t_2} =: \frac{n}{d}
\end{align*}
\begin{itemize}
\item The operator $\seqop$ introduces $1$ transition and $2$ arcs.
Therefore,
\[\C(M_1 \seqop M_2) = \frac{n + 2}{d + 2t_1p_2 + 2t_2p_1 + 2p_1 + 2p_2 - 2} < \frac{n}{d}\]
since $2t_1p_2 + 2t_2p_1 + 2p_1 + 2p_2 - 2 > 2$ ($p_i \geq 2$ and $t_i \geq 1$ for all $i \in \{1,2\}$ by definition of workflow nets).
\item The operator $\parop$ introduces $2$ places, $2$ transitions and $6$ arcs.
Therefore,
\[\C(M_1 \parop M_2) = \frac{n + 6}{d + 2t_1p_2 + 4t_1 + 2t_2p_1 + 4t_2 + 4p_1 + 4p_2 + 4} < \frac{n}{d}\]
since $2t_1p_2 + 4t_1 + 2t_2p_1 + 4t_2 + 4p_1 + 4p_2 + 4 > 6$.
\item The operator $\choiceop$ introduces $2$ places, $4$ transitions and $8$ arcs.
Therefore,
\[\C(M_1 \choiceop M_2) = \frac{n + 8}{d + 2t_1p_2 + 6t_1 + 2t_2p_1 + 6t_2 + 8p_1 + 8p_2 + 16} < \frac{n}{d}\]
since $2t_1p_2 + 6t_1 + 2t_2p_1 + 6t_2 + 8p_1 + 8p_2 + 16 > 8$.
\item The operator $\loopop$ introduces $4$ places, $6$ transitions and $12$ arcs.
Therefore,
\begin{align*}
\C(M_1 \loopop M_2) &= \frac{n + 12}{d + 2t_1p_2 + 8t_1 + 2t_2p_1 + 8t_2 + 12p_1 + 12p_2 + 48}
\end{align*}
which is $< \frac{n}{d}$, since $2t_1p_2 + 8t_1 + 2t_2p_1 + 8t_2 + 12p_1 + 12p_2 + 48 > 12$.
\end{itemize}
Thus, $\C(M_1 \oplus M_2) < \C(M_1) + \C(M_2)$ for all $\oplus \in \{\seqop, \parop, \choiceop, \loopop\}$ and $\C$ is subadditive for any of these operations.

\propitemf{\propdef}{\yes} 
The density of a workflow net is always defined, since thre are at least two places and one transition in a workflow net.
We therefore cannot divide by zero while evaluating the complexity score.
Furthermore, since $|F| \geq 0$ for every workflow net, the values of the complexity metric are non-negative.

\propitemf{\propmin}{\no} 
$\C$ has no minimum value that can be reached by a workflow net: 
Consider the net $W_{\min,k}^{\name}$ of Figure~\ref{fig:density-min}.
This workflow net has a complexity of $\C(W_{\min,k}^{\name}) = \frac{1}{k}$.
Therefore, for any workflow net $M$ with complexity $c = \C(M) > 0$ we can find a workflow net with smaller complexity score by setting $k > \frac{1}{c}$.
A workflow net with complexity $0$ cannot exist, since all workflow nets contain at least one arc.
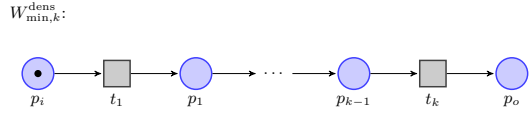
\begin{figure}[ht]
\begin{center}
\scalebox{\scalefactor}{
\begin{tikzpicture}[node distance = 1.5cm,>=stealth',bend angle=0,auto]
	\node [place,tokens=1] (start) [label=below:$p_i$] {};
	\node [transition] (t1) [right of=start,label=below:$t_1$] {}
	edge [pre] (start);
	\node [place] (p1) [right of=t1,label=below:$p_1$] {}
	edge [pre] (t1);
	\node (dots) [right of=p1] {$\dots$}
	edge [pre] (p1);
	\node [place] (pk1) [right of=dots,label=below:$p_{k-1}$] {}
	edge [pre] (dots);
	\node [transition] (tk) [right of=pk1,label=below:$t_k$] {}
	edge [pre] (pk1);
	\node [place] (end) [right of=tk,label=below:$p_o$] {}
	edge [pre] (tk);
	\node at (0,1.15) {$W_{\min,k}^{\name}$:};
\end{tikzpicture}}
\end{center}
\caption{A workflow net $W_{\min,k}^{\name}$ with $k + 1$ places, $k$ transitions and $2k$ edges. For this net and some $k \geq 1$, we get $\C(W_{\min,k}^{\name}) = \frac{2k}{2 \cdot k \cdot (k+1-1)} = \frac{1}{k}$.}
\label{fig:density-min}
\end{figure}

\propitemf{\propinf}{\yes} 
Figure~\ref{fig:density-min} shows the construction of infinitely many workflow nets with pairwise distinct complexity according to $\C$.
With this construction, we get $|\{c \in \mathbb{R} \mid \exists M \in \mathcal{M}: \C(M) = c\}| \geq |\mathbb{N}| = \infty$.

\propitemf{\propnotsup}{\yes} The analysis of property {\propfive} shows that $\C$ is not superadditive.

\propitemf{\propadd}{\no} The analysis of property {\propfive} shows that $\C$ is not additive.
\end{description}

\subsection{Other Complexity Measures}
Lieben et al.~\cite{LieDJJ18} highlight that there are two complexity measures that don't fit into any of their discovered dimensions.
These complexity measures are the number of empty sequence flows and the number of duplicate tasks.
The first evaluates how often a concurrent connector was used in the model that doesn't model concurrent behavior.
The latter counts how often a transition label is reused in the model.

\def\name{\duplicatename}
\def\C{\duplicate}
\subsubsection{Number of Duplicate Tasks}
The possibility to repeat transition labels often facilitates finding a compact model for a language.
However, duplicate labels can also deteriorate the understandability of a process model.
If activity labels are repeated throughout the net, it is difficult to see in which cases the activity can occur.
Therefore, the amount of label repetitions in a process model can be used as a complexity measure~\cite{LaWMHRA11}.
Let $A$ be an alphabet and $W = (P,T,F,\ell,p_i,p_o)$ be a labeled workflow net with labels from $A$.
For a label $a \in A$, the set of transitions in $W$ with that label is $\ell^{-1}(a) := \{t \in T \mid \ell(t) = a\}$.
We define:
\begin{equation}
\label{eq:dup}
\C(W) = \sum_{a \in A} \left(\max\{|\ell^{-1}(a)|, 1\} - 1\right)
\end{equation}
Figure~\ref{fig:dup-examples} shows example nets and their respective complexity scores.
\begin{figure}[ht]
\begin{center}
\begin{minipage}{0.2525\textwidth}
\centering
\scalebox{\scalefactor}{
\begin{tikzpicture}[node distance = 1.5cm,>=stealth',bend angle=0,auto]
	\node [place,tokens=1] (start) [label=below:$p_i$] {};
	\node [transition] (t1) [right of=start,label=center:$a$] {}
	edge [pre] (start);
	\node [place] (p1) [right of=t1,label=below:$p_o$] {}
	edge [pre] (t1);
	\node at (0,1.15) {$W_1^{\name}$:};
	\draw[opacity=0] ($(start)-(0.25,1.75)$) rectangle ($(p1) + (0.25,1.75)$);
\end{tikzpicture}}
\end{minipage}
\begin{minipage}{0.25\textwidth}
\centering
\scalebox{\scalefactor}{
\begin{tikzpicture}[node distance = 1.5cm,>=stealth',bend angle=0,auto]
	\node [place,tokens=1] (start) [label=below:$p_i$] {};
	\node [transition] (t2) [right of=start,label=center:$a$] {}
	edge [pre] (start);
	\node [transition] (t1) [above of=t2,label=center:$a$] {}
	edge [pre] (start);
	\node [transition] (t3) [below of=t2,label=center:$a$] {}
	edge [pre] (start);
	\node [place] (p1) [right of=t2,label=below:$p_o$] {}
	edge [pre] (t1)
	edge [pre] (t2)
	edge [pre] (t3);
	\node at (0,1.15) {$W_2^{\name}$:};
	\draw[opacity=0] ($(start)-(0.25,1.75)$) rectangle ($(p1) + (0.25,1.75)$);
\end{tikzpicture}}
\end{minipage}
\begin{minipage}{0.425\textwidth}
\centering
\scalebox{\scalefactor}{
\begin{tikzpicture}[node distance = 1.5cm,>=stealth',bend angle=0,auto]
	\node [place,tokens=1] (start) [label=below:$p_i$] {};
	\node [transition] (t1) [right of=start,label=center:$\tau$] {}
	edge [pre] (start);
	\node [place] (p1) [right of=t1] {}
	edge [pre] (t1);
	\node [transition] (t2) [right of=p1,label=center:$a$] {}
	edge [pre] (p1);
	\node [place] (p2) [right of=t2,label=below:$p_o$] {}
	edge [pre] (t2);
	\node at (0,1.15) {$W_3^{\name}$:};
	\draw[opacity=0] ($(start)-(0.25,1.75)$) rectangle ($(p2) + (0.25,1.75)$);
\end{tikzpicture}}
\end{minipage}
\end{center}
\caption{Three workflow nets, $W_1^{\name}$, $W_2^{\name}$ and $W_3^{\name}$, with $\C(W_2^{\name}) = 2$ and $\C(W_1^{\name}) = \C(W_3^{\name}) = 0$.}
\label{fig:dup-examples}
\end{figure}

\begin{description}
\propitemf{\propone}{\yes} 
For the two workflow nets $W_1^{\name}$ and $W_2^{\name}$ of Figure~\ref{fig:dup-examples}, we get the complexity scores $\C(W_1^{\name}) = 0 \neq 2 = \C(W_2^{\name})$.

\propitemf{\proptwo}{\yes} 
$\C$ gives only non-negative integer values as complexity scores. 
Let $A$ be a finite set of activities and $c \in \mathbb{N}_0$ be a complexity score. 
A workflow net with complexity score $c$ has $c$ labels that already appeared in the net. 
There are $c \cdot |A|$ possibilities, which label these duplicate transitions have.
But in total, the net cannot contain more than $c \cdot |A| + |A| + 1$ transitions, since otherwise there would be more than $c$ duplicates in the net.
Therefore, there can only be finitely many workflow nets with complexity $c$. 

\propitemf{\propthree}{\yes} 
The workflow nets $W_1^{\name}$ and $W_3^{\name}$ of Figure~\ref{fig:dup-examples} are different in structure, but both get complexity score $\C(W_1^{\name}) = 0 = \C(W_3^{\name})$.

\propitemf{\propfour}{\yes} 
Take the workflow nets $W_1^{\name}$ and $W_2^{\name}$ of Figure~\ref{fig:dup-examples}.
Their languages are $L(W_1^{\name}) = L(W_2^{\name}) = \{\varepsilon, a\}$, but their complexity scores are $\C(W_1^{\name}) = 0 \neq 2 = \C(W_2^{\name})$.

\propitemf{\propfive}{\yes} 
$\C$ is monotone for any operation $\oplus \in \{\seqop, \parop, \choiceop, \loopop\}$, since for any $M_1, \dots, M_n \in \mathcal{M}$, a model $\oplus(M_1, \dots, M_n)$ cannot have fewer duplicate labels than any of the models $M_1, \dots, M_n$ has.

\propitemf{\propsix}{\yes} 
Take the workflow net $W_1^{\name}$ of Figure~\ref{fig:dup-examples}.
We construct a workflow net $W_{1,1}^{\name}$ by taking $W_1^{\name}$ and changing the label of its only transition to $b$.
By setting the labels of all transitions that are newly introduced by one of the operations to $\tau$, we get:
\begin{itemize}
\item $\C(W_1^{\name} \seqop W_1^{\name}) = 1 \neq 0 = \C(W_{1,1}^{\name} \seqop W_1^{\name})$,
\item $\C(W_1^{\name} \parop W_1^{\name}) = 2 \neq 1 = \C(W_{1,1}^{\name} \parop W_1^{\name})$,
\item $\C(W_1^{\name} \choiceop W_1^{\name}) = 4 \neq 3 = \C(W_{1,1}^{\name} \choiceop W_1^{\name})$,
\item $\C(W_1^{\name} \loopop W_1^{\name}) = 6 \neq 5 = \C(W_{1,1}^{\name} \loopop W_1^{\name})$.
\end{itemize}

\propitemf{\propseven}{\no} 
By our definition of permutations, we must keep the transitions and their labels in a permutation of a net.
This means that the complexity score remains the same for all permutations of a workflow net.

\propitemf{\propeight}{\yes} 
Uniformly changing the labels of a workflow net can only change the set of duplicate tasks, but not the number of them, since a relabeling is a bijection by our definition. 
Therefore, $\C$ fulfills this property.

\propitemf{\propnine}{\yes} 
Take the workflow net $W_1^{\name}$ of Figure~\ref{fig:dup-examples}. 
If we give all transitions that are newly introduced by one of the operations the label $\tau$, we get:
\begin{itemize}
\item $\C(W_1^{\name} \seqop W_1^{\name}) = 1 > 0 + 0 = \C(W_1^{\name}) + \C(W_1^{\name})$,
\item $\C(W_1^{\name} \parop W_1^{\name}) = 2 > 0 + 0 = \C(W_1^{\name}) + \C(W_1^{\name})$,
\item $\C(W_1^{\name} \choiceop W_1^{\name}) = 4 > 0 + 0 = \C(W_1^{\name}) + \C(W_1^{\name})$,
\item $\C(W_1^{\name} \loopop W_1^{\name}) = 6 > 0 + 0 = \C(W_1^{\name}) + \C(W_1^{\name})$.
\end{itemize}

\propitemf{\propdef}{\yes} 
This metric is obviously defined for every labeled workflow net. 
Because we take the maximum of $|\ell^{-1}(a)|$ and $1$ for any $a \in A$, every summand of the complexity score is non-negative. 
Therefore, the total score of this complexity measure is non-negative.

\propitemf{\propmin}{\yes} 
The workflow net $W_1^{\name}$ of Figure~\ref{fig:dup-examples} gets the smallest possible complexity score $\C(W_1^{\name}) = 0$.
Since we sum over non-negative values to calculate the complexity score, it is not possible to get a score less than $0$.

\propitemf{\propinf}{\yes} 
Let $c \in \mathbb{N}_0$. 
Figure~\ref{fig:dup-inf} shows how to construct a workflow net with complexity score $c$, so $|c \in \mathbb{R} \mid \exists M \in \mathcal{M}: \C(M) = c\} \geq \mathbb{N}_0 = \infty$.
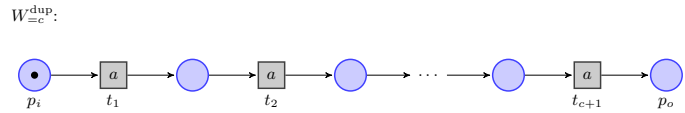
\begin{figure}[ht]
\begin{center}
\scalebox{\scalefactor}{
\begin{tikzpicture}[node distance = 1.5cm,>=stealth',bend angle=0,auto]
	\node [place,tokens=1] (start) [label=below:$p_i$] {};
	\node [transition] (t1) [right of=start,label=center:$a$,label=below:$t_1$] {}
	edge [pre] (start);
	\node [place] (p1) [right of=t1] {}
	edge [pre] (t1);
	\node [transition] (t2) [right of=p1,label=center:$a$,label=below:$t_2$] {}
	edge [pre] (p1);
	\node [place] (p2) [right of=t2] {}
	edge [pre] (t2);
	\node (dots) [right of=p2] {$\dots$}
	edge [pre] (p2);
	\node [place] (p3) [right of=dots] {}
	edge [pre] (dots);
	\node [transition] (t3) [right of=p3,label=center:$a$,label=below:$t_{c + 1}$] {}
	edge [pre] (p3);
	\node [place] (end) [right of=t3,label=below:$p_o$] {}
	edge [pre] (t3);
	\node at (0,1.15) {$W_{= c}^{\name}$:};
\end{tikzpicture}}
\end{center}
\caption{A construction plan for workflow nets $W_{=c}^{\name}$ with $c + 1$ transitions and complexity $\C(W_{=c}^{\name}) = c$.}
\label{fig:dup-inf}
\end{figure}

\propitemf{\propnotsup}{\no} 
Whether $\C$ fulfills this property depends on the set of activity labels $A$.
Let $M_1, M_2 \in \mathcal{M}$ be two models with distinct transition labels.
If $A$ is sufficiently large to give every transition that is newly introduced by $\oplus \in \{\seqop, \parop, \choiceop, \loopop\}$ a different label that doesn't occur in $M_1$ or $M_2$, we get $\C(M_1 \oplus M_2) = \C(M_1) + \C(M_2)$.
On the other hand, if $A$ is not sufficiently large to do so, we add at least one new duplicate label to $M_1 \oplus M_2$.
In this case, we get $\C(M_1 \oplus M_2) > \C(M_1) + \C(M_2)$.

\propitemf{\propadd}{\no} The analysis of property {\propnine} shows that $\C$ is not additive for any of the operations.
\end{description}

\def\name{\emptyseqname}
\def\C{\emptyseq}
\subsubsection{Number of Empty Sequence Flows}
The measure for empty sequence flows is based on work of Gruhn et al.~\cite{GruL09} on reducing the complexity of BPMN.
They found that, in BPMN models, there are often edes that directly connect and-splits and -joins. 
These edges can be removed completely, as they don't change the behavior of the net.
Regarding workflow nets, there can't be any edges between two and-connectors, since and-connectors are always transitions and Petri nets are bipartite.
Instead, empty sequence flows in workflow nets are places that have only and-splits in their preset and and-joins in their postset. 
We therefore define:
\begin{equation}
\label{eq:empty}
\C(W) = |\{p \in P \mid \pre{p} \subseteq \mathcal{S}_{\text{and}}^W \land \post{p} \subseteq \mathcal{J}_{\text{and}}^W\}|
\end{equation}
Figure~\ref{fig:empty-examples} shows two example nets and their respective complexity scores.
\begin{figure}[ht]
\begin{center}
\begin{minipage}{0.45\textwidth}
\centering
\scalebox{\scalefactor}{
\begin{tikzpicture}[node distance = 1.5cm,>=stealth',bend angle=0,auto]
	\node [place,tokens=1] (start) [label=below:$p_i$] {};
	\node [transition] (t1) [right of=start,label=center:$a$] {}
	edge [pre] (start);
	\node [place] (p1) [right of=t1] {}
	edge [pre] (t1);
	\node [transition] (t2) [right of=p1,label=center:$b$] {}
	edge [pre] (p1);
	\node [place] (end) [right of=t2,label=below:$p_o$] {}
	edge [pre] (t2);
	\node at (0,1.15) {$W_1^{\name}$:};
	\draw[opacity=0] ($(start)-(0.25,1.75)$) rectangle ($(end) + (0.25,1.75)$);
\end{tikzpicture}}
\end{minipage}
\begin{minipage}{0.45\textwidth}
\centering
\scalebox{\scalefactor}{
\begin{tikzpicture}[node distance = 1.5cm,>=stealth',bend angle=0,auto]
	\node [place,tokens=1] (start) [label=below:$p_i$] {};
	\node [transition] (t1) [right of=start,label=center:$a$] {}
	edge [pre] (start);
	\node [place] (p1) [above right of=t1] {}
	edge [pre] (t1);
	\node [place] (p2) [below right of=t1] {}
	edge [pre] (t1);
	\node [transition] (t2) [below right of=p1,label=center:$b$] {}
	edge [pre] (p1)
	edge [pre] (p2);
	\node [place] (end) [right of=t2,label=below:$p_o$] {}
	edge [pre] (t2);
	\node at (0,1.15) {$W_2^{\name}$:};
	\draw[opacity=0] ($(start)-(0.25,1.75)$) rectangle ($(end) + (0.25,1.75)$);
\end{tikzpicture}}
\end{minipage}
\end{center}
\caption{Two workflow nets, $W_1^{\name}$ and $W_2^{\name}$, with $\C(W_1^{\name}) = 0$ and $\C(W_2^{\name}) = 2$.}
\label{fig:empty-examples}
\end{figure}

\begin{description}
\propitemf{\propone}{\yes} 
For the two workflow nets $W_1^{\name}$ and $W_2^{\name}$ of Figure~\ref{fig:empty-examples}, we get the complexity scores $\C(W_1^{\name}) = 0 \neq 2 = \C(W_2^{\name})$.

\propitemf{\proptwo}{\no} 
Figure~\ref{fig:empty-fin-inf} shows how to construct infinitely many workflow nets with the same complexity score $k \in \mathbb{N}$, where $k \geq 2$.
\begin{figure}[ht]
\begin{center}
\scalebox{\scalefactor}{
\begin{tikzpicture}[node distance = 1.5cm,>=stealth',bend angle=0,auto]
	\node [place,tokens=1] (start) [label=below:$p_i$] {};
	\node [transition] (t1) [right of=start,label=below:$t_1$] {}
	edge [pre] (start);
	\node [place] (p1) [right of=t1,label=below:$p_1$] {}
	edge [pre] (t1);
	\node (dots) [right of=p1] {$\dots$}
	edge [pre] (p1);
	\node [transition] (t2) [right of=dots,label=below:$t_n$] {}
	edge [pre] (dots);
	\node [place] (p2) [right of=t2,label=below:$p_n$] {}
	edge [pre] (t2);
	\node [transition] (t3) [right of=p2,label=below:$t_{n+1}$] {}
	edge [pre] (p2);
	\node (dots2) [right of=t3] {$\vdots$};
	\node [place] (p3) [above of=dots2,label=below:$p_{n+1}$] {}
	edge [pre] (t3);
	\node [place] (p4) [below of=dots2,label=below:$p_{n+k}$] {}
	edge [pre] (t3);
	\node [transition] (t4) [right of=dots2,label=below:$t_{n+2}$] {}
	edge [pre] (p3)
	edge [pre] (p4);
	\node [place] (end) [right of=t4,label=below:$p_o$] {}
	edge [pre] (t4);
	\node at (0,1.15) {$W_{n,k}^{\name}$:};
	\draw[opacity=0] ($(start)-(0.25,1.75)$) rectangle ($(end) + (0.25,1.75)$);
\end{tikzpicture}}
\end{center}
\caption{A workflow net with $n + k + 2$ places, $n + 2$ transitions and complexity score $\C(W_{n,k}^{\name}) = k$.}
\label{fig:empty-fin-inf}
\end{figure}
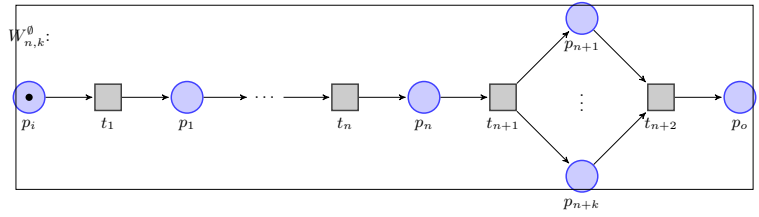

\propitemf{\propthree}{\yes} 
Let $k \in \mathbb{N}$ with $k \geq 2$. 
For any $n \in \mathbb{N}$, $W_{n,k}$ gets the complexity score $k$, so we have found two workflow nets that receive the same complexity score from $\C$.

\propitemf{\propfour}{\yes} 
Take the workflow nets $W_1^{\name}$ and $W_2^{\name}$ of Figure~\ref{fig:empty-examples}.
Their languages are $L(W_1^{\name}) = \{\varepsilon, a, ab\} = L(W_2^{\name})$, but their respective complexity scores are $\C(W_1^{\name}) = 0 \neq 2 = \C(W_2^{\name})$.

\propitemf{\propfive}{\yes} 
None of the operations of Definition~\ref{def:operations} can produce empty sequence flows, since a workflow net needs to consist of at least one transition. 
We therefore get the following Theorem:
\begin{theorem}
\label{thm:empty-add}
Let $M_1, \dots, M_n \in \mathcal{M}$ be workflow nets. Then, 
\begin{itemize}
\item $\C(\seqop(M_1, \dots, M_n) = \C(M_1) + \dots + \C(M_n)$, 
\item $\C(\parop(M_1, \dots, M_n) = \C(M_1) + \dots + \C(M_n)$, 
\item $\C(\choiceop(M_1, \dots, M_n) = \C(M_1) + \dots + \C(M_n)$, 
\item $\C(\loopop(M_1, \dots, M_n) = \C(M_1) + \dots + \C(M_n)$. 
\end{itemize}
\end{theorem}
\begin{proof}
Obvious by definition of $\C$ and the operations $\seqop, \parop, \choiceop, \loopop$. 
\qed
\end{proof}
Therefore, for any $M_1, \dots, M_n \in \mathcal{M}$ and any operation $\oplus \in \{\seqop, \parop, \choiceop, \loopop\}$, we get $\C(\oplus(M_1, \dots, M_n)) = \C(M_1) + \dots + \C(M_n) \geq \C(M_i)$ for any model $M_i \in \{M_1, \dots, M_n\}$.

\propitemf{\propsix}{\no} 
Let $M_1, M_2, M_3 \in \mathcal{M}$ be three workflow nets with complexity scores $\C(M_1) = \C(M_2)$.
Theorem~\ref{thm:empty-add} gives:
\begin{itemize}
\item $\C(M_1 \seqop M_3) = \C(M_1) + \C(M_3)$ \\
\phantom{$\C(M_1 \seqop M_3)$ }$= \C(M_2) + \C(M_3) = \C(M_2 \seqop M_3)$,
\item $\C(M_1 \parop M_3) = \C(M_1) + \C(M_3)$ \\
\phantom{$\C(M_1 \parop M_3)$ }$= \C(M_2) + \C(M_3) = \C(M_2 \parop M_3)$,
\item $\C(M_1 \choiceop M_3) = \C(M_1) + \C(M_3)$ \\
\phantom{$\C(M_1 \choiceop M_3)$ }$= \C(M_2) + \C(M_3) = \C(M_2 \choiceop M_3)$,
\item $\C(M_1 \loopop M_3) = \C(M_1) + \C(M_3)$ \\
\phantom{$\C(M_1 \loopop M_3)$ }$= \C(M_2) + \C(M_3) = \C(M_2 \loopop M_3)$.
\end{itemize}

\propitemf{\propseven}{\yes} 
The workflow nets $W_3^{\name}$ and $W_4^{\name}$ of Figure~\ref{fig:empty-perm} are permutations of each other, but $\C(W_3^{\name}) = 0 \neq 2 = \C(W_4^{\name})$.
\begin{figure}[ht]
\begin{center}
\centering
\scalebox{\scalefactor}{
\begin{tikzpicture}[node distance = 1.5cm,>=stealth',bend angle=0,auto]
	\node [place,tokens=1] (start) [label=below:$p_i$] {};
	\node [transition] (t1) [right of=start,label=below:$t_1$] {}
	edge [pre] (start);
	\node [place] (p1) [above right of=t1] {}
	edge [pre] (t1);
	\node [place] (p2) [below right of=t1] {}
	edge [pre] (t1);
	\node [transition] (t2) [right of=p1,label=below:$t_2$] {}
	edge [pre] (p1);
	\node [transition] (t3) [right of=p2,label=below:$t_3$] {}
	edge [pre] (p2);
	\node [place] (p3) [right of=t2] {}
	edge [pre] (t2);
	\node [place] (p4) [right of=t3] {}
	edge [pre] (t3);
	\node [transition] (t4) [below right of=p3,label=below:$t_4$] {}
	edge [pre] (p3)
	edge [pre] (p4);
	\node [place] (end) [right of=t4,label=below:$p_o$] {}
	edge [pre] (t4);
	\node at (0,1.15) {$W_3^{\name}$:};
	\draw[opacity=0] ($(start)-(0.25,1.75)$) rectangle ($(end) + (0.25,1.75)$);
\end{tikzpicture}}
\ \\
\centering
\scalebox{\scalefactor}{
\begin{tikzpicture}[node distance = 1.5cm,>=stealth',bend angle=0,auto]
	\node [place,tokens=1] (start) [label=below:$p_i$] {};
	\node [transition] (t1) [right of=start,label=below:$t_1$] {}
	edge [pre] (start);
	\node [place] (p1) [above right of=t1] {}
	edge [pre] (t1);
	\node [place] (p2) [below right of=t1] {}
	edge [pre] (t1);
	\node [transition] (t2) [below right of=p1,label=below:$t_2$] {}
	edge [pre] (p1)
	edge [pre] (p2);
	\node [place] (p3) [right of=t2] {}
	edge [pre] (t2);
	\node [transition] (t3) [right of=p3,label=below:$t_3$] {}
	edge [pre] (p3);
	\node [place] (p4) [right of=t3] {}
	edge [pre] (t3);
	\node [transition] (t4) [right of=p4,label=below:$t_4$] {}
	edge [pre] (p4);
	\node [place] (end) [right of=t4,label=below:$p_o$] {}
	edge [pre] (t4);
	\node at (0,1.15) {$W_4^{\name}$:};
	\draw[opacity=0] ($(start)-(0.25,1.75)$) rectangle ($(end) + (0.25,1.75)$);
\end{tikzpicture}}
\end{center}
\caption{Two workflow nets, $W_3^{\name}$ and $W_4^{\name}$, which are permutations of each other. $\C(W_3^{\name}) = 0$, but $\C(W_4^{\name}) = 2$.}
\label{fig:empty-perm}
\end{figure}

\propitemf{\propeight}{\yes} The number of empty sequence flows is independent of the labeling, so $\C$ fulfills this property.

\propitemf{\propnine}{\no} Theorem~\ref{thm:empty-add} implies that $\C$ is subadditive.

\propitemf{\propdef}{\yes} 
$\C$ is defined for any workflow net $W$, since $S_{\text{and}}^{W}$ and $J_{\text{and}}^{W}$ are.
It only returns non-negative values, since the smallest possible set has cardinality $0$.

\propitemf{\propmin}{\yes} 
The workflow net $W_1^{\name}$ of Figure~\ref{fig:empty-examples} has no empty sequence flows, so it gets a complexity score of $0$. 
Because $\C$ can only return non-negative values, this is the smallest possible complexity score.

\propitemf{\propinf}{\yes} 
Let $k \in \mathbb{N}$ with $k \geq 2$. 
Figure~\ref{fig:empty-fin-inf} shows how to construct a workflow net with complexity score $k$, so $\C$ can return infinitely many values and we get $|\{c \in \mathbb{R} \mid \exists M \in \mathcal{M}: \C(M) = c\}| \geq |\{k \in \mathbb{N} \mid k \geq 2\}| = \infty$.

\propitemf{\propnotsup}{\no} Theorem~\ref{thm:empty-add} implies that $\C$ is superadditive.

\propitemf{\propadd}{\yes} Theorem~\ref{thm:empty-add} implies that $\C$ is additive.
\end{description}

\newpage
\section{Discussion of the Results}
\label{sec:discussion}
Table~\ref{table:results} shows the results of our analyses.
For properties that depend on which operation of Figure~\ref{fig:operations} we choose, we added a more detailed analysis in Table~\ref{table:operation-details}.

\renewcommand{\arraystretch}{1.1}
\begin{table}
\caption{The results of our analyses. Entries with an asterisk imply that the answer depends on the used operation.}
\label{table:results}
\centering
\begin{tabular}{|r|l|l|l|l|l|l|l||p{1cm}|p{1cm}|p{1cm}|} \hline
& \multicolumn{7}{|c||}{\textsc{Token Behavior Complexity}} & \multicolumn{3}{c|}{\textsc{Node IO Complexity}} \\ \cline{2-11}

& $\size$ & $\mismatch$ & $\connhet$ & $\crossconn$ & $\tokensplit$ & $\separability$ & $\controlflow$ & $\maxconn$ & $\sequentiality$ & $\avgconn$ \\ \hline \hline

\propone
& $\yes$ & $\yes$ & $\yes$ & $\yes$ & $\yes$ & $\yes$ & $\yes$
& $\yes$ & $\yes$ & $\yes$ \\ \hline
\proptwo
& $\yes$ & $\no$ & $\no$ & $\no$ & $\no$ & $\no$ & $\no$ 
& $\no$ & $\no$ & $\no$ \\ \hline
\propthree
& $\yes$ & $\yes$ & $\yes$ & $\yes$ & $\yes$ & $\yes$ & $\yes$ 
& $\yes$ & $\yes$ & $\yes$ \\ \hline
\propfour
& $\yes$ & $\yes$ & $\yes$ & $\yes$ & $\yes$ & $\yes$ & $\yes$ 
& $\yes$ & $\yes$ & $\yes$ \\ \hline
\propfive
& $\yes$ & $\no$ & $\no$ & $\no$ & $\yes$ & $\no$ & $\yes$ 
& $\yes$ & $\no$ & $\no$ \\ \hline
\propsix
& $\no$ & $\yes$ & $\yes$ & $\yes$ & $\no$ & $\yes$ & $\no$ 
& $\yes$ & $\yes$ & $\yes$ \\ \hline
\propseven
& $\yes$ & $\yes$ & $\yes$ & $\yes$ & $\yes$ & $\yes$ & $\yes$ 
& $\yes$ & $\yes$ & $\yes$ \\ \hline
\propeight
& $\yes$ & $\yes$ & $\yes$ & $\yes$ & $\yes$ & $\yes$ & $\yes$ 
& $\yes$ & $\yes$ & $\yes$ \\ \hline
\propnine
& $\yes$ & $\no$ & $\yes$ & $\yes$ & $\no^*$ & $\no^*$ & $\no^*$
& $\no$ & $\no^*$ & $\no$ \\ \hline \hline

\propdef
& $\yes$ & $\yes$ & $\no$ & $\yes$ & $\yes$ & $\yes$ & $\yes$ 
& $\no$ & $\yes$ & $\no$ \\ \hline
\propmin
& $\yes$ & $\yes$ & $\yes$ & $\no$ & $\yes$ & $\yes$ & $\yes$ 
& $\yes$ & $\yes$ & $\yes$ \\ \hline
\propinf
& $\yes$ & $\yes$ & $\yes$ & $\yes$ & $\yes$ & $\yes$ & $\yes$ 
& $\yes$ & $\yes$ & $\yes$ \\ \hline
\propnotsup 
& $\no$ & $\yes$ & $\yes$ & $\yes$ & $\no$ & $\yes$ & $\no$ 
& $\yes$ & $\yes$ & $\yes$ \\ \hline
\propadd
& $\no$ & $\no$ & $\no$ & $\no$ & $\no^*$ & $\no$ & $\no^*$
& $\no$ & $\no$ & $\no$ \\ \hline
\end{tabular}
\ \\
\ \\
\ \\
\begin{tabular}{|r|p{1cm}|p{1cm}|p{1cm}||p{2cm}|p{2cm}||l|l|} \hline
& \multicolumn{3}{|c||}{\textsc{Path Complexity}} & \multicolumn{2}{c||}{\textsc{Degree of Connectedness}} & \multicolumn{2}{c|}{\textsc{Other}} \\ \cline{2-8}

& $\depth$ & $\diameter$ & $\cyclicity$ & $\netconn$ & $\density$ & $\duplicate$ & $\emptyseq$ \\ \hline \hline

\propone
& $\yes$ & $\yes$ & $\yes$ 
& $\yes$ & $\yes$ 
& $\yes$ & $\yes$ \\ \hline
\proptwo
& $\no$ & $\no$ & $\no$ 
& $\no$ & $\no$ 
& $\yes$ & $\no$ \\ \hline
\propthree
& $\yes$ & $\yes$ & $\yes$ 
& $\yes$ & $\yes$ 
& $\yes$ & $\yes$ \\ \hline
\propfour
& $\yes$ & $\yes$ & $\yes$ 
& $\yes$ & $\yes$ 
& $\yes$ & $\yes$ \\ \hline
\propfive
& $\yes$ & $\no^*$ & $\no$ 
& $\no$ & $\no$ 
& $\yes$ & $\yes$ \\ \hline
\propsix
& $\no^*$  & $\no$ & $\yes$ 
& $\yes$ & $\yes$ 
& $\yes$ & $\no$ \\ \hline
\propseven
& $\yes$ & $\yes$ & $\yes$ 
& $\yes$ & $\yes$ 
& $\no$ & $\yes$ \\ \hline
\propeight
& $\yes$ & $\yes$ & $\yes$ 
& $\yes$ & $\yes$ 
& $\yes$ & $\yes$ \\ \hline
\propnine
& $\yes$ & $\yes$ & $\no^*$
& $\no$ & $\no$
& $\yes$ & $\no$ \\ \hline \hline

\propdef
& $\yes$ & $\yes$ & $\yes$ 
& $\yes$ & $\yes$ 
& $\yes$ & $\yes$ \\ \hline
\propmin
& $\yes$ & $\yes$ & $\yes$ 
& $\yes$ & $\no$ 
& $\yes$ & $\yes$ \\ \hline
\propinf
& $\yes$ & $\yes$ & $\yes$ 
& $\yes$ & $\yes$ 
& $\yes$ & $\yes$ \\ \hline
\propnotsup 
& $\yes$ & $\no^*$ & $\yes$
& $\yes$ & $\yes$ 
& $\no$ & $\no$ \\ \hline
\propadd
& $\no$ & $\no$ & $\no$
& $\no$ & $\no$ 
& $\no$ & $\yes$ \\ \hline
\end{tabular}
\end{table}
\begin{table}
\caption{A more detailed overview for measures where it depends on the operation whether a property is fulfilled.}
\label{table:operation-details}
\centering
\begin{tabular}{|r|c|c|c|c|c|c|c|} \hline
 & $\controlflow$ & $\separability$ & $\tokensplit$ & $\sequentiality$ & $\depth$ & $\diameter$ & $\cyclicity$ \\ \hline \hline
\propfive & $\yes$ & $\no$ & $\yes$ & $\no$ & $\yes$ & $\seqop: \yes$ & $\no$ \\
& & & & & & $\parop: \yes$ & \\
& & & & & & $\choiceop: \yes$ & \\
& & & & & & $\loopop: \no$ & \\ \hline
\propsix & $\no$ & $\yes$ & $\no$ & $\yes$ & $\seqop$: $\yes$ & $\no$ & $\yes$ \\ 
& & & & & $\parop$: $\no$ & & \\
& & & & & $\choiceop$: $\no$ & & \\
& & & & & $\loopop$: $\no$ & & \\ \hline
\propnine & $\seqop$: $\no$ & $\seqop$: $\no$ & $\seqop$: $\no$ & $\seqop$: $\no$ & $\yes$ & $\yes$ & $\seqop$: $\no$ \\
 & $\parop$: $\yes$ & $\parop$: $\yes$ & $\parop$: $\yes$ & $\parop$: $\yes$ & & & $\parop$: $\no$ \\ 
 & $\choiceop$: $\yes$ & $\choiceop$: $\yes$ & $\choiceop$: $\no$ & $\choiceop$: $\yes$ & & & $\choiceop$: $\no$ \\ 
 & $\loopop$: $\yes$ & $\loopop$: $\yes$ & $\loopop$: $\no$ & $\loopop$: $\yes$ & & & $\loopop$: $\yes$ \\ \hline
\propnotsup & $\no$ & $\yes$ & $\no$ & $\yes$ & $\yes$ & $\seqop$: $\no$ & $\yes$ \\
 & & & & & & $\parop$: $\yes$ & \\
 & & & & & &  $\choiceop$: $\yes$ & \\
 & & & & & & $\loopop$: $\yes$ & \\ \hline
\propadd & $\seqop$: $\yes$ & $\no$ & $\seqop$: $\yes$ & $\no$ & $\no$ & $\no$ & $\no$ \\
 & $\parop$: $\no$ &  & $\parop$: $\no$ & & & & \\ 
 & $\choiceop$: $\no$ &  & $\choiceop$: $\yes$ & & & & \\ 
 & $\loopop$: $\no$ &  & $\loopop$: $\yes$ & & & & \\ \hline
\end{tabular}
\end{table}

We split the discussion into two parts: 
We begin by analyzing the properties we used with regard to our findings.
This aids people that need to choose complexity measures to interpret our results appropriately.
Afterward, we compare the examined complexity measures based on which properties they fulfill.

\subsection{Interpreting the Properties}
Weyuker proposes that all complexity measures for software programs should fulfill her properties~\cite{Wey88}.
For complexity measures of process models, we disagree.

Looking at Table~\ref{table:results}, our first observation is that some of Weyuker's properties are fulfilled by all measures we analyzed.
In particular, this is true for the properties {\propone}, {\propthree} and {\propfour}.
This comes to no surprise, since useful complexity measures should fulfill these properties:
A measure is not reflecting the complexity of a net if it gives every net the same score or finds unique numerical representations of the net or its language.
Therefore, we argue that complexity measures for process model complexity should fulfill these three properties.

Another immediate observation is that almost all complexity measures fulfill {\propseven} and {\propeight}.
The complexity measure $\duplicate$ is the only exception, and we will see in the second part of the discussion why this is the case.
In our opinion, it makes sense for complexity measures to react to permuting the net.
Permutations can rearrange the entire structure of a process model, which should be captured by a complexity measure.
We therefore argue that {\propseven} is a property all complexity measures should fulfill.
However, we don't get to the same conclusion for {\propeight}.
Even though most measures are robust to relabeling the transitions, the labeling has an undeniable impact on the understandability of a net.
Therefore, {\propeight} has a descriptive nature and separates measures that consider the labeling from measures that don't.

Property {\proptwo} is only fulfilled by two complexity measures.
As already argued in Section~\ref{sec:properties}, we think that {\proptwo} does not fully convey the intended meaning.
Limiting the number of workflow nets that receive the same complexity score is too restrictive.
Instead, one should make sure that there aren't only few scores awarded to workflow nets. 
Our newly introduced property {\propinf} examines this and should be fulfilled by every useful complexity measure.

Weyuker's property {\propsix} is designed to test whether complexity measures are sensitive to dependencies between different parts of the net, but needs to be taken carefully in the context of process models.
For example, the measure $\mismatch$ can ``repair'' existing connector mismatches by introducing even more of them to the net.
This is no expected behavior for a complexity measure, but leads to {\propsix} being fulfilled for $\mismatch$.
Since it is possible to fulfill {\propsix} in this way, we propose to take this property as optional for a complexity measure.
We advise taking care if a measure fulfills this property and check whether it does so for the intended reasons.

The properties {\propfive} and {\propnine} are optional for complexity measures.
Property {\propfive} is very intuitive, but a composed model can always lower the complexity by introducing more structure to a net.
Not fulfilling property {\propnine}, on the other hand, might even be advantageous for certain use-cases:
If a measure doesn't fulfill {\propnine}, we cannot increase the complexity of two nets beyond the sum of their complexity scores when combining them with one of the operations shown in Figure~\ref{fig:operations}. 
Our new properties {\propnotsup} and {\propadd} deepen these insights, but are also intended as optional properties.

Finally, it is dependent on the use-case whether our new properties {\propdef} and {\propmin} are necessary for a complexity measure.
Especially since we can define special values for undefined cases, complexity measures that don't fulfill {\propdef} can be as useful as those that do.
However, in these cases, we use an altered definition for measuring complexity, and one could argue that this is a different measure than the original one.
Regarding this argument, we propose that {\propdef} should be fulfilled by complexity measures.
The property {\propmin} is important for applications that iteratively make existing process models less complex, to avoid nondeterminism.
For other use-cases, it is sufficient to be aware of whether the complexity measure has a minimum value or not.
Therefore, it is optional for a complexity measure to fulfill {\propmin}.
An overview over the results of this part of the discussion can be found in Table~\ref{table:prop-types}.

\def\normative{\textcolor{violet}{N}}
\def\descriptive{\textcolor{cyan!50!blue}{D}}
\begin{table*}
\caption{A table classifying which properties of Section~\ref{sec:properties} should be fulfilled by a complexity measure (\normative: normative) or are optional for a complexity measure (\descriptive: descriptive).}
\label{table:prop-types}
\centering
\begin{tabular}{|r|c|c|c|c|c|c|c|c|c|} \hline
& \multicolumn{9}{|c|}{\textsc{Weyuker's Properties}} \\ \cline{2-10}

\textbf{Property} & \propone & \proptwo & \propthree & \propfour & \propfive & \propsix & \propseven & \propeight & \propnine \\ \hline

\textbf{Classification} & \normative & \descriptive & \normative & \normative & \descriptive & \descriptive & \normative & \descriptive & \descriptive \\ \hline
\end{tabular}
\ \\
\ \\
\ \\
\begin{tabular}{|r|c|c|c|c|c|} \hline
& \multicolumn{5}{|c|}{\textsc{Extensions}} \\ \cline{2-6}

\textbf{Property} & \propdef & \propmin & \propinf & \propnotsup & \propadd \\ \hline

\textbf{Classification} & \normative & \descriptive & \normative & \descriptive & \descriptive \\ \hline
\end{tabular}
\end{table*}

\subsection{Comparison of Complexity Measures}
It comes to no surprise that the complexity measure based on the size of a workflow net fulfills almost all of Weyuker's properties. 
Its simplicity and strong connection to understandability~\cite{ReiM11} further emphasize that the size of a workflow net is an important factor for complexity.
However, $\size$ is superadditive, so combining two nets with one of the operations of Figure~\ref{fig:operations} is guaranteed to return a net with higher complexity than the sum of the input nets.
The complexity measures based on connector mismatches and connector heterogeneity are not monotone, since adding certain structures to a workflow net can lower the scores of these complexity measures.
If users favor monotonicity, they should avoid these measures and use the control flow complexity measure instead, which also takes the different connectors in a net into account.
Furthermore, users of the connector heterogeneity measure should be aware that $\connhet$ is not defined for workflow nets that don't contain any connectors.
The token split measure is additive for all operations except $\parop$.
Using the operator $\parop$ to combine $n$ nets, the complexity score of the resulting net is exactly by $n - 1$ higher than the sum of the complexities, which can often be tolerated in practice.
This measure has the most of what we would call desirable features in the dimension \textsc{Token Behavior Complexity}.
Separability is not monotone, since all operations of Figure~\ref{fig:operations}, except $\choiceop$, introduce new cut-vertices.
The measure is also not monotone with respect to $\choiceop$ if we take {\propfive} strictly:
If we use a workflow net with complexity $1$ as the first argument of $\choiceop$, the resulting net will always fail by definition of {\propfive} if it contains at least one cut-vertex.
If we weaken {\propfive} to avoid this special case, $\separability$ would be monotone for $\choiceop$.
Like the control flow complexity measure, $\separability$ is subadditive only for $\seqop$.
$\controlflow$, however, is also additive for this operation.
The measure $\crossconn$ fulfills {\propsix} for the operator $\seqop$ only because the weights of xor-connectors depend on whether they are the input or output place of the net.
Furthermore, $\crossconn$ has no minimum value.

Regarding \textsc{Node IO Complexity}, only the complexity measure based on the maximum connector degree is monotone. 
This measure is also subadditive for workflow nets that have at least one connector.
For workflow nets without connectors, however, this measure is undefined and setting the score to $0$ in these cases would destroy subadditivity.
The same is true for $\avgconn$.
Sequentiality $\sequentiality$ is a good choice in this dimension, if {\propnine} needs to be fulfilled.

With two exceptions, $\diameter$ and $\depth$ of the \textsc{Path Complexity} dimension have similar properties:
Regarding the $\seqop$ operator, $\diameter$ is superadditive and $\depth$ is sensitive to compositions.
More importantly, $\diameter$ is not monotone for compositions with the $\loopop$ operator, so if this operator is to be used, one should favor $\depth$ over $\diameter$.

The dimension labeled \textsc{Degree of Connectedness} contains no monotone measures.
This could hint to a gap for monotone complexity measures that take the connectedness of the net into account.
Furthermore, all measures are subadditive, so if this property is undesired, this dimension offers no alternatives.
The measures $\netconn$ and $\density$ are quite similar in their definition and therefore in the properties they fulfill, but $\density$ offers no minimum complexity score, while $\netconn$ does.

The complexity measure based on duplicate tasks is the only one that doesn't fulfill {\propseven} and {\propeight}, since it solely depends on the labeling of the net.
Its superadditivity, however, depends on the number of task labels in the alphabet that aren't already used in the models.
The number of empty sequence flows is additive for all operations of Figure~\ref{fig:operations} and therefore not sensitive to compositions.
It is a simple complexity measure, but shows immediate ways to make a model simpler.

\section{Conclusion}
\label{sec:conclusion}
In this paper, we used and extended the properties defined by Weyuker~\cite{Wey88} to analyze and compare popular complexity measures for process models~\cite{Men08}.
We discussed which of the inspected properties are normative and which are descriptive for process models to give a sense of their importance.
Furthermore, we compared complexity measures that are in the same complexity dimension discovered by Lieben et al.~\cite{LieDJJ18} and highlighted when to prefer which measure.
Defining and analyzing the structuredness measure was out of scope of this paper, so possible future work involves the analysis of this measure.
Moreover, one can think of more properties that highlight interesting characteristics of complexity measures.
More composition rules than those used by the ETM would be interesting to analyze. Furthermore, investigating how to define permutations of Petri nets is interesting for a more expressive {\propseven} property.
The analysis could also be extended to other workflow net complexity measures, such as in~\cite{LasA09}.

We are confident that our analyses and discussions shed a new light on popular complexity measures that will help analysts and algorithm designers to choose measures fitting to their  needs. 
We also hope to start a discussion on complexity and simplicity measures, which are often overlooked during the evaluation of process models.

\FloatBarrier
\

\end{document}